\documentclass[a4paper,11pt]{article}

\usepackage{vmargin}
\setmarginsrb{1in}{1in}{1in}{1in}{0pt}{0pt}{0pt}{6mm}

\usepackage{mymacros}
\usepackage{mydebug}

\newcommand{\cdel}{\textsc{CVD}\xspace}
\newcommand{\wcdel}{\textsc{WCVD}\xspace}

\newcommand{\OO}{{\mathcal O}}

\newcommand{\opt}{{\sf opt}}
\newcommand{\fopt}{{\sf fopt}}
\newcommand{\numdhd}{48}

\newcommand{\defparproblemOpt}[3]{
  \vspace{3mm}
  \noindent\fbox{
  \begin{minipage}{.95\textwidth}
  \begin{tabular*}{\textwidth}{@{\extracolsep{\fill}}lr} \textsc{#1}  \\ \end{tabular*}
  {\bf{Input:}} #2  \\
  {\bf{Question:}} #3
  \end{minipage}
  }
  \vspace{2mm}
}

\usepackage{algorithm}
\usepackage{algorithmic}
\newcommand{\alg}[1]{\mbox{\sf #1}}  %% Names of algorithms in small caps

\newcommand{\dho}{DH-obstruction}
\newcommand{\approxDH}{d \cdot \log^2 n}

\newcommand{\biclique}{\emph{biclique}}
\newcommand{\x}{{\bf x}}

\newcommand{\DHbddObs}{50}

\newcommand{\pfd}{{\sc Planar  $\mathscr{F}$-Minor-Free Deletion}\xspace}
\newcommand{\wpfd}{{\sc Weighted Planar  $\mathscr{F}$-Minor-Free Deletion}\xspace}
\newcommand{\wfd}{{\sc Weighted $\mathscr{F}$-Minor-Free Deletion}\xspace}
\newcommand{\wpfdshort}{{\sc WP$\mathscr{F}$-MFD}\xspace}
\newcommand{\wcvd}{\textsc{Weighted Chordal Vertex Deletion}\xspace}

\newcommand{\wretad}{\textsc{Weighted Rankwidth-$\eta$ Vertex Deletion}\xspace}
\newcommand{\wretadshort}{\textsc{WR-$\eta$VD}\xspace}

\newcommand{\wroned}{\textsc{Weighted Rankwidth-$1$ Vertex Deletion}\xspace}
\newcommand{\wronedshort}{\textsc{WR-$1$VD}\xspace}
\newcommand{\wDHfull}{\textsc{Weighted Distance Hereditary Vertex Deletion}\xspace}
\newcommand{\wDH}{\textsc{WDHVD}\xspace}

\newcommand{\DHfull}{\textsc{Distance Hereditary Vertex Deletion}\xspace}
\newcommand{\DHsmall}{\textsc{DHVD}\xspace}

\newcommand{\genfvd}{{\sc Weighted $\mathcal{F}$ Vertex Deletion}\xspace}

\newcommand{\afg}{$\OO(\log^{\OO(1)} n)$}
\newcommand{\afcvd}{$\OO(\log^{2} n)$}
\newcommand{\afpfd}{$\OO(\log^{2} n)$}
\newcommand{\afpfdr}{$\OO(\log^{1.5} n)$}
\newcommand{\afdh}{$\OO(\log^{3} n)$}

\title{Polylogarithmic  Approximation Algorithms for  Weighted-$\mathcal{F}$-Deletion Problems\thanks{The research leading to these results received funding from the European Research Council under the European Union’s Seventh Framework Programme (FP/2007-2013) / ERC Grant Agreement no.~306992.}}
\author{
Akanksha Agrawal\thanks{
University of Bergen, Bergen, Norway. \texttt{akanksha.agrawal@uib.no}.
} \and 
Daniel Lokshtanov\thanks{
University of Bergen, Bergen, Norway. \texttt{daniello@ii.uib.no}.
} \and 
Pranabendu Misra\thanks{
Institute of Mathematical Sciences, Chennai, India. \texttt{pranabendu@imsc.res.in}.
}
\and 
Saket Saurabh\thanks{
University of Bergen, Bergen, Norway. The Institute of Mathematical Sciences HBNI, Chennai, India. \texttt{saket@imsc.res.in}.
}
\and 
Meirav Zehavi\thanks{
University of Bergen, Bergen, Norway. \texttt{meirav.zehavi@uib.no}.
}}

\date{}
 
\begin{document}
\begin{titlepage}
\def\thepage{}
\thispagestyle{empty}
\maketitle
\begin{abstract} 
%!TEX root = mainChordal.tex

Let $\mathcal F$ be a family of graphs. A canonical vertex deletion problem corresponding to $\mathcal F$ is defined as follows: given an $n$-vertex undirected graph $G$ and a weight function $w: V(G)\rightarrow\mathbb{R}$, find a minimum weight subset $S\subseteq V(G)$ such that $G- S$ belongs to $\cal F$. This is known as \genfvd problem. In this paper we devise a recursive scheme to obtain \afg-approximation algorithms for such problems, building upon the classic technique of finding \emph{balanced separators} in a graph. Roughly speaking, our scheme applies to those problems, where an optimum solution $S$ together with a well-structured set $X$, form a balanced separator of the input graph. In this paper, we obtain the first \afg-approximation algorithms for the following vertex deletion problems. 
\begin{itemize}
\setlength{\itemsep}{-2pt}

\item We give an \afcvd-factor approximation algorithm for \wcvd (\wcdel), the vertex deletion problem to the family of chordal graphs. On the way to this algorithm, we also obtain a constant factor approximation algorithm for {\sc Multicut} on chordal graphs. 

\item We give an \afdh-factor approximation algorithm for \wDHfull (\wDH), also known as \wroned (\wronedshort). This is the vertex deletion problem to the family of distance hereditary graphs, or equivalently, the family of graphs of rankwidth 1. 
\end{itemize}

Our methods also allow us to obtain in a clean fashion a $\OO(\log^{1.5} n)$-approximation algorithm for the \genfvd problem when $\cal F$ is a minor closed family excluding at least one planar graph. For the {\em unweighted} version of the problem constant factor approximation algorithms are were known~[Fomin et al., FOCS~2012], while for the weighted version considered here an $\OO(\log n  \log\log n)$-approximation algorithm follows from~[Bansal et al. SODA~2017]. We believe that our recursive scheme can be applied to obtain \afg-approximation algorithms for many other problems as well.

%Let ${\mathscr F}$ be  a finite set of graphs containing a planar graph, and ${\cal F}={\mathscr G}({\mathscr F})$ be the family of graphs where every graph $H\in{\mathscr G}({\mathscr F})$ excludes all graphs in ${\mathscr F}$ as minors. The vertex deletion problem corresponding to   ${\cal F}={\mathscr G}({\mathscr F})$ is the  \wpfd (\wpfdshort) problem. 

%We give randomized and deterministic approximation algorithms for \wpfdshort with ratios \afpfdr{} and \afpfd{}, respectively. Previously, only a randomized constant factor approximation algorithm for  the {\em unweighted} version of the problem was known~[Fomin et al., FOCS~2012]. 

\end{abstract}
\end{titlepage}
\newpage 

%!TEX root = mainChordal.tex

\section{Introduction}
Let $\cal F$ be a family of undirected graphs. 
Then a natural optimization problem 
%associated with $\cal F$ 
is as follows. 

\defparproblemOpt{\genfvd}{An undirected graph $G$ and a weight function $w: V(G)\rightarrow\mathbb{R}$.}
{Find a minimum weight subset $S\subseteq V(G)$ such that $G- S$ belongs to $\cal F$.}

%The \genfvd problem defines a wide subclass of node (or vertex) removal problems studied from the 1970s. 
\noindent
The \genfvd problem captures a wide class of node (or vertex) deletion problems that have been studied from the 1970s.
For example, when $\mathcal F$ is the family of independent sets, forests, bipartite graphs, planar graphs, and chordal graphs, then the corresponding vertex deletion problem corresponds to {\sc Weighted Vertex Cover}, {\sc Weighted Feedback Vertex Set}, 
{\sc Weighted Vertex Bipartization} (also called {\sc Weighted Odd Cycle Transversal}), 
{\sc Weighted Planar Vertex Deletion} and {\sc Weighted Chordal Vertex Deletion}, respectively.
By a classic theorem of Lewis and Yannakakis~\cite{LewisY80}, the decision version of the \genfvd problem---deciding whether there exists a set $S$ weight at most $k$, such that removing $S$ from $G$ results in a graph with property $\Pi$---is {\sf NP}-complete for every non-trivial hereditary property\footnote{A graph property $\Pi$ is simply a family of graphs, and it is called {\em non-trivial} if there exists an infinite number of graphs that are in $\Pi$, as well as an infinite number of graphs that are not in $\Pi$. A non-trivial graph property $\Pi$ is called \emph{hereditary} if $G\in \Pi$ implies that every induced subgraph of $G$ is also in $\Pi$.} $\Pi$. 
%We may similarly define the corresponding problem for edge deletion.
%We will use the name \genfvduw to denote the unweighted version of \genfvd. 

%It is well known that 
%{\sc Weighted Vertex Cover} is polynomial time solvable when the input is restricted to bipartite graphs. Thus the general theorem of  Lewis and Yannakakis  \cite{LewisY80} showing that   the decision version of the \genfvd problem is {\sf NP}-complete for every non-trivial hereditary property $\Pi$ does not apply when we restrict the input to bipartite graphs. Yannakakis~\cite{} characterized those non-trivial properties for which decision version of the bipartite restriction of the  \genfvd problem is polynomial and those for which it remains {\sf NP}-complete. 
Characterizing the graph properties, for which the corresponding vertex deletion problems can be approximated within a bounded factor in polynomial time, is a long standing open problem in approximation algorithms \cite{Yannakakis94}. 
In spite of a long history of research, we are still far from a complete characterization.
%It has led to a long and fruitful line of research,\todo{cite some results}, which we briefly summarize as follows.
Constant factor approximation algorithms  for  {\sc Weighted Vertex Cover}  are known since 1970s \cite{Bar-YehudaE81,NemT74}. 
Lund and Yannakakis observed that the vertex deletion problem for any hereditary property with a ``finite number of minimal forbidden induced subgraphs'' can be approximated within a constant ratio \cite{LundY93}. 
%There are other ways of defining a \todo{DL:remove this sentence?} family of forbidden graphs, we will define one such way later. 
%For example, when $\cal F$ is the family of independent sets, forests, bipartite graphs, 
%planar graphs, 
%and chordal graphs respectively,
%the corresponding forbidden set ${\sf Forbid}({\cal F})$ 
%\footnote{Given a graph family $\cal F$, by ${\sf Forbid}({\cal F})$ we denote the family of graphs such that  $G\in \cal F$ if and only if $G$ does not contain any graph in  ${\sf Forbid}({\cal F})$ as an induced subgraph.}
%consists of an edge, all cycles, all odd cycles, 
%all non-planar graphs\todo{remove planar graphs because this is not the forbidden set?}, 
%and all cycles of length at least $4$, respectively.
%
%Observe that some, of these forbidden sets ${\sf Forbid}({\cal F})$ are finite, such as for the family of independent sets, while the others are infinite. 
%Lund and Yannakakis~\cite{LundY93} 
They conjectured that for every nontrivial, hereditary property $\Pi$ with an infinite forbidden set, the corresponding vertex deletion problem cannot be approximated within a constant ratio. 
%However, it later shown that {\sc Weighted Feedback Vertex Set} admits a constant factor approximation~\cite{BafnaBF99,BarYGJ98}. 
However, it was later shown that {\sc Weighted Feedback Vertex Set}, which doesn't have a finite forbidden set, admits a constant factor approximation~\cite{BafnaBF99,BarYGJ98}, thus disproving their conjecture.
%Thus it is not solely the finiteness of the set of forbidden subgraphs that governs whether the vertex deletion problem of a family is constant factor approximable or not.
%Thus the boundary between the families for which the corresponding vertex deletion problem is approximable within a constant factor, and the families for which it is not, is not demarcated by whether the family has a finite set of forbidden 
%
%dividing line of   approximability lies somewhere else. 
On the other hand a result by Yannakakis  \cite{Yannakakis79} shows that, for a wide range of graph properties $\Pi$, approximating the minimum number of vertices to delete in order to obtain a \emph{connected} graph with the property $\Pi$ within a factor $n^{1-\varepsilon}$ is NP-hard.
We refer to~\cite{Yannakakis79} for the precise list of graph properties to which this result applies to, but it is worth mentioning the list includes the class of acyclic graphs and the class of outerplanar graphs. 
%\todo{add some more results}

In this paper, we explore the approximability of \genfvd for several different families $\mathcal F$ and design \afg-factor approximation algorithms for these problems. More precisely, 
%in this paper, we give new \afg-approximation algorithms for three vertex deletion problems for which no   \afg-approximation algorithms were previously known. Our 
our results are as follows.  
\begin{enumerate}
\setlength{\itemsep}{-2pt}
\item  Let ${\mathscr F}$ be  a finite set of graphs that includes a planar graph. Let ${\cal F}={\mathscr G}({\mathscr F})$ be the family of graphs such that every graph $H\in{\mathscr G}({\mathscr F})$ does not contain a graph from ${\mathscr F}$ as a minor. The vertex deletion problem corresponding to  ${\cal F}={\mathscr G}({\mathscr F})$ is known as the \wpfd (\wpfdshort). The \wpfdshort problem is a very generic problem and by selecting different sets of forbidden minors ${\mathscr F}$, one can obtain  various fundamental  problems such as {\sc Weighted Vertex Cover}, {\sc Weighted Feedback Vertex Set} or {\sc Weighted Treewidth $\eta$-Deletion}. Our first result is a randomized \afpfdr-factor (deterministic \afpfd-factor)  approximation algorithm for \wpfdshort, for any finite $\mathscr F$ that contains a planar graph. 

We remark that a different approximation algorithm for the same class of problems with a slightly better approximation ratio of $\OO(\log n \log\log n)$ follows from recent work of  Bansal, Reichman, and Umboh~\cite{BansalRU1} (see also the discussion following Theorem~\ref{thm:approx_thmpfd}). Therefore, our first result should be interpreted as a clean and gentle introduction to our methods.

%This problem encompasses several well-known problems such as 

%Previously,  
%a randomized constant factor approximation algorithm for  the {\em unweighted} version of the problem was known~[FOCS 2012]. 
 
\item We give an \afcvd-factor approximation algorithm for \wcvd (\wcdel), the vertex deletion problem corresponding to the family of chordal graphs. 
On the way to this algorithm, we also obtain a constant factor approximation algorithm for {\sc Weighted Multicut} in chordal graphs.

\item We give an \afdh-factor approximation algorithm for \wDHfull (\wDH). This is also known as the \wroned (\wronedshort) problem. 
This is the vertex deletion problem corresponding to the family of distance hereditary  graphs, or equivalently graphs of rankwidth 1. 
\end{enumerate}
All our algorithms follow the same recursive scheme, that find ``well structured balanced separators'' in the graph by exploiting the properties of the family $\cal F$.
In the following, we first describe the methodology by which we design all these approximation algorithms. Then, we give a brief overview, consisting of  known results and our contributions, for each problem we study.

%Loosely speaking our scheme applies to problems where an optimum solution $S$ together with a simple but possibly large structure $X$ forms a balanced separator of the input graph. We believe that this scheme will apply to way more problems than listed here. To obtain our approximation algorithm for \wcdel we need to solve instances of {\sc Multicut} on chordal graphs. For {\sc Multicut} on chordal graphs we give a constant factor approximation algorithm; generalizing the factor $2$ approximation on trees. 

%\vspace{-0.3cm}

\paragraph{Our Methods.} Multicommodity max-flow min-cut theorems are a classical technique in designing approximation algorithms, which was pioneered by Leighton and Rao in their seminal paper~\cite{BalancedSeparator}.  
This approach can be viewed as using balanced vertex (or edge) separators\footnote{A \emph{balanced vertex separator} is a set of vertices $W$, such that every connected component of $G-W$ contains at most half of the vertices of $G$.} in a graph to obtain a divide-and-conquer approximation algorithm.
In a typical application, the optimum solution $S$, forms a balanced separator of the graph.
Thus, the idea is to find an minimum cost balanced separator $W$ of the graph and add it to the solution, and then recursively solve the problem on each of the connected components. This leads to an \afg-factor approximation algorithm for the problem in question. 

Our recursive scheme is a strengthening of this approach which exploits the structural properties of the family $\cal F$.
Here the optimum solution $S^*$ need not be a balanced separator of the graph. Indeed, a balanced separator of the graph could be much larger than $S^*$.
Rather, $S^*$ along with a possibly large but well-structured subset of vertices $X$, forms a balanced separator of the graph.
We then exploit the presence of such a balanced separator in the graph
to compute an approximate solution.
%Then, we exploit the structure of the set $X$, as well as the existance of balanced separator extending $X$ to compute an approximate solution.
%All our algorithms are based on this strengthening of the scheme of Leighton and Rao~\cite{BalancedSeparator}.
% same recursive scheme.
%Next we abstract steps of our recursive scheme. 
%To explain the steps of the scheme, we will show how it applies to  \genfvd. 
%The following is a sketch of our recrsive scheme.
Consider a family $\cal F$ for which \genfvd is amenable to our approach, and let $G$ be an instance of this problem.
Let $S$ be the approximate solution that we will compute.
Our approximation algorithm has the following steps: 
\begin{enumerate}
\setlength{\itemsep}{-2pt}
\item Find a well-structured set $X$, such that $G-X$ has a balanced separator $W$ which is not too costly.
% and remove it from $G$.
\item Next, compute the balanced separator $W$ of $G-X$ using the known factor $\OO(\sqrt{\log n})$-approximation algorithm (or deterministic $\OO(\log n)$-approximation algorithm) for  {\sc Weighted Vertex Separators}~\cite{FeigeHL08,BalancedSeparator}. Then add $W$ into the solution set $S$ and recursively solve the problem on each connected component of $G- (X\cup S)$. Let $S_1,\cdots, S_\ell$ be the solutions returned by the recursive calls. We add $S_1,\cdots, S_\ell$ to the solution $S$. 

\item Finally, we add $X$ back into the graph and consider the instance $(G-S) \cup X$.
Observe that, $V(G - S)$ can be partitioned into $V'\uplus X$, where $G[V']$ belongs to $\cal F$ and $X$ is a well-structured set.
We call such instances, the \emph{special case} of \genfvd.
We apply an approximation algorithm that exploits the structural properties of the special case to compute a solution.
\end{enumerate} 
\noindent
%As we shall see later, for \wpfdshort, it turns out the structure $X$ is a certain {\em constant sized set of vertices} where the constant depends on the family $\mathscr F$.
%While for \wcdel and \wDH, it is a {\em clique} and a {\em biclique},
%respectively.
%
Now consider the problem of finding the structure $X$.
%Thus, in some sense our scheme applies to problems where an optimum solution $S$ together with a simple but possibly large structure $X$ forms a balanced separator of the input graph. 
%
One way is to enumerate all the candidates for $X$ and then pick the one where $G-X$ has a balanced vertex separator of least cost --- this separator plays the role of $W$. 
However, the number of candidates for $X$ in a graph could be too many to enumerate in polynomial time.  
For example, in the case of \wcvd, the set $X$ will be a clique in the graph, and the number of maximal cliques in a graph on $n$ vertices could be as many as $3^{\frac{n}{3}}$~\cite{moon1965cliques}.
Hence, we cannot enumerate and test every candidate structure in polynomial time.
However, we can exploit certain structural properties of family $\calF$, to reduce the number of candidates for $X$ in the graph.
In our problems, we ``tidy up'' the graph by removing ``short obstructions'' that forbid the graph from belonging to the family $\cal F$. 
Then one can obtain an upper bound on the number of candidate structures. 
In the above example, recall that a graph $G$ is chordal if and  only if there are no induced cycles of length $4$ or more. 
%Thus, to modify an input graph $G$ to be chordal we need to delete a set $S$ from $G$ such that $S$ intersects all the induced cycles of length at least $4$. 
% To obtain the desired approximation algorithm one can write a canonical linear programming formulation (LP) for \wcvd\ where the objective function is to 
% \todo[inline]{ write the lp formulation here.}
% One solves the LP in polynomial time and round
It is known that a graph $G$ without any induced cycle of length $4$ has at most $\OO(n^2)$ maximal cliques~\cite{C4FreeNumCliques}. 
Observe that, we can greedily compute a set of vertices which intersects all induced cycles of length $4$ in the graph.
Therefore, at the cost of factor $4$ in the approximation ratio,
we can ensure that the graph has only polynomially many maximal cliques.
%does not contain any induced cycle of length $4$.
Hence, one can enumerate all maximal cliques in the remaining graph ~\cite{TsukiyamaIAS77} to test for $X$.
%Similarly, for \wpfdshort, it turns out the structure $X$ is a certain {\em constant sized set of vertices} where the constant depends on the family $\mathscr F$, and for \wDH, the structure is a {\em biclique}.
%
%and use it for our approximation algorithm.   
 % For several family of graphs $\mathscr F$, there exists a 
% family of graphs known as ${\sf Forbid}({\mathscr F})$ such that $H\in \mathscr F$ if and only if $H$ does not contain a graph in  ${\sf Forbid}({\mathscr F})$ as an induced subgraph or subgraph or a minor. the family ${\sf Forbid}({\mathscr F})$ need not be finite. 
%There is nothing special about intersecting all induced cycles of length $4$ in $G$ at the cost of factor $4$ in the approximation ratio. 
%
%In fact, this is true about any  vertex deletion problem for which we can write a linear programming  (LP) 
%formulation of the following form. We have a constraint for each structure in $G$ that must intersect the solution 
%set $S$ (i.e. $G- S$ belongs to $\mathscr F$), and we can solve this LP in polynomial time.  
% One first solves the LP and then selects every vertex 
%that has been assigned a value at least $\frac{1}{q}$, where $q$ is an integer,  to the solution. Thus, one can 
%assume that $G$ does not have any obstruction of size at most $q$, at the cost of factor $q$ in the approximation ratio. We use this idea several times to tidy our graphs and to show that there are polynomially many simple structures that can play the role of $X$. 

Next consider the task of solving an instance of the special case of
the problem. 
We again apply a recursive scheme, but now with the advantage of a much more structured graph.
By a careful modification of an LP solution to the instance, we eventually reduce it to instances of {\sc Weighted Multicut}.
In the above example, for \wcvd we obtain instances of {\sc Weighted Multicut} on a chordal graph. 
We follow this approach for all three problems that we study in this paper.
%We discuss each of them in more detail in the rest of this section.
We believe our recursive scheme can be applied to obtain \afg-approximation algorithms for {\sc Weighted $\mathcal{F}$ Vertex (Edge) Deletion} corresponding to several other graph families $\cal F$. %by exploiting the structural properties of these graphs.
%Let us now discuss the problems we study in more detail.

\vspace{-.3cm}
\paragraph{Weighted Planar $\mathscr F$-Minor-Free Deletion.}
%{\bf \wpfd}.} 
Let ${\mathscr F}$ be  a finite set of graphs containing a planar graph. Formally, \wpfd is defined as follows. 

\defparproblemOpt{\wpfd (\wpfdshort)}{An undirected graph $G$ and a weight function $w: V(G)\rightarrow\mathbb{R}$.}
{Find a minimum weight subset $S\subseteq V(G)$ such that $G- S$ does not contain any graph in 
$\mathscr F$ as a minor.}

\noindent 
The \wpfdshort problem is a very generic problem that encompasses several known problems. To explain the versatility of the problem, we require a few definitions. A graph $H$ is called a {\em minor} of a graph $G$ if we can obtain $H$ from $G$ by a sequence of vertex deletions, edge deletions and edge contractions, and a family of graphs 
$\cal F$ is called {\em minor closed} if $G\in \cal F$ implies that every minor of $G$ is also in $\cal F$. 
Given a graph family $\cal F$, by ${\sf ForbidMinor}({\cal F})$ we denote the family of graphs such that 
$G\in \cal F$ if and only if $G$ does not contain any graph in  ${\sf ForbidMinor}({\cal F})$ as a minor. 
By the celebrated Graph Minor Theorem of Robertson and Seymour, every minor closed family $\cal F$ is characterized by a finite family of forbidden minors~\cite{RobertsonS04}. That is,  ${\sf ForbidMinor}({\cal F})$ has finite size. Indeed, the size of ${\sf ForbidMinor}({\cal F})$ depends on the family $\cal F$. 
Now for a finite collection of graphs $\mathscr F$, as above, we may define the \wfd problem.
And observe that, even though the definition of \wfd we only consider finite sized ${\mathscr F}$, this problem actually encompasses deletion to every minor closed family of graphs. 
Let $\mathscr{G}$ be the set of all finite undirected graphs, and let $\mathscr{L}$ be the family of all finite subsets of  $\mathscr{G}$. Thus,   every element ${\mathscr F} \in \mathscr{L}$ is a finite set of graphs, and throughout the paper we assume that  ${\mathscr F} $ is explicitly given.  In this paper, we show that when  ${\mathscr F}\in \mathscr{L}$ contains at least one planar graph, then it is possible to obtain an \afg-factor approximation algorithm for \wpfdshort. 

The case where $\mathscr F$ contains a planar graph, while being considerably more restricted than the general case, already encompasses a number of the well-studied  
 instances of \wpfdshort. For  example, when ${\mathscr F}=\{K_2\}$, a complete graph on two
vertices, this is the {\sc Weighted Vertex Cover} problem. When ${\cal F}=\{C_3\}$, a cycle on three
vertices, this is the {\sc Weighted Feedback Vertex Set} problem. Another fundamental problem, which is also a special case of \wpfdshort, is {\sc Weighted Treewidth-$\eta$ Vertex Deletion} or  {\sc Weighted $\eta$-Transversal}. Here the task is to 
delete a minimum weight vertex subset to obtain a graph of treewidth at most $\eta$. Since any graph of treewidth $\eta$ excludes a 
$(\eta+1)\times (\eta+1)$ grid as a minor, we have that the set $\mathscr F$ of forbidden minors of treewidth 
$\eta$ graphs contains a planar graph. 
%minor of the$(\eta+1)\times (\eta+1)$ times grid and this minor is planar
%Indeed, the class of graphs of treewidth at most $\eta$ can be characterized by a
%finite set of forbidden minors, and because no graph of treewidth $\eta$ has an $(\eta+1)
%\times (\eta+1)$ grid as a minor, we can always add this grid, a planar graph, to $\cal F$. 
{\sc Treewidth-$\eta$ Vertex Deletion}  plays an important role in 
 generic efficient polynomial time approximation schemes based on Bidimensionality theory~\cite{FominLRS11,FominLS12}.  Among other examples of    \pfd problems that can be found in the literature on approximation and parameterized algorithms, are 
%Examples of  other studied  special variants of  \fd{}   are
  the cases of ${\mathscr F}$ being 
$\{K_{2,3}, K_4\}$, $ \{K_4\}$, $\{\theta_c\}$,  and $ \{K_{3}, T_2\}$, which correspond to removing vertices to
obtain an  outerplanar graph, a series-parallel graph,  a diamond graph,   and a graph  of pathwidth 1,  respectively. 
 
%Given a graph family $\cal F$, by ${\sf Forbid}({\cal F})$ we denote the family of graphs such that 
%$G\in \cal F$ if and only if $G$ does not contain any graph in  ${\sf Forbid}({\cal F})$ as an induced subgraph. There are other ways of defining a family of forbidden graphs, we will define one such way later. Observe that when 
%$\cal F$ the family of independent sets, forests, bipartite graphs, planar graphs, and chordal graphs 
%then the corresponding ${\sf Forbid}({\cal F})$ consists of an edge, all cycles, all odd cycles, 
%all non-planar graphs, and all cycles of length at least $4$. Clearly, ${\sf Forbid}({\cal F})$ is finite for the family of independent sets and infinite for other families.  Lund and Yannakakis\cite{LundY93}  conjectured that for every nontrivial, hereditary property with an infinite number of minimal forbidden subgraphs, the  vertex deletion problem cannot be approximated with constant ratio.  However, it appeared later that  {\sc Weighted Feedback Vertex Set} admits  a constant factor approximation~\cite{BarYGJ98,BafnaBF99}

Apart from the case of {\sc Weighted Vertex Cover}~\cite{Bar-YehudaE81,NemT74} and {\sc Weighted Feedback Vertex Set}~\cite{BafnaBF99,BarYGJ98}, there was not much progress on approximability/non-approximability of \wpfdshort until the work 
of Fiorini, Joret, and Pietropaoli~\cite{Fiorini:2009ipco}, which gave a constant factor approximation algorithm for 
%\wpfdshort for 
the case of \wpfdshort where ${\mathscr F}$ is a diamond graph, i.e., a graph with two vertices and three parallel edges. In 2011, Fomin et al.~\cite{FominLMPS16}  considered \pfd (i.e. the unweighted version of \wpfdshort) in full generality and designed a randomized (deterministic) $\OO(\log^{1.5} n)$-factor ($\OO(\log^{2} n)$-factor) approximation algorithm for it. Later,  Fomin et al.~\cite{FominLMS12} gave a randomized constant factor approximation algorithm for \pfd.  
Our algorithm for \wpfdshort extends this result to the weighted setting, at the cost of increasing the approximation factor to $\log^{\OO(1)} n$.
%; albeit  at the cost of larger factor of approximation. 
%In particular, our result states the following. 

\begin{theorem}\label{thm:approx_thmpfd}
For every set  ${\mathscr F} \in \mathscr{L}$, \wpfdshort admits a randomized (deterministic) \afpfdr-factor (\afpfd-factor)  approximation algorithm. 
 %can be approximated with a constant factor ratio.  
 \end{theorem}

We mention that Theorem~\ref{thm:approx_thmpfd} is subsumed by a recent related result of Bansal, Reichman, and Umboh~\cite{BansalRU1}. They studied the edge deletion version of the {\sc Treewidth-$\eta$ Vertex Deletion} problem, under the name {\sc Bounded Treewidth Interdiction Problem}, and gave a bicriteria approximation algorithm.
In particular, for a graph $G$ and an integer $\eta>0$, they gave a polynomial time algorithm that finds a subset of edges $F'$ of $G$ such that $|F'| \leq \OO((\log n \log \log n)\cdot  {\sf opt})$ and the treewidth of $G-F'$ 
%(by deleting the edges of $F'$ from $G$) 
is $\OO(\eta \log \eta)$. With some additional effort~\cite{BansalRU2} their algorithm can be made to work for the {\sc Weighted Treewidth-$\eta$ Vertex Deletion} problem as well. 
%\footnote{Here,  $F^*$ is a minimum sized set such that $G-F^*$ has treewidth has at most $\eta$.}
In our setting where $\eta$ is a fixed constant,  this immediately implies a factor $\OO(\log n \log \log n)$ approximation algorithm for \wpfdshort.\footnote{One can run their algorithm first and remove the solution output by their algorithm to obtain a graph of treewidth at most $\OO(\eta \log \eta)$. Then one can find an optimal solution using standard dynamic programming.}
While the statement of Theorem~\ref{thm:approx_thmpfd} is subsumed by~\cite{BansalRU1}, the proof gives a simple and clean introduction to our methods. 
%We remark that their result can be also extended to the weighted setting.
%We remark that even though their result is stated for unweighted graphs, their algorithm easily extends to the edge weighted setting because we can just copy each edge a number of times proportionate to its weight.

%\begin{theorem}\label{thm:approx_thm}
%For every set  ${\cal F} \in \mathscr{F}$ containing a planar graph, 
%\ofd{}
%admits a  randomized constant  ratio approximation algorithm. 
% %can be approximated with a constant factor ratio.  
% \end{theorem}
%Let us remark that for all known constant factor approximation algorithms of vertex deletion to a  hereditary property $\pi$, property $\pi$ is either
%  characterized by an   finite number of minimal forbidden subgraphs or by finite number of forbidden minors, one of which is planar.   Theorem~\ref{thm:approx_thm} together with the result of Lund and Yannakakis, not only encompass all known  vertex deletion problems with constant factor approximation ratio but significantly extends   known tractability 
%  borders for   such types of problems.
 
 \vspace{-0.3cm} 
 \paragraph{Weighted Chordal Vertex Deletion.} Formally, the \wcvd problem is defined as follows. 

\defparproblemOpt{{\sc Weighted Chordal Vertex Deletion (WCVD)}}{An undirected graph $G$ and a weight function $w: V(G)\rightarrow\mathbb{R}$.}{Find a minimum weight subset $S\subseteq V(G)$ such that $G- S$ is a chordal graph.}

The class of chordal graphs is a natural class of graphs that has been extensively studied from the viewpoints of Graph Theory and Algorithm Design. Many important problems that are {\sf NP}-hard on general graphs, such as {\sc Independent Set}, and {\sc Graph Coloring} are solvable in polynomial time once restricted to the class of chordal graphs \cite{Golumbic80}. Recall that a graph is chordal if and only if it does not have any induced cycle of length $4$ or more. 
%Thus, the class of chordal graphs is closely related to the class of acyclic graphs. 
Thus, {\sc Chordal Vertex Deletion (CVD)} can be viewed as a natural variant of the classic {\sc Feedback Vertex Set (FVS)}. Indeed, while the objective of {\sc FVS} is to eliminate all cycles, the \cdel problem only asks us to eliminate induced cycles of length $4$ or more. Despite the apparent similarity between the objectives of these two problems, the design of approximation algorithms for \wcdel is very challenging. In particular, chordal graphs can be dense---indeed, a clique is a chordal graph. As we cannot rely on the sparsity of output, our approach must deviate from those employed by approximation algorithms from {\sc FVS}. That being said, chordal graphs still retain some properties that resemble those of trees, and these properties are utilized by our algorithm.

Prior to our work, only two non-trivial approximation algorithms for \cdel were known. The first one, by Jansen and Pilipczuk \cite{JansenPili2016}, is a deterministic $\OO(\opt^2\log\opt\log n)$-factor approximation algorithm, and the second one, by Agrawal et al.~\cite{soda17chordal}, is a deterministic $\OO(\opt\log^2 n)$-factor approximation algorithm. The second result implies that \cdel admits an $\OO(\sqrt{n}\log n)$-factor approximation algorithm.\footnote{If $\opt\geq \sqrt{n}/\log n$, we output a greedy solution to the input graph, and otherwise we have that $\opt\log^2 n\leq \sqrt{n}\log n$, hence we call the $\OO(\opt\log^2 n)$-factor approximation algorithm.}
In this paper we obtain the first \afg-approximation algorithm for \wcdel. 
%More precisely, we obtain the following theorem.
\begin{theorem}
\label{thm:newApprox2}
\cdel admits a deterministic $\OO(\log^2 n)$-factor approximation algorithm.
\end{theorem}

While this approximation algorithm follows our general scheme, it also requires us to incorporate several new ideas. 
In particular, to implement the third step of the scheme, we need to design a different 
%approximation algorithm 
%(which the scheme employs as a black box). Namely, we design an 
$\OO(\log n)$-factor approximation algorithm for the special case of \wcdel where the vertex-set of the input graph $G$ can be partitioned into two sets, $X$ and $V(G)\setminus X$, such that $G[X]$ is a clique and $G[V(G)\setminus X]$ is a chordal graph. This approximation algorithm is again based on recursion, but it is more involved.
%, but it is not based on our general scheme. 
At each recursive call, it carefully manipulates a fractional solution of a special form. Moreover, to ensure that its current problem instance is divided into two subinstances that are independent and simpler than their origin, we introduce multicut constraints. In addition to these constraints, we keep track of the complexity of the subinstances,
which is measured via the cardinality of the maximum independent set in the graph.
Our multicut constraints result in an instance of {\sc Weighted Multicut}, which we ensure is on a chordal graph. Formally, the {\sc Weighted Multicut} problem is defined as follows. 

\defparproblemOpt{{\sc Weighted Multicut}}{An undirected graph $G$, a weight function $w: V(G)\rightarrow\mathbb{R}$ and a set $T=\{(s_1,t_1),\ldots,(s_k,t_k)\}$ of $k$ pairs of vertices of $G$.}{Find a minimum weight subset $S\subseteq V(G)$ such that for any pair $(s_i,t_i)\in{\cal T}$, $G- S$ does not have any path between $s_i$ and $t_i$.}

For {\sc Weighted Multicut} on chordal graphs, no constant-factor approximation algorithm was previously known.
We remark that {\sc Weighted Multicut} is NP-hard on trees~\cite{GVY96}, and hence it is also NP-hard on chordal graphs. We design the first such algorithm, which our main algorithm employs as a black box. 
\begin{theorem}
\label{thm:multicutGen}
{\sc Weighted Multicut} admits a constant-factor approximation algorithm on chordal graphs. 
\end{theorem}

This algorithm is inspired by the work of Garg, Vazirani and Yannakakis on {\sc Weighted Multicut} on trees \cite{GVY96}. Here, we carefully exploit the well-known characterization of the class of chordal graphs as the class of graphs that admit clique forests. We believe that this result is of independent interest. The algorithm by  Garg, Vazirani and Yannakakis \cite{GVY96} is a classic primal-dual algorithm. A more recent algorithm, by Golovin, Nagarajan and Singh \cite{golovin06}, uses total modularity to obtain a different algorithm for {\sc Multicut} on trees.

 \vspace{-.3cm}
\paragraph{Weighted Distance Hereditary Vertex Deletion.} We start by formally defining the \wDHfull problem. 

\defparproblemOpt{\wDHfull (\wDH)}{An undirected graph $G$ and a weight function $w: V(G)\rightarrow\mathbb{R}$.}
{Find a minimum weight subset $S\subseteq V(G)$ such that $G- S$ is a distance hereditary graph.}
 
 A graph $G$ is a \emph{distance hereditary graph} 
 (also called a completely separable graph~\cite{hammer1990completely}) if the distances between vertices in every connected induced subgraph of $G$ are the same as in the graph $G$. Distance hereditary graphs were named and first studied by Hworka~\cite{howorka1977characterization}. However,  an equivalent family of graphs was earlier studied by Olaru and Sachs~\cite{sachs1970berge} and shown to be perfect. 
 It was later discovered that these graphs are precisely the graphs of rankwidth 1~\cite{Oum05}. 
 %That is, distance hereditary graphs are precisely those graphs which has rankwidth one. 
 %Let us discuss the notion of rankwidth of a graph.
 
 Rankwidth is a graph parameter introduced by Oum and Seymour~\cite{OumS06} to approximate yet another graph parameter called Cliquewidth. 
 The notion of cliquewidth was defined by  Courcelle and Olariu~\cite{CourcelleO00} as a measure of how ``clique-like''  the input graph is. This is similar to the notion of treewidth, which measures how ``tree-like'' the input graph is. 
 One of the main motivations was that several {\sf NP}-complete problems
 become tractable on the family of cliques (complete graphs), the assumption was that these algorithmic properties extend to ``clique-like'' graphs~\cite{CourcelleMR00}. 
%Indeed, it has been shown that graphs of bounded cliquewidth have good algorithmic properties.  That is,  many {\sf NP}-hard graph problems can be solved in polynomial time, if the input graphs have bounded cliquewidth~\cite{CourcelleMR00}. In fact, every graph that has treewdith at most $\eta$, its cliquewidth is upper bounded by $3\cdot 2^{\eta -1}$~\cite{CorneilR05}. However, the converse does not hold. That is,  if a graph has bounded cliquewidth it does not imply that its treewidth is bounded. For example, the family of cliques has unbounded treewidth, but the cliquewidth of this family is at most $2$.  
However, computing cliquewidth and the corresponding cliquewidth decomposition seems to be computationally intractable. 
This then motivated the notion of rankwidth, which is a graph parameter that approximates cliquewidth well while also being algorithmically tractable~\cite{OumS06,Oum08}. For more information on cliquewidth and rankwidth, we refer to the surveys by Hlinen{\'{y}} et al.~\cite{HlinenyOSG08} and Oum~\cite{Oum16}. 

As algorithms for {\sc Treewidth-$\eta$ Vertex Deletion} are applied as subroutines to solve many graph problems, we believe that algorithms for \wretad (\wretadshort) will be useful in this respect.
In particular,  {\sc Treewidth-$\eta$ Vertex Deletion} has been considered in designing efficient approximation, kernelization and fixed parameter tractable algorithms for \wpfdshort and its unweighted counterpart \pfd~\cite{BansalRU1,FominLMPS16,FominLRS11,FominLS12,FominLST10}. Along  similar lines, we believe that \wretadshort and its unweighted counterpart will be useful in designing efficient approximation, kernelization and fixed parameter tractable algorithms for \genfvd where $\cal F$ is characterized by a finite family of forbidden {\em vertex minors}~\cite{Oum05}.

Recently, Kim and Kwon~\cite{OsmallfreeDHVDkim} designed an 
$\OO(\opt^2 \log n)$-factor approximation algorithm for \DHfull (\DHsmall). This result implies that \DHsmall  admits an $\OO(n^{2/3}\log n)$-factor approximation algorithm. In this paper,
we take first step towards obtaining good approximation algorithm for  \wretadshort  by 
%we design 
designing a \afg-factor approximation algorithm for 
%\wronedshort\ or
\wDH. 
%In particular,  we obtain the following theorem. 
\begin{theorem}
\label{thm:newApprox2DH}
\wDH or \wronedshort  admits an \afdh-factor approximation algorithm. 
\end{theorem}
%\todo[inline]{Say about the result of kim and company .. and that our approximation}
We note that several steps of our approximation algorithm for \wronedshort can be generalized for an approximation algorithm for \wretadshort  and thus we believe that our approach should yield an  
\afg-factor approximation algorithm for  \wretadshort. We leave that as an interesting open problem for the future.

%Approximation algorithms are powerful tools that are often applied to deal with computationally hard problems in real world applications. \hly{TODO: Define chordal graphs, Chordal Vertex Deletion, Multicut.}
%
%\hly{TODO: Motivation (1-2 Paragraphs); why the class of chordal graphs is important; natural generalization of wfvs; why it is important to consider weights.}
%
%\hly{TODO: Previous Work (1-2 Paragraphs); focus only on approximation algorithms.} \cite{GVY96}
%%%%%%%Theorem~\ref{thm:approx_thm} has a number of interesting applications. ADD HERE BIDIMENSIONALITY,  what else?

%
%
%This is very interesting, in view of the fact that there is for example no
% known approximation algorithm with bounded worst-case ratio for the
% feedback-node set (or any other problem of the class), whereas the node cover problem can be easily approximated within ratio 2, but also because it would shed more  
% light into the nature of NP-complete problems from the combinatorial point of view and into their
%behaviour with respect to approximation algorithms \cite{LewisY80}.
%
%
% 
%Lund-Yannakakis
%\cite{LundY93}: The	vertex deletion	problem	for	any	hereditary	property	with	a finite number of minimal forbidden subgraphs can be approximated with a constant ratio. They also conjectured that 
%  for every nontrivial, hereditary property with an infinite number of minimal forbidden subgraphs the  vertex deletion problem cannot be approximated with constant ratio. 

%
%!TEX root = mainChordal.tex

\section{Preliminaries}
%\todo[inline]{read preliminaries, fix minor free family, graphs} 
For a positive  integer $k$, we use  $[k]$  as a shorthand for  $\{1,2,\ldots,k\}$. Given a function $f: A\rightarrow B$ and a subset $A'\subseteq A$, we let $f|_{A'}$ denote the function $f$ restricted to the domain $A'$.

\bigskip
{\noindent\bf Graphs.} Given a graph $G$, we let $V(G)$ and $E(G)$ denote its vertex-set and edge-set, respectively. In this paper, we only consider undirected graphs. We let $n=|V(G)|$ denote the number of vertices in the graph $G$, where $G$ will be clear from context. The \emph{open neighborhood}, or simply the \emph{neighborhood}, of a vertex $v\in V(G)$ is defined as $N_G(v) = \{w \mid \{v,w\} \in E(G) \}$. The \emph{closed neighborhood} of $v$ is defined as $N_G[v] = N_G(v) \cup \{ v \}$. 
The \emph{degree} of $v$ is defined as $d_G(v) = |N_G(v)|$. We can extend the definition of the neighborhood of a vertex to a set of vertices as follows. Given a subset $U \subseteq V(G)$, $N_G(U) = \bigcup_{u\in U} N_G(u)$ and $N_G[U] = \bigcup_{u\in U} N_G[u]$. The \emph{induced subgraph} $G[U]$ is the graph with vertex-set $U$ and edge-set $\{\{u,u'\}~|~u,u'\in U, \text{ and } \{u,u'\} \in E(G)\}$. Moreover, we define $G - U$ as the induced subgraph $G[V(G) \setminus U]$. We omit subscripts when the graph $G$ is clear from context. For graphs $G$ and $H$, by $G \cap H$, we denote the graph with vertex set as $V(G)\cap V(H)$ and edge set as $E(G) \cap E(H)$. An \emph{independent set} in $G$ is a set of vertices such that there is no edge in $G$ between any pair of vertices in this set. The \emph{independence number} of $G$, denoted by $\alpha(G)$, is defined as the cardinality of the largest independent set in $G$. A \emph{clique} in $G$ is a set of vertices such that there is an edge in $G$ between every pair of vertices in this set.

A \emph{path} $P=(x_1, x_2, \ldots, x_\ell)$ in $G$ is a subgraph of $G$ where
$V(P) = \{ x_1, x_2, \ldots, x_\ell \} \subseteq V(G)$ and $E(P) = \{ \{x_1,x_2\}, \{x_2, x_3\}, \ldots, \{x_{\ell-1},x_\ell\}\}\subseteq E(G)$, where $\ell\in[n]$.
The vertices $x_1$ and $x_\ell$ are called the \emph{endpoints} of the path $P$ and the remaining vertices in $V(P)$ are called the \emph{internal vertices} of $P$.
We also say that $P$ is a path between $x_1$ and $x_\ell$ or connects $x_1$ and $x_\ell$. A \emph{cycle} $C=(x_1, x_2, \ldots, x_\ell)$ in $G$ is a subgraph of $G$ where
$V(C) = \{ x_1, x_2, \ldots, x_\ell \} \subseteq V(G)$ and $E(C) = \{ \{x_1,x_2\}, \{x_2, x_3\}, \ldots, \{x_{\ell-1},x_\ell\}, \{x_\ell, x_1\}\} \subseteq E(G)$, i.e., it is a path with an additional edge between $x_1$ and $x_\ell$.
% Let $P$ be a path in the graph $G$ on at least three vertices.
% We say that $\{u,v\} \in E(G)$ is a \emph{chord} of $P$ if $u,v \in V(P)$ but $\{u,v\} \notin E(P)$.
%For a cycle $C$ on at least four vertices, we say that $\{u,v\} \in E(G)$ is a {\em chord} of $C$ if $u,v \in V(C)$ but $\{u,v\} \notin E(C)$. A cycle $C$ is \emph{chordless} if it contains at least four vertices and has no chords. <--- MOVED
% A path $P$ or cycle $C$ is \emph{chordless} if it has no chords.
% Chordless cycles are also called \emph{holes} of the graph.
The graph $G$ is \emph{connected} if there is a path between every pair of vertices in $G$,
otherwise $G$ is \emph{disconnected}. A connected graph without any cycles is a \emph{tree}, and a collection of trees is a \emph{forest}. A maximal connected subgraph of $G$ is called a \emph{connected component} of $G$. Given a function $f: V(G)\rightarrow\mathbb{R}$ and a subset $U\subseteq V(G)$, we denote $f(U)=\sum_{v\in U}f(v)$. Moreover, we say that a subset $U\subseteq V(G)$ is a {\em balanced separator for $G$} if for each connected component $C$ in $G- U$, it holds that $|V(C)|\leq \frac{2}{3}|V(G)|$. We refer the reader to~\cite{diestel-book} for details on standard graph theoretic notations and terminologies that are not explicitly defined here.

\paragraph{Forest Decompositions.} A \emph{forest decomposition} of a graph $G$ is a pair $(F,\beta)$ where $F$ is forest, and $\beta:V(T) \rightarrow 2^{V(G)}$ is a function that satisfies the following:
\begin{enumerate}[(i)]
\setlength{\itemsep}{-2pt}
    \item $\bigcup_{v \in V(F)} \beta(v) = V(G)$;
    \item for any edge $\{v,u\} \in E(G)$, there is a node $w \in V(F)$ such that $v,u \in \beta(w)$;
    \item for any $v \in V(G)$, the collection of nodes $T_v = \{ u \in V(F) \mid v \in \beta(u)\}$ is a subtree~of~$F$.
\end{enumerate}

For $v \in V(F)$, we call $\beta(v)$ the \emph{bag} of $v$, and for the sake of clarity of presentation, we sometimes use $v$ and $\beta(v)$ interchangeably. We refer to the vertices in $V(F)$ as {\em nodes}.
A \emph{tree decomposition} is a forest decomposition where $F$ is a tree. For a graph $G$, by ${\sf tw}(G)$ we denote the minimum over all possible \emph{tree decompositions} of $G$, the maximum size of a bag minus one in that \emph{tree decomposition}. 

%\paragraph{Branch decompositions and rankwidth.}
%\hly{TODO: Define rank width.}

\paragraph{Minors.}
Given a graph $G$ and an edge $\{u,v\} \in E(G)$, the graph $G/e$ denotes the graph obtained from $G$ by contracting the
edge $\{u,v\}$, that is, the vertices $u,v$ are deleted from $G$ and a new vertex $uv^{\star}$ is added to $G$ which is adjacent to
the all the neighbors of $u,v$ previously in $G$ (except for $u,v$). A graph $H$ that is obtained by a sequence of edge contractions in $G$ is said to be a contraction of $G$. A graph $H$
is a \emph{minor} of a $G$ if $H$ is the contraction of some subgraph of $G$. We say that a graph $G$ is $F$-minor free when $F$ is not a minor of $G$. Given a family $\cal F$ of graphs, we say that a graph $G$ is $\cal F$-minor free, if for all $F \in {\cal F}$, $F$ is not a minor of $G$. It is well known that if $H$ is a minor of $G$, then ${\sf tw}(H) \leq {\sf tw}(G)$. A graph is \emph{planar} if it is $\{K_5,K_{3,3}\}$-minor free~\cite{diestel-book}. Here, $K_5$ is a clique on $5$ vertices and $K_{3,3}$ is a complete bipartite graph with both sides of bipartition having $3$ vertices. %We will be using the following useful proposition about excluding a planar graph as a minor.

%\begin{proposition}(\cite{Robertson:graph-minor}) \label{prop:plaar-minor-free}
%Let $\cal F$ be a finite family of finite graphs such that $\cal F$ contains at least one planar graph. Then, any graph $G$ which is $\cal F$-minor free has ${\sf tw}(G) \leq c$. Here, $c=c(\cal{F})$ is a constant depending on the family $\cal F$.
%\end{proposition}

\bigskip
{\noindent\bf Chordal Graphs.} Let $G$ be a graph. For a cycle $C$ on at least four vertices, we say that $\{u,v\} \in E(G)$ is a {\em chord} of $C$ if $u,v \in V(C)$ but $\{u,v\} \notin E(C)$. A cycle $C$ is \emph{chordless} if it contains at least four vertices and has no chords. The graph $G$ is a \emph{chordal graph} if it has no chordless cycle as an induced subgraph. A \emph{clique forest} of $G$ is a forest decomposition of $G$ where every bag is a maximal clique. The following lemma shows that the class of chordal graphs is exactly the class of graphs which have a clique forest.

\begin{lemma}[\cite{Golumbic80}]\label{lem:cliqueForest}
A graph $G$ is a chordal graph if and only if $G$ has a clique forest. Moreover, a clique forest of a chordal graph  can be constructed in polynomial time.
\end{lemma}

Given a subset $U\subseteq V(G)$, we say that $U$ {\em intersects} a chordless cycle $C$ in $G$ if $U\cap V(C)\neq\emptyset$.
Observe that if $U$ \emph{intersects} every chordless cycle of $G$, then $G- U$ is a chordal graph.

\section{Approximation Algorithm for \wpfdshort}\label{sec:planar}

In this section we prove Theorem~\ref{thm:approx_thmpfd}. 
We can assume that the weight $w(v)$ of each vertex $v\in V(G)$ is positive, else we can insert $v$ into any solution. %We can do so because $\mathscr{F}$-minor free property is minor closed. 
Below we state a result from~\cite{Robertson:graph-minor},  which will be useful in our algorithm. 

\begin{proposition}[\cite{Robertson:graph-minor}] \label{prop:plaar-minor-free}
Let $\mathscr F$ be a finite set of graphs such that $\mathscr F$ contains a planar graph. Then, any graph $G$ 
that excludes any graph from $\mathscr F$ as a minor 
%which is $\mathscr{F}$-minor free 
satisfies ${\sf tw}(G) \leq c=c(\mathscr{F})$.
\end{proposition}

We let $c=c(\mathscr{F})$ to be the constant returned by Proposition~\ref{prop:plaar-minor-free}. The approximation algorithm for \wpfdshort comprises of two components. The first component handles the special case where the vertex set of input graph $G$ can be partitioned into two sets $C$ and $X$ such that $|C| \leq c+1$ and $H=G[X]$ is an $\mathscr{F}$-minor free graph. We note that there can be edges between vertices in $C$ and vertices in $H$. We show that for these special instances, in polynomial time we can compute the size of the optimum solution and a set realizing it. 

The second component is a recursive algorithm that solves general instances of the problem. Here, we gradually disintegrate the general instance until it becomes an instance of the special type where we can resolve it in polynomial time. More precisely, for each guess of $c+1$ sized subgraph $M$ of $G$, we find a small separator $S$ (using an approximation algorithm) that together with $M$ breaks the input graph into two graphs significantly smaller than their origin. It first removes $M \cup S$, and solves each of the two resulting subinstances by calling itself recursively; then, it inserts $M$ back into the graph, and uses the solutions it obtained from the recursive calls to construct an instance of the special case which is then solved by the first component. 

\subsection{Constant sized graph + $\mathscr{F}$-minor free graph}\label{sec:exactPF}

We first handle the special case where the input graph $G$ consists of a graph $M$ of size at most $c+1$ and an $\mathscr{F}$-minor free graph $H$. We refer to this algorithm as \alg{Special-WP}. More precisely, along with the input graph $G$ and the weight function $w$, we are also given a graph $M$ with at most $c+1$ vertices and an $\mathscr{F}$-minor free graph $H$ such that $V(G)=V(M) \cup V(H)$, where the vertex-sets $V(M)$ and $V(H)$ are disjoint. Note that the edge-set $E(G)$ may contain edges between vertices in $M$ and vertices in $H$. We will show that such instances may be solved optimally in polynomial time. We start with the following easy observation.

\begin{observation} 
Let $G$ be a graph with $V(G)= X \uplus Y$, such that $|X| \leq c+1$ and $G[Y]$ is an $\mathscr{F}$-minor free graph. Then,  the  treewidth of $G$ is at most $2c+1$.
\end{observation}

\begin{lemma}
Let $G$ be a graph of treewidth $t$ with a non-negative weight function $w$ on the vertices, and let $\mathscr{F}$ be a finite family of graphs.
Then, one can compute a minimum weight vertex set $S$ such that $G - S$ is $\mathscr{F}$-minor free, in time ${\it f}(q,t)\cdot n$, where $n$ is the number of vertices in $G$ and $q$ is a constant that depends only on $\mathscr{F}$.
\end{lemma}
\begin{proof} This follows from the fact that finding such a set $S$ is expressible as an $\sf MSO$-optimization formula $\phi$ whose length, $q$, depends only on the family $\mathscr{F}$~\cite{FominLMS12}. Then, by Theorem 7~\cite{Borie1992}, we can compute an optimal sized set $S$ in time ${\it f}(q,t) \cdot n$.
\end{proof}

Now, we apply the above lemma to the graph $G$ and the family $\mathscr{F}$,
and obtain a minimum weight set $S$ such that $G- S$ is $\mathscr{F}$-minor free.

%\subsection{Constant sized graph + $\mathscr{F}$-minor free graph}\label{sec:exactPF}
%
%In this subsection we handle the special case where the input graph $G$ consists of a graph $M$ of size at most $c+1$ and an $\mathscr{F}$-minor free graph $H$. More precisely, along with the input graph $G$ and the weight function $w$, we are also given a graph $M$ with at most $c+1$ vertices and an $\mathscr{F}$-minor free graph $H$ such that $V(G)=V(M) \cup V(H)$, where the vertex-sets $V(M)$ and $V(H)$ are disjoint. Note that the edge-set $E(G)$ may contain edges between vertices in $M$ and vertices in $H$. This special instance can be solved in polynomial time. \hly{TODO: Write how we solve in poly time.}

%%%%%%%%%%%%%%
\subsection{General Graphs}\label{sec:approxGenGraphspfd}
We proceed to handle general instances by developing a $\approxDH$-factor approximation algorithm for \wpfdshort, \alg{Gen-WP-APPROX}, thus proving the correctness of Theorem \ref{thm:approx_thmpfd}. The exact value of the constant $d$ will be determined later.
%The exact value of the constant $d \geq \max\{96,2c\}$ is determined later.\footnote{Recall that $c$ is the constant we fixed to ensure that the approximation ratio of \alg{APPROX} is bounded by $c\cdot\log n$.} This algorithm is based on recursion, and during its execution, we often encounter instances that are of the form of the Clique+Chordal special case of \wcdel, which will be dealt with using the algorithm \alg{APPROX} of Section \ref{sec:approxCliqueChordal}.

\bigskip
{\noindent\bf  Recursion.} We define each call to our algorithm \alg{Gen-WP-APPROX} to be of the form $(G',w')$, where $(G',w')$ is an instance of \wpfdshort\ such that $G'$ is an induced subgraph of $G$, and we denote $n'=|V(G')|$.  

\bigskip
{\noindent\bf Goal.} For each recursive call \alg{Gen-WP-APPROX}$(G',w')$, we aim to prove the following.

\begin{lemma}\label{lemma:approxRecursiveCallGenPF}
\alg{Gen-WP-APPROX} returns a solution that is at least $\opt$ and at most $\frac{d}{2}\cdot\log^2n'\cdot\opt$. Moreover, it returns a subset $U \subseteq V(G')$ that realizes the solution.
\end{lemma}

At each recursive call, the size of the graph $G'$ becomes smaller. Thus, when we prove that Lemma~\ref{lemma:approxRecursiveCallGenPF} is true for the current call, we assume that the approximation factor is bounded by $\frac{d}{2}\cdot\log^2\widehat{n}\cdot\opt$ for any call where the size $\widehat{n}$ of the vertex-set of its graph is strictly smaller than~$n'$.

  %  \bigskip
  %{\noindent\bf Initialization.} Initially, we set $(G',w')=(G,w)$. 
  \bigskip
{\noindent\bf Termination.} 
In polynomial time we can test whether $G'$ has a minor $F \in \mathscr{F}$~\cite{robertson-minor-testing}. Furthermore, for each $M \subseteq V(G)$ of size at most $c+1$, we can check if $G - M$ has a minor $F \in \mathscr{F}$. If $G - M$ is $\mathscr{F}$-minor free then we are in a special instance, where $G - M$ is $\mathscr{F}$ minor free and $M$ is a constant sized graph. We optimally resolve this instance in polynomial time using the algorithm \alg{Special-WP}. Since we output an optimal sized solution in the base cases, we thus ensure that at the base case of our induction Lemma \ref{lemma:approxRecursiveCallGenPF} holds.

\bigskip
{\noindent\bf Recursive Call.} For the analysis of a recursive call, let $S^*$ denote a hypothetical set that realizes the optimal solution $\opt$ of the current instance $(G',w')$. Let $(F,\beta)$ be a forest decomposition of $G'- S^*$ of width at most $c$, whose existence is guaranteed by Proposition~\ref{prop:plaar-minor-free}. Using standard arguments on forests we have the following observation.

\begin{observation}\label{obs:standardForest}
There exists a node $v\in V(F)$ such that $\beta(v)$ is a balanced separator for $G'- S^*$.
\end{observation} 

From Observation~\ref{obs:standardForest} we know that there exists a node $v\in V(F)$ such that $\beta(v)$ is a balanced separator for $G'- S^*$. This together with the fact that $G'- S^*$ has treewidth at most $c$ results in the following observation.

\begin{observation} \label{obs:balanced-separator-PF}
There exist a subset $M \subseteq V(G')$ of size at most $c+1$ and a subset $S\subseteq V(G') \setminus M$ of weight at most $\opt$ such that $M \cup S$ is a balanced separator for $G'$.
\end{observation}

This gives us a polynomial time algorithm as stated in the following lemma.

\begin{lemma}\label{lem:dividePF}
There is a deterministic (randomized) algorithm which in polynomial-time finds $M \subseteq V(G')$ of size at most $c+1$ and a subset $S\subseteq V(G') \setminus M$ of weight at most $q \cdot \log n' \cdot\opt$  ($q^*\cdot \sqrt{\log n'} \cdot\opt$) for some fixed constant $q$ ($q^*$) such that $ M \cup S$ is a balanced separator for $G'$.
\end{lemma}
\begin{proof}
Note that we can enumerate every $M \subseteq V(G')$ of size at most $c+1$ in time $\OO(n^{c})$. For each such $M$, we can either run the randomized $q^* \cdot \sqrt{\log n'}$-factor approximation algorithm by 
Feige et al.~\cite{FeigeHL08} or the deterministic $q \cdot \log n'$-factor approximation algorithm by Leighton and Rao \cite{BalancedSeparator} to find a balanced separator $S_M$ of $G'- M$. Here, $q$ and $q^*$ are fixed constants. By Observation~\ref{obs:balanced-separator-PF}, there is a set $S$ in $\{S_M: M \subseteq V(G')$ and $M \leq c+1\}$ such that $w(S) \leq q\cdot \log n' \cdot\opt$ ($w(S) \leq q^*\cdot \sqrt{\log n'}\cdot\opt$). Thus, the desired output is a pair $(M,S)$ where $M$ is one of the vertex subset of size at most $c+1$ such that $S_M=S$.
%we run the $q \cdot \log n'$-factor approximation algorithm by Leighton and Rao \cite{BalancedSeparator} to find a balanced separator $S_M$ of $G'\setminus M$. Here, $q$ is some fixed constant. We let $S$ denote some set of minimum weight among the sets in $\{S_M: M \subseteq V(G')$ and $M \leq c+1\}$. By Observation~\ref{obs:balanced-separator-PF}, $w(S) \leq q\cdot\log n'\cdot\opt$. Thus, the desired output is a pair $(M,S)$ where $M$ is one of the vertex subset of size at most $c+1$ such that $S_M=S$.
\end{proof}

We call the algorithm in Lemma~\ref{lem:dividePF} to obtain a pair $(M,S)$. Since $M\cup S$ is a balanced separator for $G'$, we can partition the set of connected components of $G'- (M\cup S)$ into two sets, ${\cal A}_1$ and ${\cal A}_2$, such that for $V_1=\bigcup_{A\in{\cal A}_1}V(A)$ and $V_2=\bigcup_{A\in{\cal A}_2}V(A)$ it holds that $n_1,n_2\leq\frac{2}{3}n'$ where $n_1=|V_1|$ and $n_2=|V_2|$. We remark that we use different 
algorithms for finding a balanced separator in Lemma~\ref{lem:dividePF} based on whether we are looking for a randomized algorithm or a deterministic algorithm.  
%The $\OO(\log n)$-factor approximation algorithm by Leighton and Rao \cite{BalancedSeparator} is deterministic while the $\OO(\sqrt{\log n})$-factor approximation algorithm by Feige et al.~\cite{FeigeHL08}, as the algorithm by Feige et al. is randomized. 

%the $\OO(\sqrt{\log n})$-factor approximation algorithm by 
%$\OO(\log n)$-factor approximation algorithm by Leighton and Rao \cite{BalancedSeparator} in Lemma~\ref{lem:dividePF}  to find the balanced separator instead of the $\OO(\sqrt{\log n})$-factor approximation algorithm by Feige et al.~\cite{FeigeHL08}, as the algorithm by Feige et al. is randomized. 

Next, we define two inputs of (the general case of) \wpfdshort: $I_1=(G'[V_1],w'|_{V_1})$ and $I_2=(G'[V_2],w'|_{V_2})$. Let $\opt_1$ and $\opt_2$ denote the optimal solutions to $I_1$ and $I_2$, respectively. Observe that since $V_1\cap V_2=\emptyset$, it holds that $\opt_1+\opt_2\leq \opt$. We solve each of the subinstances by recursively calling algorithm \alg{Gen-WP-APPROX}. By the inductive hypothesis, we thus obtain two sets, $S_1$ and $S_2$, such that $G'[V_1]- S_1$ and $G'[V_2]- S_2$ are $\mathscr{F}$-minor free graphs, and $w'(S_1)\leq \frac{d}{2}\cdot\log^2n_1\cdot\opt_1$ and $w'(S_2)\leq \frac{d}{2}\cdot\log^2n_2\cdot\opt_2$.

We proceed by defining an input of the special case of \wpfdshort: $J=(G'[(V_1\cup V_2\cup M)\setminus (S_1\cup S_2)],w'|_{(V_1\cup V_2\cup M)\setminus (S_1\cup S_2)})$. Observe that $G'[V_1\setminus S_1]$ and $G'[V_2\setminus S_2]$ are $\mathscr F$-minor free graphs and there are no edges between vertices in $V_1$ and vertices in $V_2$ in $G'-M$, and $M$ is of constant size. Therefore, we resolve this instance by calling algorithm \alg{Special-WP}. We thus obtain a set, $\widehat{S}$, such that $G'[(V_1\cup V_2\cup M)\setminus (S_1\cup S_2\cup\widehat{S})]$ is a $\mathscr{F}$-minor graph, and $w'(\widehat{S})\leq \opt$ (since $|(V_1\cup V_2\cup M)\setminus (S_1\cup S_2)|\leq n'$ and the optimal solution of each of the special subinstances is at most $\opt$).

Observe that any obstruction in $G' - S$ is either completely contained in $G'[V_1]$, or completely contained in $G'[V_2]$, or it contains at least one vertex from $M$. This observation, along with the fact that $G'[(V_1\cup V_2\cup M)\setminus (S_1\cup S_2\cup\widehat{S})]$ is a $\mathscr{F}$-minor free graph, implies that $G' - T$ is a $\mathscr{F}$-minor free graph where $T= S \cup S_1\cup S_2\cup\widehat{S}$. Thus, it is now sufficient to show that $w'(T)\leq\frac{d}{2}\cdot (\log n')^2 \cdot \opt$.

By the discussion above, we have that

\[\begin{array}{ll}
w'(T) & \leq w'(S) + w'(S_1) + w'(S_2) + w'(\widehat{S})\\
& \leq q\cdot\log n'\cdot\opt + \frac{d}{2}\cdot((\log n_1)^2 \cdot \opt_1 + (\log n_2)^2 \cdot \opt_2) + \opt\\
\end{array}\]

Recall that $n_1,n_2\leq\frac{2}{3}n'$ and $\opt_1+\opt_2\leq\opt$. Thus, we have that

\[\begin{array}{ll}
w'(T) & < q \cdot \log n' \cdot \opt + \frac{d}{2}\cdot (\log \frac{2}{3}n')^2 \cdot \opt + \opt\\
& < \frac{d}{2} \cdot (\log n')^2 \cdot \opt + \log n' \cdot \opt \cdot (q+1+ \frac{d}{2}\cdot (\log \frac{3}{2})^2- \frac{d}{2}\cdot 2 \cdot \log \frac{3}{2}).
\end{array}\]

Overall, we conclude that to ensure that $w'(T)\leq\frac{d}{2}\cdot\log^2n'\cdot\opt$, it is sufficient to ensure that $q+1+ \frac{d}{2}\cdot (\log \frac{3}{2})^2- \frac{d}{2}\cdot 2 \cdot \log \frac{3}{2}\leq 0$, which can be done by fixing $\displaystyle{d = \frac{2}{2 \log\frac{3}{2} - (\log \frac{3}{2})^2}\cdot(q+1)}$.

If we use the $\OO(\sqrt{\log n})$-factor approximation algorithm by Feige et al.~\cite{FeigeHL08} 
for finding a balance separator in Lemma~\ref{lem:dividePF}, then we can do the analysis similar to the deterministic case and obtain a randomized factor-\afpfdr approximation algorithm for \wpfdshort.

%
%!TEX root=mainChordal.tex

\section{Weighted Chordal Vertex Deletion on General Graphs}\label{sec:newApproxAlg}
In this section we prove Theorem~\ref{thm:newApprox2}. 
%the following theorem.
%\begin{theorem}\label{thm:newApprox2}
%\wcdel admits an $\OO(\log^{2} n)$-factor approximation algorithm.
%\end{theorem}
Clearly, we can assume that the weight $w(v)$ of each vertex $v\in V(G)$ is positive, else we can insert $v$ into any solution. 

Roughly speaking, our approximation algorithm consists of two components. The first component handles the special case where the input graph $G$ consists of a clique $C$ and a chordal graph $H$. Here, we also assume that the input graph has no ``short'' chordless cycle. This component is comprised of a recursive algorithm that is based on the method of divide and conquer. The algorithm keeps track of a fractional solution ${\bf x}$ of a special form that it carefully manipulated at each recursive call, and which is used to analyze the approximation ratio. In particular, we ensure that ${\bf x}$ does not assign high values, and that it assigns 0 to vertices of the clique $C$ as well as vertices of some other cliques. To divide a problem instance into two instances, we find a maximal clique $M$ of the chordal graph $H$ that breaks $H$ into two ``simpler'' chordal graphs. The clique $C$ remains intact at each recursive call, and the maximal clique $M$ is also a part of both of the resulting instances. Thus, to ensure that we have simplified the problem, we measure the complexity of instances by examining the maximum size of an independent set of their graphs. Since the input graph has no ``short'' chordless cycle, the maximum depth of the recursion tree is bounded by $\OO(\log n)$. Moreover, to guarantee that we obtain instances that are independent, we incorporate multicut constraints while ensuring that we have sufficient ``budget'' to satisfy them. We ensure that these multicut constraints are associated with chordal graphs, which allows us to utilize the algorithm we design in Section \ref{sec:multicut}.

The second component is a recursive algorithm that solves general instances of the problem. Initially, it easily handles ``short'' chordless cycles. Then, it gradually disintegrates a general instance until it becomes an instance of the special form that can be solved using the first component. More precisely, given a problem instance, the algorithm divides it by finding a maximal clique $M$ (using an exhaustive search which relies on the guarantee that $G$ has no ``short'' chordless cycle) and a small separator $S$ (using an approximation algorithm) that together break the input graph into two graphs significantly smaller than their origin. It first removes $M\cup S$ and solves each of the two resulting subinstances by calling itself recursively; then, it inserts $M$ back into the graph, and uses the solutions it obtained from the recursive calls to construct an instance of the special case solved by the first component.

\subsection{Clique+Chordal Graphs}\label{sec:approxCliqueChordal}

In this subsection we handle the special case where the input graph $G$ consists of a clique $C$ and a chordal graph $H$. More precisely, along with the input graph $G$ and the weight function $w$, we are also given a clique $C$ an a chordal graph $H$ such that $V(G)=V(C)\cup V(H)$, where the vertex-sets $V(C)$ and $V(H)$ are disjoint. Here, we also assume that $G$ has no chordless cycle on at most 48 vertices. Note that the edge-set $E(G)$ may contain edges between vertices in $C$ and vertices in $H$. We call this special case the {\em Clique+Chordal special case}. Our objective is to prove the following result.

\begin{lemma}\label{lemma:newApproxSpecial}
The Clique+Chordal special case of \wcdel admits an $\OO(\log n)$-factor approximation algorithm.
\end{lemma}

We assume that $n\geq 64$,\footnote{This assumption simplifies some of the calculations ahead.} else the input instance can be solve by brute-force. Let $c$ be a fixed constant (to be determined). In the rest of this subsection, we design a $c\cdot\log n$-factor approximation algorithm for the Clique+Chordal special case of \wcdel.

\bigskip
{\noindent\bf Recursion.} Our approximation algorithm is a recursive algorithm. We call our algorithm \alg{CVD-APPROX}, and define each call to be of the form $(G',w',C,H',{\bf x})$. Here, $G'$ is an induced subgraph of $G$ such that $V(C)\subseteq V(G')$, and $H'$ is an induced subgraph of $H$. The argument ${\bf x}$ is discussed below. We remark that we continue to use $n$ to refer to the size of the vertex-set of the input graph $G$ rather than the current graph $G'$.

\bigskip
{\noindent\bf Arguments.} While the execution of our algorithm progresses, we keep track of two arguments: the size of a maximum independent set of the current graph $G'$, denoted by $\alpha(G')$, and a fractional solution $\bf x$. Due to the special structure of $G'$, the computation of $\alpha(G')$ is simple:

\begin{observation}
The measure $\alpha(G')$ can be computed in polynomial time.
\end{observation}

\begin{proof}
Any maximum independent set of $G'$ consists of at most one vertex from $C$ and an independent set of $H'$.
It is well known that the computation of the size of a maximum independent set of a chordal graph can be performed in polynomial time \cite{Golumbic80}. Thus, we can compute $\alpha(H')$ in polynomial time. Next, we iterate over every vertex $v\in V(C)$, and we compute $\alpha_v = \alpha(\widehat{H}) + 1$ for the graph $\widehat{H} = H'\setminus N_{G'}(v)$ in polynomial time (since $\widehat{H}$ is a chordal graph). Overall, we return $\max\{\alpha(H'),\max_{v\in V(C)}\{\alpha_v\}\}$.
\end{proof}

The necessity of tracking $\alpha(G')$ stems from the fact that our recursive algorithm is based on the method of divide-and-conquer, and to ensure that when we divide the current instance into two instances we obtain two ``simpler'' instances,  we need to argue that some aspect of these instances has indeed been simplified. Although this aspect cannot be the size of the instance (since the two instances can share many common vertices), we show that it can be the size of a maximum independent set. 
 
A fractional solution ${\bf x}$ is a function ${\bf x}: V(G')\rightarrow [0,\infty)$ such that for every chordless cycle $Q$ of $G'$ it holds that ${\bf x}(V(Q))\geq 1$. An optimal fractional solution minimizes the weight $w'({\bf x})=\sum_{v\in V(G')}w'(v)\cdot{\bf x}(v)$. Clearly, the solution to the instance $(G',w')$ of \wcdel is at least as large as the weight of an optimal fractional solution.  Although we initially compute an optimal fractional solution ${\bf x}$ (at the initialization phase that is described below), during the execution of our algorithm, we manipulate this solution so it may no longer be optimal. Prior to any call to \alg{CVD-APPROX} with the exception of the first call, we ensure that ${\bf x}$ satisfies the following invariants:
\begin{itemize}
\item {\bf Low-Value Invariant}: For any $v\in V(G')$, it holds that ${\bf x}(v) < (\frac{c\cdot\log n + 9}{c\cdot\log n})^{\delta}\cdot \frac{1}{c\cdot\log n}$. Here, $\delta$ is the depth of the current recursive call in the recursion tree.\footnote{The depth of the first call is defined to be 1.}
\item {\bf Zero-Clique Invariant}: For any $v\in V(C)$, it holds that ${\bf x}(v) = 0$.
\end{itemize}

\bigskip
{\noindent\bf Goal.} The depth of the recursion tree will be bounded by $q\cdot\log n$ for some fixed constant $q$. The correctness of this claim is proved when we explain how to perform a recursive call. For each recursive call \alg{CVD-APPROX} $(G',w',C,H',{\bf x})$ with the exception of the first call, we aim to prove the following.

\begin{lemma}\label{lemma:approxRecursiveCall}
For any $\delta\in\{1,2,\ldots,q\cdot\log n\}$, each recursive call to \alg{CVD-APPROX}  of depth $\delta \geq 2$ returns a solution that is at least $\opt$ and at most $(\frac{c\cdot\log n}{c\cdot\log n + 9})^{\delta-1}\cdot c\cdot\log(\alpha(G'))\cdot w'({\bf x})$. Moreover, it returns a subset $U\subseteq V(G')$ that realizes the solution.
\end{lemma}

At the initialization phase, we see that in order to prove Lemma \ref{lemma:newApproxSpecial}, it is sufficient to prove Lemma \ref{lemma:approxRecursiveCall}.

\bigskip
{\noindent\bf Initialization.} Initially, the graphs $G'$ and $H'$ are simply set to be the input graphs $G$ and $H$, and the weight function $w'$ is simply set to be input weight function $w$. Moreover, we compute an optimal fractional solution ${\bf x}={\bf x}_{\mathrm{init}}$ by using the ellipsoid method. Recall that the following claim holds.

\begin{observation}
The solution of the instance $(G',w')$ of \wcdel is lower bounded by $w'({\bf x}_{\mathrm{init}})$.
\end{observation} 

% Moreover, it holds that $\alpha(G')\leq n$, and therefore to prove Lemma \ref{lemma:newApproxSpecial}, it is sufficient to return a solution that is at least $\opt$ and at most $c\cdot\log(\alpha(G))\cdot w({\bf x})$ (along with a subset that realizes the solution). Part of the necessity of the stronger claim given by Lemma \ref{lemma:approxRecursiveCall} will become clear at the end of the initialization phase.
Thus, to prove Lemma \ref{lemma:newApproxSpecial}, it is sufficient to return a solution that is at least $\opt$ and at most $c\cdot\log n\cdot w'({\bf x})$.
We would like to proceed by calling our algorithm recursively. For this purpose, we first need to ensure that ${\bf x}$ satisfies the low-value and zero-clique invariants, to which end we use the following notation. We let $h({\bf x})=\{v\in V(G'): {\bf x}(v)\geq 1/(c\cdot \log n)\}$ denote the set of vertices to which ${\bf x}$ assigns high values. Moreover, given a clique $M$ in $G'$, we let $({\bf x}\setminus M): V(G')\rightarrow [0,\infty)$ denote the function that assigns 0 to any vertex in $M$ and $\displaystyle{(1+3\cdot\max_{u\in V(G')}\{{\bf x}(u)\}){\bf x}(v)}$ to any other vertex $v\in V(G')$. Now, to adjust ${\bf x}$ to be of the desired form both at this phase and at later recursive calls, we rely on the two following lemmata.

\begin{lemma}\label{lem:adjustLowVal}
Define $\widehat{G}=G'- h({\bf x})$, $\widehat{w}=w'|_{V(\widehat{G})}$ and $ \widehat{\bf x}={\bf x}|_{V(\widehat{G})}$.
Then, $c\cdot\log n\cdot w'(\widehat{\bf x}) +  w'(h({\bf x})) \leq c\cdot\log n\cdot w'({\bf x})$.
\end{lemma}

\begin{proof}
By the definition of $h({\bf x})$, it holds that $w'(\widehat{\bf x})\leq w'({\bf x})-\frac{1}{c\cdot \log n}\cdot w'(h({\bf x}))$. % Since $\widehat{G}$ is an induced subgraph of $G'$, it also holds that $\alpha(\widehat{G})\leq\alpha(G')$.
Thus, $c\cdot\log n\cdot w'(\widehat{\bf x}) +  w'(h({\bf x})) \leq c\cdot\log n\cdot w'({\bf x})$.
\end{proof}

Thus, it is safe to update $G'$ to $G'- h({\bf x})$, $w'$ to $w'|_{V(\widehat{G})}$, $H'$ to $H'- h({\bf x})$ and ${\bf x}$ to ${\bf x}|_{V(\widehat{G})}$, where we ensure that once we obtain a solution to the new instance, we add $w'(h({\bf x}))$ to this solution and $h({\bf x})$ to the set realizing it. 

\begin{lemma}\label{lem:adjustZeroClique}
Given a clique $M$ in $G'$, the function $({\bf x}\setminus M)$ is a valid fractional solution such that $w'({\bf x}\setminus M)\leq (1+3\cdot\max_{v\in V(G')}\{{\bf x}(v)\})w'({\bf x})$. 
\end{lemma}

\begin{proof}
To prove that $({\bf x}\setminus M)$ is a valid fractional solution, let $Q$ be some chordless cycle in $G'$. We need to show that $({\bf x}\setminus M)(V(Q))\geq 1$. Since $M$ is a clique, $Q$ can contain at most two vertices from $M$. Thus, since $\bf x$  is a valid fractional solution, it holds that ${\bf x}(V(Q)\setminus V(M))\geq 1 - 2\cdot\max_{u\in V(G')}\{{\bf x}(u)\}$. By the definition $({\bf x}\setminus M)$, this fact implies that $({\bf x}\setminus M)(V(Q))=({\bf x}\setminus M)(V(Q)\setminus V(M)) \geq (1+3\cdot\max_{u\in V(G')}\{{\bf x}(u)\})(1-2\cdot\max_{u\in V(G')}\{{\bf x}(u)\}) \geq \min\{(1+\frac{3}{c\cdot\log n})(1-\frac{2}{c\cdot\log n}),1\} = \min\{1 + 1/(c\cdot\log n) - 6/((c\cdot \log n)^2),1\}\geq 1$, where the last inequality relies on the assumption $n\geq 64$.

For the proof of the second part of the claim, note that $w'({\bf x}\setminus M)=(1+3\cdot\max_{v\in V(G')}\{{\bf x}(v)\})$ $w'({\bf x}|_{V(G')\setminus V(M)})\leq (1+3\cdot\max_{v\in V(G')}\{{\bf x}(v)\})w'({\bf x})$.
\end{proof}

Next, it is possible to call \alg{CVD-APPROX} recursively with the fractional solution $({\bf x}\setminus C)$. In the context of the low-value invariant, observe that indeed, for any $v\in V(G')$, it now holds that $({\bf x}\setminus C)(v) = (1+3\cdot\max_{u\in V(G')}\{{\bf x}(u)\}){\bf x}(v) < (1+\frac{3}{c\cdot\log n})\cdot\frac{1}{c\cdot\log n} < (\frac{c\cdot\log n + 9}{c\cdot\log n})^{\delta}\cdot \frac{1}{c\cdot\log n}$ for $\delta=1$.
Similarly, by Lemma \ref{lem:adjustZeroClique}, $w'({\bf x}\setminus C)\leq (\frac{c\cdot\log n + 9}{c\cdot\log n})^{\delta}\cdot w'({\bf x})$ for $\delta=1$. It is also clear that $\alpha(G')\leq n$. Thus, if Lemma \ref{lemma:approxRecursiveCall} is true, we return a solution that is at least $\opt$ and at most $c\cdot\log n\cdot w({\bf x})$ as desired. In other words, to prove Lemma \ref{lemma:newApproxSpecial}, it is sufficient that we next focus only on the proof of Lemma \ref{lemma:approxRecursiveCall}. The proof of this lemma is done by induction. When we consider some recursive call, we assume that the solutions returned by the additional recursive calls that it performs, which are associated with graphs $\widetilde{G}$ such that $\alpha(\widetilde{G})\leq\frac{3}{4}\alpha(G')$, comply with the demands of the lemma.

\bigskip
{\noindent\bf Termination.} Once $G'$ becomes a chordal graph, we return 0 as our solution and $\emptyset$ as the set that realizes it. Clearly, we thus satisfy the demands of Lemma \ref{lemma:approxRecursiveCall}. In fact, we thus also ensure that the execution of our algorithm terminates once $\alpha(G')<24$:

\begin{lemma}\label{lem:approxTerminate}
If $\alpha(G')<24$, then $G'$ is a chordal graph.
\end{lemma}

\begin{proof}
Suppose, by way of contradiction, that $G'$ is not a chordal graph. Then, it contains a chordless cycle $Q$. Since $G'$ is an induced subgraph of $G$, where $G$ is assumed to exclude any chordless cycle on at most 48 vertices, we have that $|V(Q)|>48$. Note that if we traverse $Q$ in some direction, and insert every second vertex on $Q$ into a set, excluding the last vertex in case $|V(Q)|$ is odd, we obtain an independent set. Thus, we have that $\alpha(G)\geq 24$, which is a contradiction.
\end{proof}

Thus, since we will ensure that each recursive calls is associated with a graph whose independence number is at most $3/4$ the independence number of the current graph, we have the following observation.

\begin{observation}\label{obs:depth-chordal}
The maximum depth of the recursion tree is bounded by $q\cdot\log n$ for some fixed constant $q$.
\end{observation}

\bigskip
{\noindent\bf Recursive Call.} Since $H'$ is a chordal graph, it admits a clique forest (Lemma \ref{lem:cliqueForest}). In particular, it contains only $\OO(n)$ maximal cliques, and one can find the set of these maximal cliques in polynomial time \cite{Golumbic80}. By standard arguments on trees, we deduce that $H'$ has a maximal clique $M$ such that after we remove $M$ from $G'$ we obtain two (not necessarily connected) graphs, $\widehat{H}_1$ and $\widehat{H}_2$, such that $\alpha(\widehat{H}_1),\alpha(\widehat{H}_2)\leq\frac{2}{3}\alpha(H')$, and that the clique $M$ can be found in polynomial time. Let $G_1=G'[V(\widehat{H}_1)\cup V(M)\cup V(C)]$, $H_1=H'[V(\widehat{H}_1)\cup V(M)]$, $G_2=G'[V(\widehat{H}_2)\cup V(M)\cup V(C)]$ and $H_2=H'[V(\widehat{H}_2)\cup V(M)]$, and observe that $\alpha(G_1),\alpha(G_2)\leq\frac{2}{3}\alpha(G')+2\leq\frac{3}{4}\alpha(G')$. Here, the last inequality holds because $\alpha(G')\geq 24$, else by Lemma \ref{lem:approxTerminate}, the execution should have already terminated.

\begin{figure}[t]
\centering
\includegraphics[scale=0.6]{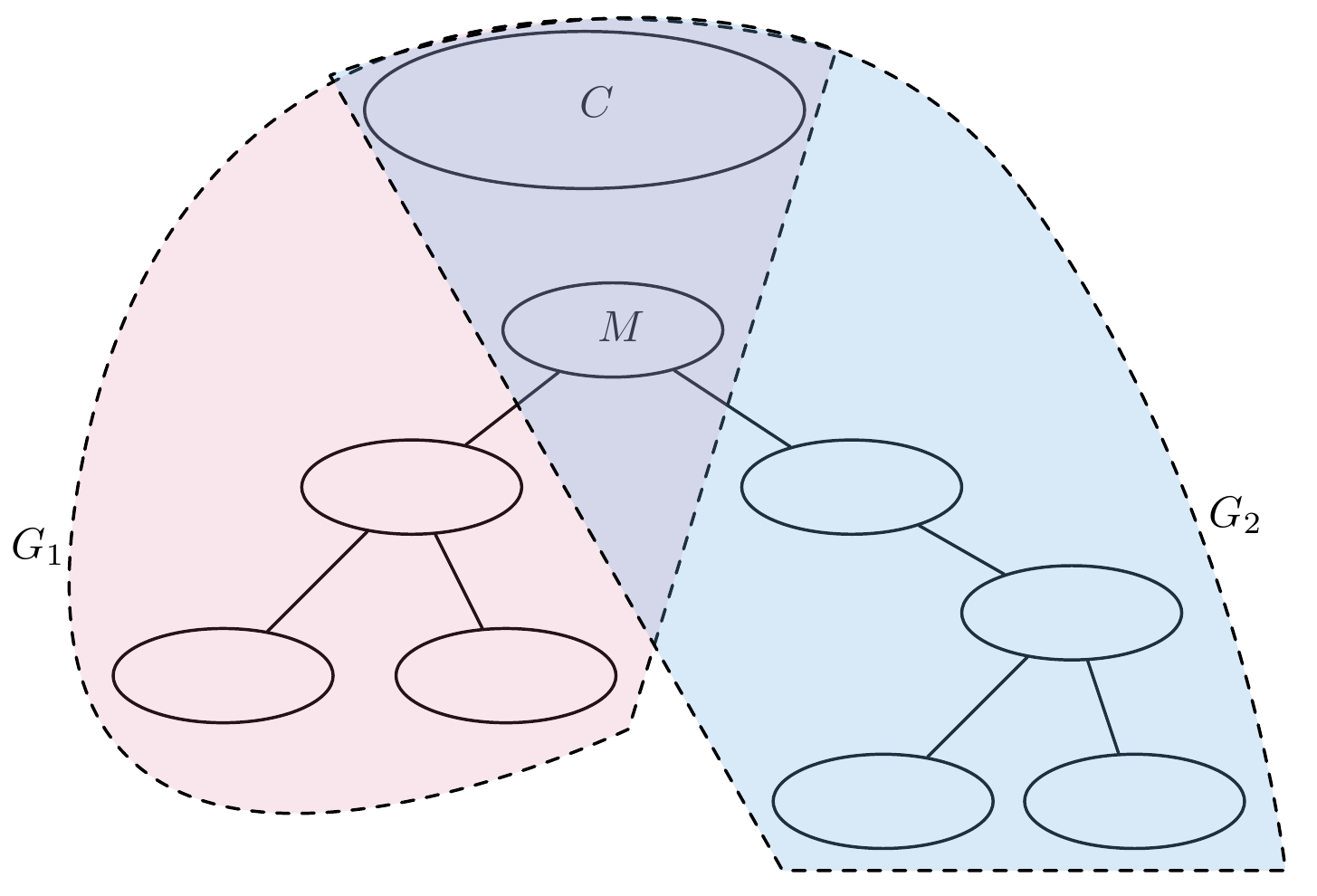}
\caption{Subinstances created by a recursive call}
\label{fig:clique-special-case-recursion}
\end{figure}

We proceed by replacing $\bf x$ by $({\bf x}\setminus M)$. For the sake of clarity, we denote ${\bf x}^*=({\bf x}\setminus M)$. % That is, from now on we abuse notation and use $\bf x$ to refer to $({\bf x}\setminus M)$. % We further adjust the current instance by relying on Lemma \ref{lem:adjustLowVal} so that $\bf x$ satisfies the low-value invariant (in the same manner as it is adjusted in the initialization phase). In particular, we remove $h({\bf x})$ from $G'$, $H'$, $G_1$, $H_1$, $G_2$ and $H_2$, and we let $(G^*,w^*,C,H^*,{\bf x}^*)$, $G_1^*$, $H_1^*$, $G_2^*$ and $H_2^*$ denote the resulting instance and graphs. Observe it still holds that $\alpha(G^*_1),\alpha(G^*_2)\leq\frac{2}{3}\alpha(G')+2$.
By Lemmata \ref{lem:adjustLowVal} and \ref{lem:adjustZeroClique}, to prove Lemma \ref{lemma:approxRecursiveCall}, it is now sufficient to return a solution that is at least $\opt$ and at most $\displaystyle{(1/(1+3\cdot(\frac{c\cdot\log n + 9}{c\cdot\log n})^{\delta}\cdot \frac{1}{c\cdot\log n}))\cdot(\frac{c\cdot\log n}{c\cdot\log n+9})^{\delta-1}\cdot \log \alpha(G') \cdot w({\bf x}^*)}$, along with a set that realizes it. Moreover, for any $v\in V(G')$, it holds that ${\bf x}^*(v) < \displaystyle{(1+3\cdot(\frac{c\cdot\log n + 9}{c\cdot\log n})^{\delta}}$ $\displaystyle{\cdot \frac{1}{c\cdot\log n})\cdot(\frac{c\cdot\log n + 9}{c\cdot\log n})^{\delta}}$ $\displaystyle{\cdot \frac{1}{c\cdot\log n}}$.
Note that by Observation~\ref{obs:depth-chordal}, by setting $c\geq 9q$, we have that $\displaystyle{(\frac{c\cdot\log n + 9}{c\cdot\log n})^{\delta}\leq e < 3}$,
and therefore $\displaystyle{1+3\cdot(\frac{c\cdot\log n + 9}{c\cdot\log n})^{\delta}\cdot \frac{1}{c\cdot\log n}\leq \frac{c\cdot\log n + 9}{c\cdot\log n}}$. In particular, to prove Lemma \ref{lemma:approxRecursiveCall}, it is sufficient to return a solution that is at least $\opt$ and at most $\displaystyle{(\frac{c\cdot\log n}{c\cdot\log n+9})^{\delta}\cdot \log \alpha(G') \cdot w({\bf x}^*)}$.

% Observe that we need only hit every chordless cycle in $G^*$.
Next, we define two subinstances, $I_1=(G_1,w|_{V(G_1)},C,H_1,{\bf x}^*|_{V(G_1)})$ and $I_2=(G_2,w|_{V(G_2)},$ $C,H_2,{\bf x}^*|_{V(G_2)})$ (see Figure~\ref{fig:clique-special-case-recursion}). We solve each of these subinstances by a recursive call to \alg{CVD-APPROX} (by the above discussion, these calls are valid --- we satisfy the low-value and zero-clique invariants). Thus, we obtain two solutions, $s_1$ to $I_1$ and $s_2$ to $I_2$, and two sets that realize these solutions, $S_1$ and $S_2$. By the inductive hypothesis, we have the following observations.

\begin{observation}\label{obs:inductionHit-chordal}
$S_1\cup S_2$ intersects any chordless cycle in $G'$ that lies entirely in either $G_1$ or $G_2$.
\end{observation}

\begin{observation}\label{obs:inductionWeight-chordal}
Given $i\in\{1,2\}$, $s_i\leq (\frac{c\cdot\log n}{c\cdot\log n+9})^{\delta}\cdot c\cdot\log(\alpha(G_i))\cdot w({\bf x}^*_i)$.
\end{observation}

Moreover, since ${\bf x}^*(V(C)\cup V(M))=0$, we also have the following observation.

\begin{observation}\label{obs:inductionSum-chordal}
$w({\bf x}^*_1) + w({\bf x}^*_2) = w({\bf x}^*)$.
\end{observation}

\begin{figure}[t]
\centering
\includegraphics[scale=0.6]{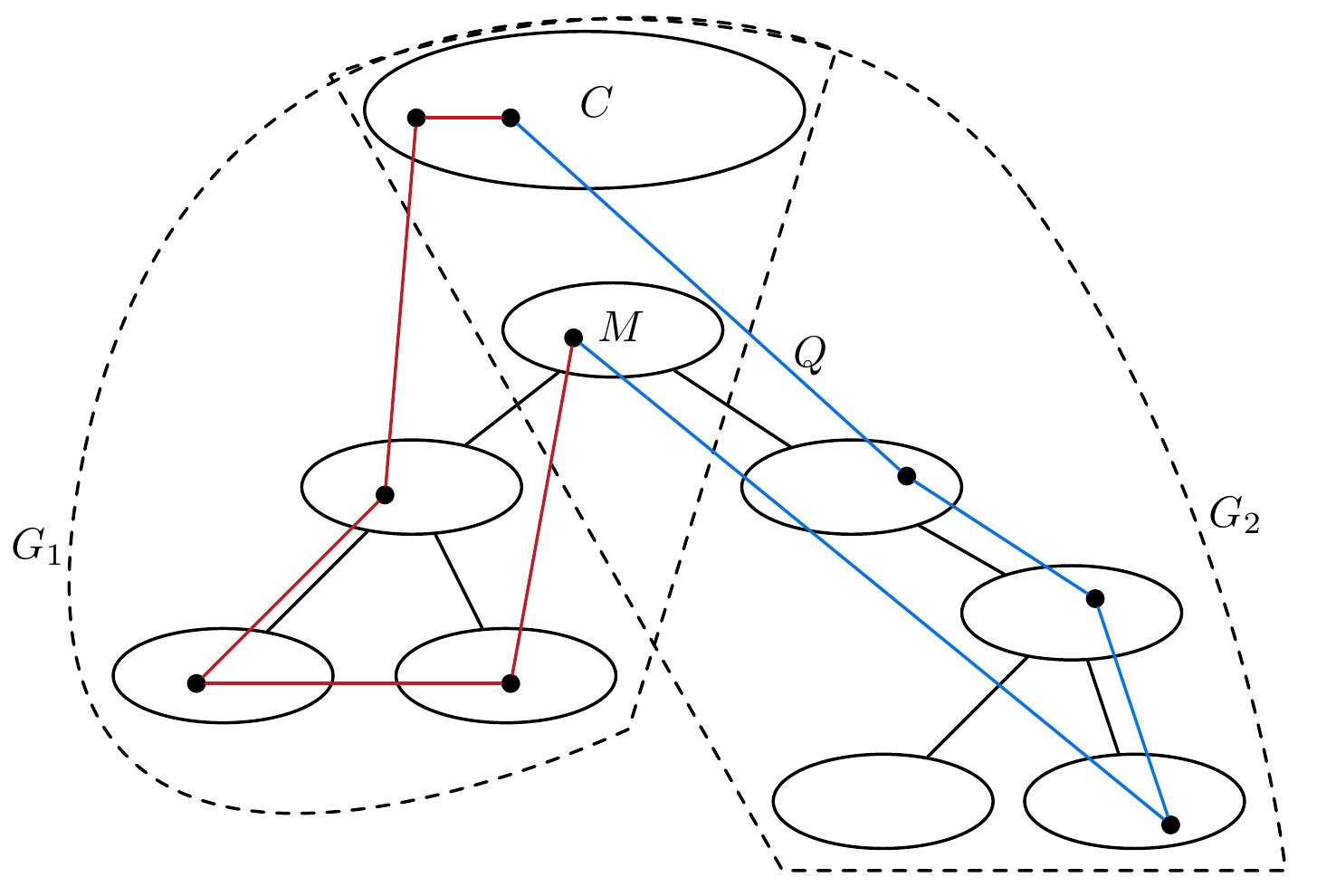}
\caption{An illustration of a bad cycle}
\label{fig:bad-cycle-cvd}
\end{figure}

We say that a cycle of $G'$ is {\em bad} if it is a chordless cycle that belongs entirely to neither $G_1$ nor $G_2$ (see Figure~\ref{fig:bad-cycle-cvd}). Next, we show how to intersect bad cycles.

\bigskip
{\noindent\bf Bad Cycles.} For any pair $(v,u)$ of vertices $v\in V(C)$ and $u\in V(M)$, we let ${\cal P}_1(v,u)$ denote the set of any (simple) path $P_1$ between $v$ and $u$ whose internal vertices belong only to $G_1$ and which does not contain a vertex $v' \in C$ and a vertex $u' \in M$ such that $\{v',u'\} \in E(G')$. Symmetrically, we let ${\cal P}_2(v,u)$ denote the set of any path $P_2$ between $v$ and $u$ whose internal vertices belong only to $G_2$ and which does not contain a vertex $v' \in C$ and a vertex $u' \in M$ such that $\{v',u'\} \in E(G')$. We note here that when $\{v,u\} \in E(G')$ then ${\cal P}_1(v,u)={\cal P}_2(v,u) = \emptyset$.

We first examine the relation between bad cycles and pairs $(v,u)$ of vertices $v\in V(C)$ and $u\in V(M)$.

\begin{lemma}\label{lem:relationBadCyc}
For any bad cycle $Q$ there exist a pair $(v,u)$ of vertices $v\in V(C)$, $u\in V(M)$, a path $P_1\in{\cal P}_1(v,u)$ such that $V(P_1)\subseteq V(Q)$, and a path $P_2\in{\cal P}_2(v,u)$ such that $V(P_2)\subseteq V(Q)$.
\end{lemma}

\begin{proof}
Let $Q$ be some bad cycle. By the definition of a bad cycle, $Q$ must contain at least one vertex $a$ from $H_1\setminus V(M)$ and at least one vertex $b$ from $H_2\setminus V(M)$. Since $C$ and $M$ are cliques, $Q$ can contain at most two vertices from $C$ and at most two vertices from $M$, and if it contains two vertices from $C$ (resp.~$M$), then these two vertices are neighbors. Moreover, since the set $V(C)\cup V(M)$ contains all vertices common to $G_1$ and $G_2$, $Q$ must contain at least one vertex $v\in V(C)$ and at least one vertex $u\in V(M)$ with $\{v,u\} \notin E(G')$. Overall, we conclude that the subpath of $Q$ between $v$ and $u$ that contains $a$ belongs to ${\cal P}_1(v,u)$, while the subpath of $Q$ between $v$ and $u$ that contains $b$ belongs to ${\cal P}_2(v,u)$.
\end{proof}

In light Lemma \ref{lem:relationBadCyc}, to intersect bad cycles, we now examine how the fractional solution ${\bf x}^*$ handles pairs $(v,u)$ of vertices $v\in V(C)$ and $u\in V(M)$.

\begin{lemma}\label{lem:fracCutPair}
For each pair $(v,u)$ of vertices $v\in V(C)$ and $u\in V(M)$ with $\{v,u\} \notin E(G')$, there exists $i\in\{1,2\}$ such that for any path $P\in{\cal P}_i(v,u)$, ${\bf x}^*(V(P))\geq 1/2$.
\end{lemma}

\begin{proof}
Suppose, by way of contradiction, that the lemma is incorrect. Thus, there exist a pair $(v,u)$ of vertices $v\in V(C)$ and $u\in V(M)$ with $\{v,u\} \notin E(G')$, a path $P_1\in{\cal P}_1(v,u)$ such that ${\bf x}^*(V(P_1))<1/2$, and a path $P_2\in{\cal P}_2(v,u)$ such that ${\bf x}^*(V(P_2))<1/2$. Since ${\bf x}^*$ is a valid fractional solution, we deduce that $G'[V(P_1)\cup V(P_2)]$ does not contain any chordless cycle. Consider a shortest subpath $\widehat{P}_1$ of $P_1$ between a vertex $a_1\in V(C)$ and a vertex $b_1\in V(M)$, and a shortest subpath $\widehat{P}_2$ of $P_2$ between a vertex $a_2\in V(C)$ and a vertex $b_2\in V(M)$. Since neither $P_1$ nor $P_2$ contains any edge such that one of its endpoints belongs to $V(C)$ while the other endpoint belongs to $V(M)$, we have that $|V(\widehat{P}_1)|,|V(\widehat{P}_2)|\geq 3$. Furthermore, since vertices common in $P_1$ and $P_2$ must belong to $V(C)\cup V(M)$, we have that $\widehat{P}_1$ does not contain internal vertices that belong to $\widehat{P}_2$ or adjacent to internal vertices on $\widehat{P}_2$. Overall, since $C$ and $M$ are cliques, we deduce that $G'[V(\widehat{P}_1)\cup V(\widehat{P}_2)]$ contains a chordless cycle. To see this, let $a$ be the vertex closest to $b_2$ on $\widehat{P}_2$ that is a neighbor of $a_1$. Observe that $a$ exists as $a_1$ and $a_2$ are neighbors, and $a\neq b_2$. Moreover, we assume without loss of generality that if $a=a_2$, then $a_2$ has no neighbor on $\widehat{P}_1$ apart from $a_1$. Now, let $b$ be the vertex closest to $a$ on the subpath of $\widehat{P}_2$ between $a$ and $b_2$ that is a neighbor of $b_1$. If $b\neq b_2$, then the vertex-sets of $\widehat{P}_1$ and the subpath of $\widehat{P}_1$ between $a$ and $b$ together induce a chordless cycle. Else, let $b'$ be the vertex closest to $a_1$ on $\widehat{P}_1$ that is a neighbor of $b_2$. Then, the vertex-sets of the subpath of $\widehat{P}_1$ between $a_1$ and $b'$ and the subpath of $\widehat{P}_1$ between $a$ and $b_2$ together induce a chordless cycle.
 Since $G'[V(\widehat{P}_1)\cup V(\widehat{P}_2)]$ is an induced subgraph of $G'[V(P_1)\cup V(P_2)]$, we have reached a contradiction.
\end{proof}

Given $i\in\{1,2\}$, let $2{\bf x}^*_i$ denote the fractional solution that assigns to each vertex the value assigned by ${\bf x}^*_i$ times 2. Moreover, let $\widehat{G}_1 = G_1\setminus (V(C)\cup V(M))$ and $\widehat{G}_2 = G_2\setminus (V(C)\cup V(M))$. Observe that $\widehat{G}_1$ and $\widehat{G}_2$ are chordal graphs. Now, for every pair $(v,u)$ such that $v\in V(C), u\in V(M)$, we perform the following operation. We initialize ${\cal T}_1(v,u)=\emptyset$. Next, we consider every pair $(v',u')$ such that $v'\in V(C)$, $u'\in V(M)$, $\{v,v'\}\cap N_{G_1}(u')=\emptyset$ and $\{u,u'\}\cap N_{G_1}(v')=\emptyset$, and insert each pair in $\{(a,b): a\in N_{G_1}(v')\cap V(\widehat{G}_1), b\in N_{G_1}(u')\cap V(\widehat{G}_1), \widehat{G}_1$ has a path between $a$ and $b\}$ into ${\cal T}_1(v,u)$.  We remark that the vertices in a pair in ${\cal T}_1(v,u)$ are not necessarily distinct.
The definition of ${\cal T}_2(v,u)$ is symmetric to the one of ${\cal T}_1(v,u)$.

The following lemma translates Lemma \ref{lem:fracCutPair} into an algorithm.

\begin{lemma}\label{lem:fracCutPairAlg}
For each pair $(v,u)$ of vertices $v\in V(C)$, $u\in V(M)$ and $\{v,u\}\notin E(G')$, one can compute (in polynomial time) an index $i(v,u)\in\{1,2\}$ such that for any path $P\in{\cal P}_i(v,u)$, $2{\bf x}^*_i(V(P))\geq 1$.
\end{lemma}

\begin{proof}
Let $(v,u)$ be a pair of vertices such that $v\in V(C)$, $u\in V(M)$ and $\{v,u\}\notin E(G')$. If there is $i \in\{1,2\}$ such that $P\in{\cal P}_i(v,u) = \emptyset$, then we have trivially obtained the required index which is $i(v,u)=i$. Otherwise, we proceed as follows. For any index $j\in\{1,2\}$, we perform the following procedure. For each pair $(a,b)\in{\cal T}_i(v,u)$, we use Dijkstra's algorithm to compute the minimum weight of a path between $a$ and $b$ in the graph $\widehat{G}_i$ where the weights are given by $2{\bf x}^*_i$. In case for every pair $(a,b)$ the minimum weight is at least 1, we have found the desired index $i(v,u)$. Moreover, by Lemma \ref{lem:fracCutPair} and since for all $v'\in V(C)\cup V(M)$ it holds that ${\bf x}^*_1(v')={\bf x}^*_2(v')=0$, for at least one index $j\in\{1,2\}$, the maximum weight among the minimum weights associated with the pairs $(a,b)$ should be at least 1 (if this value is at least 1 for both indices, we arbitrarily decide to fix $i(v,u)=1$).
\end{proof}

At this point, we need to rely on approximate solutions to {\sc Weighted Multicut} in chordal graphs (in this context, we will employ the algorithm given by Theorem \ref{thm:multicut} in Section \ref{sec:multicut}). Here, a fractional solution ${\bf y}$ is a function ${\bf y}: V(G')\rightarrow [0,\infty)$ such that for every pair $(s_i,t_i)\in{\cal T}$ and any path $P$ between $s_i$ and $t_i$, it holds that ${\bf y}(V(P))\geq 1$. An optimal fractional solution minimizes the weight $w({\bf y})=\sum_{v\in V(G')}w(v)\cdot{\bf y}(v)$. Let \fopt\ denote the weight of an optimal fractional solution.

By first employing the algorithm given by Lemma \ref{lem:fracCutPairAlg}, we next construct two instances of {\sc Weighted Multicut}. The first instance is $J_1=(\widehat{G}_1,w_1,{\cal T}_1)$ and the second instance is $J_2=(\widehat{G}_2,w_2,{\cal T}_2)$, where the sets ${\cal T}_1$ and ${\cal T}_2$ are defined as follows. We initialize ${\cal T}_1=\emptyset$. Now, for every pair $(v,u)$ such that $v\in V(C), u\in V(M)$, $i(v,u)=1$ and ${\cal P}_1(v,u) \neq \emptyset$, we insert each pair in ${\cal T}_1(v,u)$ into ${\cal T}_1$.
The definition of ${\cal T}_2$ is symmetric to the one of ${\cal T}_1$.

By Lemma \ref{lem:fracCutPairAlg} and since for all $v\in V(C)\cup V(M)$ it holds that ${\bf x}^*_1(v)={\bf x}^*_2(v)=0$, we deduce that $2{\bf x}^*_1$ and $2{\bf x}^*_2$ are valid solutions to $J_1$ and $J_2$, respectively. Thus, by calling the algorithm given by Theorem \ref{thm:multicut} with each instance, we obtain a solution $r_1$ to the first instance, along with a set $R_1$ that realizes it, such that $r_1\leq 2d\cdot w({\bf x}^*_1)$, and we also obtain a solution $r_2$ to the second instance, along with a set $R_2$ that realizes it, such that $r_2\leq 2d\cdot w({\bf x}^*_2)$, for some fixed constant $d$. 

By Observation \ref{obs:inductionHit-chordal} and Lemma \ref{lem:relationBadCyc}, we obtained a set $S^*=S_1\cup S_2\cup R_1\cup R_2$ for which we have the following observation.

\begin{observation}
$S^*$ intersects any chordless cycle in $G'$, and it holds that $w(S^*) \leq s_1+s_2+r_1+r_2$.
\end{observation}

Recall that to prove Lemma~\ref{lemma:approxRecursiveCall} we need to show that $s_1+s_2+r_1+r_2\leq (\frac{c\cdot\log n}{c\cdot\log n + 9})^{\delta-1}\cdot c\cdot\log(\alpha(G'))\cdot w'({\bf x})$ and we have $\delta \geq 2$. Furthermore, we have $(\frac{c\cdot\log n}{c\cdot\log n + 9})^{\delta}\cdot c\cdot\log(\alpha(G'))\cdot w'({\bf x}) \leq (\frac{c\cdot\log n}{c\cdot\log n + 9})^{\delta-1}\cdot c\cdot\log(\alpha(G'))\cdot w'({\bf x})$. This together with Lemma~\ref{lem:adjustZeroClique} implies that it is enough to show $s_1+s_2+r_1+r_2\leq (\frac{c\cdot\log n}{c\cdot\log n+9})^{\delta}\cdot c\cdot\log(\alpha(G'))\cdot w({\bf x}^*)$. Recall that for any $i\in\{1,2\}$, $r_i\leq 2d\cdot w({\bf x}^*_i)$. Thus, by Observation \ref{obs:inductionWeight-chordal} and since for any $i\in\{1,2\}$, $\alpha(G_i)\leq\frac{3}{4}\alpha(G')$, we have that
\[w(S^*)\leq (\frac{c\cdot\log n}{c\cdot\log n+9})^{\delta}\cdot c\cdot\log(\frac{3}{4}\alpha(G'))\cdot (w({\bf x}^*_1)+w({\bf x}^*_2)) + 2d\cdot (w({\bf x}^*_1)+w({\bf x}^*_2)).\]

% Recall that $\alpha(G')>(\log n)/2$ and $n\geq 2^{48}$, and therefore $\frac{2}{3}\alpha(G')+2\leq \frac{3}{4}\alpha(G')$.
By Observation \ref{obs:inductionSum-chordal}, we further deduce that 
\[w^*(S^*)\leq \left((\frac{c\cdot\log n}{c\cdot\log n+9})^{\delta}\cdot c\cdot\log(\frac{3}{4}\alpha(G'))+2d\right)\cdot w({\bf x}^*).\]

Now, it only remains to show that $(\frac{c\log n}{c\log n + 9})^{\delta}\cdot c\cdot\log(\frac{3}{4}\alpha(G'))+2d\leq (\frac{c\log n}{c\log n + 9})^{\delta}\cdot c\cdot\log\alpha(G')$, which is equivalent to $2d\leq (\frac{c\log n}{c\log n + 9})^{\delta}\cdot c\cdot\log(\frac{4}{3})$. Recall that $\delta\leq q\cdot\log n$ (Observation \ref{obs:depth-chordal}). % Indeed, it initially holds that $\alpha(G)\leq n$, at each recursive call, the size of the maximum independent set decreases to at most a factor of $3/4$ of its previous value, and the execution terminates once the value of the maximum independent set drops below $(\log n)/2$.
Thus, it is sufficient that we show that $2d\leq (\frac{c\log n}{c\log n + 9})^{q\cdot\log n}\cdot c\cdot\log(\frac{4}{3})$. However, the term $(\frac{c\log n}{c\log n + 9})^{q\cdot\log n}$ is lower bounded by $1/e^{9q}$. In other words, it is sufficient that we fix $c\geq 2\cdot e^{9q}\cdot d\cdot 1/\log(\frac{4}{3})$.

%!TEX root=mainChordal.tex
\subsection{General Graphs}\label{sec:approxGenGraphsDH}

In this subsection we handle general instances by developing a $d\cdot\log^{2}n$-factor approximation algorithm for \wcdel, \alg{Gen-CVD-APPROX}, thus proving the correctness of Theorem \ref{thm:newApprox2}. The exact value of the constant $d\geq \max\{96,2c\}$ is determined later.\footnote{Recall that $c$ is the constant we fixed to ensure that the approximation ratio of \alg{CVD-APPROX} is bounded by $c\cdot\log n$.} This algorithm is based on recursion, and during its execution, we often encounter instances that are of the form of the Clique+Chordal special case of \wcdel, which will be dealt with using the algorithm \alg{CVD-APPROX} of Section \ref{sec:approxCliqueChordal}.

\bigskip
{\noindent\bf  Recursion.} We define each call to our algorithm \alg{Gen-CVD-APPROX} to be of the form $(G',w')$, where $(G',w')$ is an instance of \wcdel\ such that $G'$ is an induced subgraph of $G$, and we denote $n'=|V(G')|$. We ensure that after the initialization phase, the graph $G'$ never contains chordless cycles on at most 48 vertices. We call this invariant the {\em $C_{48}$-free invariant}. In particular, this guarantee ensures that the graph $G'$ always contains only a small number of maximal cliques:

\begin{lemma}[\cite{C4FreeNumCliques,TsukiyamaIAS77}]\label{lem:numMaxCliques}
The number of maximal cliques of a graph $G'$ that has no chordless cycles on four vertices is bounded by $\OO({n'}^2)$, and they can be enumerated in polynomial time using a polynomial delay algorithm. 
\end{lemma}

\bigskip
{\noindent\bf Goal.} For each recursive call \alg{Gen-CVD-APPROX}$(G',w')$, we aim to prove the following.

\begin{lemma}\label{lemma:approxRecursiveCallGen}
\alg{Gen-CVD-APPROX} returns a solution that is at least $\opt$ and at most $\frac{d}{2}\cdot\log^2n'\cdot\opt$. Moreover, it returns a subset $U\subseteq V(G')$ that realizes the solution.
\end{lemma}

At each recursive call, the size of the graph $G'$ becomes smaller. Thus, when we prove that Lemma \ref{lemma:approxRecursiveCallGen} is true for the current call, we assume that the approximation factor is bounded by $\frac{d}{2}\cdot\log^2\widehat{n}\cdot\opt$ for any call where the size $\widehat{n}$ of the vertex-set of its graph is strictly smaller than~$n'$.

\bigskip
{\noindent\bf Initialization.} Initially, we set $(G',w')=(G,w)$. However, we need to ensure that the $C_{48}$-free invariant is satisfied. For this purpose, we update $G'$ as follows. First, we let ${\cal C}_{48}$ denote the set of all chordless cycles on at most 48 vertices of $G'$. Clearly, ${\cal C}_{48}$ can be computed in polynomial time and it holds that $|{\cal C}_{48}|\leq n^{48}$. Now, we construct an instance of {\sc Weighted 48-Hitting Set}, where the universe is $V(G')$, the family of 48-sets is ${\cal C}_{48}$, and the weight function is $w'$. Since each chordless cycle must be intersected, it is clear that the optimal solution to our {\sc Weighted 48-Hitting Set} instance is at most $\opt$. By using the standard $c'$-approximation algorithm for {\sc Weighted $c'$-Hitting Set} \cite{KleinbergT05}, which is suitable for any fixed constant $c'$, we obtain a set $S\subseteq V(G')$ that intersects all cycles in ${\cal C}_{48}$ and whose weight is at most $48\cdot\opt$.  Having the set $S$, we remove its vertices from $G'$. Now, the $C_{48}$-free invariant is satisfied, which implies that we can recursively call our algorithm. To the outputted solution, we add $w(S)$ and $S$. If Lemma \ref{lemma:approxRecursiveCallGen} is true, we obtain a solution that is at most $\frac{d}{2}\cdot\log^2n\cdot\opt + 48\cdot\opt\leq d\cdot\log^2n\cdot\opt$, which allows us to conclude the correctness of Theorem \ref{thm:newApprox2}. We remark that during the execution of our algorithm, we only update $G'$ by removing vertices from it, and thus it will always be safe to assume that the $C_{48}$-free invariant is satisfied.

\bigskip
{\noindent\bf Termination.} Observe that due to Lemma \ref{lem:numMaxCliques}, we can test in polynomial time whether $G'$consists of a clique and a chordal graph: we examine each maximal clique of $G'$, and check whether after its removal we obtain a chordal graph. Once $G'$ becomes such a graph that consists of a chordal graph and a clique, we solve the instance $(G',w')$ by calling algorithm \alg{CVD-APPROX}. Since $c\cdot\log n'\leq\frac{d}{2}\cdot\log^2n'$, we thus ensure that at the base case of our induction, Lemma \ref{lemma:approxRecursiveCallGen} holds.

\bigskip
{\noindent\bf Recursive Call.} For the analysis of a recursive call, let $S^*$ denote a hypothetical set that realizes the optimal solution $\opt$ of the current instance $(G',w')$. Moreover, let $(F,\beta)$ be a clique forest of $G'- S^*$, whose existence is guaranteed by Lemma \ref{lem:cliqueForest}. Using standard arguments on forests, we have the following observation.

%\begin{observation}\label{obs:standardForest}
%There exists a node $v\in V(F)$ such that $\beta(v)$ is a balanced separator for $G'- S^*$.
%\end{observation} 

%This observation directly leads us to a more concrete one:

\begin{observation}\label{obs:divideClique}
There exist a maximal clique $M$ of $G'$ and a subset $S\subseteq V(G')\setminus M$ of weight at most $\opt$ such that $M\cup S$ is a balanced separator for $G'$.
\end{observation}

The following lemma translates this observation into an algorithm.

\begin{lemma}\label{lem:divideClique}
There is a polynomial-time algorithm that finds a maximal clique $M$ of $G'$ and a subset $S\subseteq V(G')\setminus M$ of weight at most $q\cdot\log n'\cdot\opt$ for some fixed constant $q$ such that $M\cup S$ is a balanced separator for $G'$.
\end{lemma}

\begin{proof}
We examine every maximal clique of $G'$. By Lemma \ref{lem:numMaxCliques}, we need only consider $\OO({n'}^2)$ maximal cliques, and these cliques can be enumerated in polynomial time. For each such clique $M$, we run the $q\cdot \log n'$-factor approximation algorithm by Leighton and Rao \cite{BalancedSeparator} to find a balanced separator $S_M$ of $G'- M$. Here, $q$ is some fixed constant. We let $S$ denote some set of minimum weight among the sets in $\{S_M: M$ is a maximal clique of $G'\}$. By Observation~\ref{obs:divideClique}, $w(S)\leq q\cdot\log n'\cdot\opt$. Thus, the desired output is a pair $(M,S)$ where $M$ is one of the examined maximal cliques such that $S_M=S$.
\end{proof}

We call the algorithm in Lemma \ref{lem:divideClique} to obtain a pair $(M,S)$. Since $M\cup S$ is a balanced separator for $G'$, we can partition the set of connected components of $G'- (M\cup S)$ into two sets, ${\cal A}_1$ and ${\cal A}_2$, such that for $V_1=\bigcup_{A\in{\cal A}_1}V(A)$ and $V_2=\bigcup_{A\in{\cal A}_2}V(A)$ it holds that $n_1,n_2\leq\frac{2}{3}n'$ where $n_1=|V_1|$ and $n_2=|V_2|$. We remark that we used the $\OO(\log n)$-factor approximation algorithm by Leighton and Rao \cite{BalancedSeparator} in Lemma~\ref{lem:divideClique} to find the balanced separator instead of the $\OO(\sqrt{\log n})$-factor approximation algorithm by Feige et al.~\cite{FeigeHL08}, as the algorithm by Feige et al. is randomized. 

Next, we define two inputs of (the general case of) \wcdel: $I_1=(G'[V_1],w'|_{V_1})$ and $I_2=(G'[V_2],w'|_{V_2})$. Let $\opt_1$ and $\opt_2$ denote the optimal solutions to $I_1$ and $I_2$, respectively. Observe that since $V_1\cap V_2=\emptyset$, it holds that $\opt_1+\opt_2\leq\opt$. We solve each of the subinstances by recursively calling algorithm \alg{Gen-CVD-APPROX}. By the inductive hypothesis, we thus obtain two sets, $S_1$ and $S_2$, such that $G'[V_1]- S_1$ and $G'[V_2]- S_2$ are chordal graphs, and $w'(S_1)\leq \frac{d}{2}\cdot\log^2n_1\cdot\opt_1$ and $w'(S_2)\leq \frac{d}{2}\cdot\log^2n_2\cdot\opt_2$.

We proceed by defining an input of the Clique+Chordal special case of \wcdel: $J=(G'[(V_1\cup V_2\cup M)\setminus (S_1\cup S_2)],w'|_{(V_1\cup V_2\cup M)\setminus (S_1\cup S_2)})$. Observe that since $G'[V_1]- S_1$ and $G'[V_2]- S_2$ are chordal graphs and $M$ is a clique, this is indeed an instance of the Clique+Chordal special case of \wcdel.  We solve this instance by calling algorithm \alg{CVD-APPROX}. We thus obtain a set, $\widehat{S}$, such that $G'[(V_1\cup V_2\cup M)- (S_1\cup S_2\cup\widehat{S})]$ is a chordal graphs, and $w'(\widehat{S})\leq c\cdot\log n'\cdot\opt$ (since $|(V_1\cup V_2\cup M)\setminus (S_1\cup S_2)|\leq n'$ and the optimal solution of each of the subinstances is at most $\opt$).

Observe that since $M$ is a clique and there is no edge in $E(G')$ between a vertex in $V_1$ and a vertex in $V_2$, any chordless cycle of $G'- (S\cup S_1\cup S_2)$ entirely belongs to either $G'[(V_1\cup M)\setminus S_1]$ or $G'[(V_2\cup M)\setminus S_2]$. This observation, along with the fact that $G'[(V_1\cup V_2\cup M)\setminus (S_1\cup S_2\cup\widehat{S})]$ is a chordal graphs, implies that $G'- T$ is a chordal graphs where $T=S\cup S_1\cup S_2\cup\widehat{S}$. Thus, it is now sufficient to show that $w'(T)\leq\frac{d}{2}\cdot\log^2n'\cdot\opt$.

By the discussion above, we have that

\[\begin{array}{ll}
w'(T) & \leq w'(S) + w'(S_1) + w'(S_2) + w'(\widehat{S})\\
& \leq q\cdot\log n'\cdot\opt + \frac{d}{2}\cdot(\log^2n_1\cdot\opt_1 + \log^2n_2\cdot\opt_2) + c\cdot\log n'\cdot\opt.
\end{array}\]

Recall that $n_1,n_2\leq\frac{2}{3}n'$ and $\opt_1+\opt_2\leq\opt$. Thus, we have that

\[\begin{array}{ll}
w'(T) & \leq q\cdot\log n'\cdot\opt + \frac{d}{2}\cdot(\log^2\frac{2}{3}n')\cdot\opt + c\cdot\log n'\cdot\opt\\
& \leq \frac{d}{2}\cdot\log^2n'\cdot\opt + (q + c - \frac{d}{2}\log\frac{3}{2})\cdot\log n'\cdot\opt.
\end{array}\]

Overall, we conclude that to ensure that $w'(T)\leq\frac{d}{2}\cdot\log^2n'\cdot\opt$, it is sufficient to ensure that $q + c - \frac{d}{2}\log\frac{3}{2}\leq 0$, which can be done by fixing $\displaystyle{d=\frac{2}{\log\frac{3}{2}}\cdot(q+c)}$.

%!TEX root=mainChordal.tex

\section{Weighted Multicut in Chordal Graphs}\label{sec:multicut}
In this section we prove Theorem~\ref{thm:multicutGen}. 
%the following theorem.
%\begin{theorem}\label{thm:multicutGen}
%{\sc Weighted Multicut} in chordal graphs admits a constant factor approximation algorithm.
%\end{theorem}
Let us denote $c=8$. Recall that for {\sc Weighted Multicut}, a fractional solution ${\bf x}$ is a function ${\bf x}: V(G)\rightarrow [0,\infty)$ such that for every pair $(s,t)\in{\cal T}$ and any path $P$ between $s$ and $t$, it holds that ${\bf x}(V(P))\geq 1$. An optimal fractional solution minimizes the weight $w({\bf x})=\sum_{v\in V(G)}w(v)\cdot{\bf x}(v)$. Let \fopt\ denote the weight of an optimal fractional solution. Theorem \ref{thm:multicutGen} follows from the next result, whose proof is the focus of this section.

\begin{lemma}\label{thm:multicut1}
Given an instance of {\sc Weighted Multicut} in chordal graphs, one can find (in polynomial time) a solution that is at least $\opt$ and at most $4c\cdot \fopt$, along with a set that realizes~it.
\end{lemma}

\bigskip
{\noindent\bf Preprocessing.} By using the ellipsoid method, we may next assume that we have optimal fractional solution ${\bf x}$ at hand. We say that ${\bf x}$ is {\em nice} if for all $v\in V(G)$, there exists $i\in\{0\}\cup\mathbb{N}$ such that ${\bf x}(v)=\frac{i}{n}$. Let $h({\bf x})=\{v\in V(G): {\bf x}(v)\geq 1/c\}$ denote the set of vertices to which ${\bf x}$ assigns high values. 

\begin{lemma}
Define a function $\widehat{\bf x}: V(G)\rightarrow [0,\infty)$ as follows. For all $v\in V(G)$, if ${\bf x}(v)<1/2n$ then $\widehat{\bf x}(v)=0$, and otherwise $\widehat{\bf x}(v)$ is the smallest value of the form $i/n$, for some $i\in\mathbb{N}$, that is at least $2{\bf x}(v)$. Then, $\widehat{\bf x}$ is a fractional solution such that $w(\widehat{\bf x})\leq 4w({\bf x})$.
\end{lemma}

\begin{proof}
To show that $\widehat{\bf x}$ is a fractional solution, consider some path $P$ between $s$ and $t$ such that $(s,t)\in{\cal T}$. Let $\ell'({\bf x})=\{v\in V(G): {\bf x}(v)<1/2n\}$. We have that $\widehat{\bf x}(V(P)) = \sum_{v\in V(P)\setminus\ell'({\bf x})}\widehat{\bf x}(v) \geq 2\sum_{v\in V(P)\setminus\ell'({\bf x})}{\bf x}(v)$. Thus, to show that $\widehat{\bf x}(V(P))\geq 1$, it is sufficient to show that $\frac{1}{2}\leq \sum_{v\in V(P)\setminus\ell'({\bf x})}{\bf x}(v)$. Since ${\bf x}$ is a fractional solution, it holds that ${\bf x}(V(P)) = \sum_{v\in V(P)\cap\ell'({\bf x})}{\bf x}(v) + \sum_{v\in V(P)\setminus\ell'({\bf x})}{\bf x}(v)\geq 1$. Thus, $1\leq \frac{1}{2n}|V(P)\cap\ell'({\bf x})| + \sum_{v\in V(P)\setminus\ell'({\bf x})}{\bf x}(v)$. Since $|V(P)\cap\ell'({\bf x})|\leq n$, we conclude that $\frac{1}{2}\leq \sum_{v\in V(P)\setminus\ell'({\bf x})}{\bf x}(v)$.

The second part of the claim follows from the observation that for all $v\in V(G)$, $\widehat{\bf x}(v)\leq 4{\bf x}(v)$.
\end{proof}

Accordingly, we update ${\bf x}$ to $\widehat{\bf x}$. Our preprocessing step also relies on the following standard lemma.

\begin{lemma}\label{lem:adjustLowVal2}
Define $\widehat{G}=G- h({\bf x})$, $\widehat{w}=w|_{V(\widehat{G})}$ and $\widehat{\bf x}={\bf x}|_{V(\widehat{G})}$.
Then, $c\cdot w(\widehat{\bf x}) +  w(h({\bf x})) \leq c\cdot w({\bf x})$.
\end{lemma}

\begin{proof}
By the definition of $h({\bf x})$, it holds that $w(\widehat{\bf x})\leq w({\bf x})-\frac{1}{c} w(h({\bf x}))$. 
Thus, $c\cdot w(\widehat{\bf x}) +  w(h({\bf x})) \leq c\cdot (w({\bf x})-\frac{1}{c}\cdot w(h({\bf x}))) +  w(h({\bf x})) = c\cdot w({\bf x})$.
\end{proof}

We thus further update $G$ to $\widehat{G}$, $w$ to $\widehat{w}$ and ${\bf x}$ to $\widehat{\bf x}$, where we ensure that once we obtain a solution to the new instance, we add $w(h({\bf x}))$ to this solution and $h({\bf x})$ to the set realizing it. Overall, we may next focus only on the proof of the following lemma.

\begin{lemma}\label{lemma:mainMulticut}
Let $(G,w)$ be an instance of {\sc Weighted Multicut} in chordal graphs, and ${\bf x}$ be a nice fractional solution such that $h({\bf x})=\emptyset$. Then, one can find (in polynomial time) a solution that is at least $\opt$ and at most $c\cdot w({\bf x})$, along with a set that realizes it.
\end{lemma}

\bigskip
{\noindent\bf The Algorithm.} Since $G$ is a chordal graph, we can first construct in polynomial time a clique forest $(F,\beta)$ of $G$ (Lemma \ref{lem:cliqueForest}). Without loss of generality, we may assume that $F$ is a tree, else $G$ is not a connected graph and we can handle each of its connected components separately. Now, we arbitrarily root $F$ at some node $r_F$, and we arbitrarily choose a vertex $r_G\in\beta(r_F)$. We then use Dijkstra's algorithm to compute (in polynomial time) for each vertex $v\in V(G)$, the value $\displaystyle{d(v)=\min_{P\in{\cal P}(v)}{\bf x}(V(P))}$, where ${\cal P}(v)$ is the set of paths in $G$ between $r_G$ and $v$.

We define $n+1$ bins: for all $i\in\{0,1,\ldots,n\}$, the bin $B_i$ contains every vertex $v\in V(G)$ for which there exists $j\in\{0\}\cup\mathbb{N}$ such that $d(v)-{\bf x}(v) < (\frac{i}{n}+2j)\frac{1}{c}\leq d(v)$ (i.e., $0\leq d(v)-(\frac{i}{n}+2j)\frac{1}{c} < {\bf x}(v)$). 
Let $B_{i^*}$, $i^*\in\{0,1,\ldots,n\}$, be a bin that minimizes $w(B_{i^*})$. The output consists of $w(B_{i^*})$ and $B_{i^*}$.

\bigskip
{\noindent\bf Approximation Factor.} Given $r\in[0,1]$, let $\widehat{B}_r$ be the set that contains every vertex $v\in V(G)$ for which there exists $j\in\{0\}\cup\mathbb{N}$ such that $0\leq d(v)-(r+2j)\frac{1}{c}< {\bf x}(v)$. We start with the following claim.

\begin{lemma}
There exists $r^* \in[0,1]$ such that $w(\widehat{B}_{r^*})\leq c\cdot w({\bf x})$.
\end{lemma}

\begin{proof}
For any $d\geq 0$, observe that there exists exactly one $j\in\{0\}\cup\mathbb{N}$ for which there exists $r\in[0,1]$ such that $0\leq d-(r+2j)\frac{1}{c} < \frac{1}{c}$, and denote it by $j(d)$.
Suppose that we choose $r \in[0,1]$ uniformly at random.
Consider some vertex $v\in V(G)$. Then, since $h({\bf x})=\emptyset$, the probability that there exists $j\in\{0\}\cup\mathbb{N}$ such that $0\leq d(v)-(r+2j)\frac{1}{c}< {\bf x}(v)$ is equal to the probability that $0\leq d(v)-(r+2j(d(v)))\frac{1}{c} < {\bf x}(v)$. Now, the probability that $0\leq d(v)-(r+2j(d(v)))\frac{1}{c} < {\bf x}(v)$ is equal to $c\cdot{\bf x}(v)$. The expected weight $w(\widehat{B}_r)$ is $c\cdot\sum_{v\in V(G)}{\bf x}(v)\cdot w(v) =  c\cdot w({\bf x})$. Thus, there exists $r^* \in[0,1]$ such that $w(\widehat{B}_{r^*})\leq c\cdot w({\bf x})$.
\end{proof}

Now, the proof of the approximation factor follows from the next claim.

\begin{lemma}
There exists $i\in\{0,1,\ldots,n\}$ such that $B_i\subseteq\widehat{B}_{r^*}$.
\end{lemma}

\begin{proof}
Let $i$ be the smallest index in $\{0,1,\ldots,n\}$ such that $r^*\leq\frac{i}{n}$.
Consider some vertex $v\in B_i$. Then, for some $j\in\{0\}\cup\mathbb{N}$, $d(v)-{\bf x}(v) < (\frac{i}{n}+2j)\frac{1}{c}\leq d(v)$. Since $r^*\leq\frac{i}{n}$, we have that $(r^*+2j)\frac{1}{c}\leq d(v)$. Since $\bf x$ is nice, it holds that there exists $t\in\{0\}\cup\mathbb{N}$ such that $d(v)-{\bf x}(v)=\frac{t}{n}$. Thus, for any $p<\frac{1}{n}$, it holds that $d(v)-{\bf x}(v) < (\frac{i}{n}+2j-p)\frac{1}{c}$. By the choice of $i$, $\frac{i}{n} - r^* < \frac{1}{n}$, and therefore $d(v)-{\bf x}(v) < (r^*+2j)\frac{1}{c}$, which implies that $v\in\widehat{B}_{r^*}$.
\end{proof}

\bigskip
{\noindent\bf Feasibility.} We need to prove that for any pair $(s,t)\in{\cal T}$, $G- B_{i^*}$ does not have any path between $s$ and $t$. Consider some path $P=(v_1,v_2,\cdots,v_\ell)$ between $s$ and $t$. Here, $v_1=s$ and $v_\ell=t$. Suppose, by way of contradiction, that $V(P)\cap B_{i^*}=\emptyset$. Then, for all $v_i\in V(P)$, it holds that there is no $j\in\{0\}\cup\mathbb{N}$ such that $0\leq d(v_i)-(\frac{i^*}{n}+2j)\frac{1}{c} < {\bf x}(v_i)$.

Let $s\in V(F)$ be the closest node to $r_F$ that satisfies $\beta(s)\cap V(P)\neq\emptyset$ (since $F$ is a clique tree and $P$ is a path, the node $s$ is uniquely defined).
Let $v_{\widehat{i}}$ be some vertex in $\beta(s)\cap V(P)\neq\emptyset$. For the sake of clarity, let us denote the subpath of $P$ between $v_{\widehat{i}}$ and $v_\ell$ by $Q=(u_1,u_2,\cdots,u_t)$, where $u_1=v_{\widehat{i}}$ and $u_t=v_\ell$. Let $j^*$ be the smallest value in $\{0\}\cup\mathbb{N}$ that satisfies $d(u_1)-{\bf x}(u_1)<(\frac{i^*}{n}+2j^*)\frac{1}{c}$. Note that $d(u_1)<(\frac{i^*}{n}+2j^*)\frac{1}{c}$. It is thus well defined to let $p$ denote the largest index in $[t]$ such that $d(u_p)<(\frac{i^*}{n}+2j^*)\frac{1}{c}$.

First, suppose that $p\in[t-1]$. We then have that $(\frac{i^*}{n}+2j^*)\frac{1}{c}\leq d(u_{p+1})$. For all $2\leq i\leq t$, it holds that $d(u_i)\leq d(u_{i-1}) + {\bf x}(u_i)$. We thus obtain that $d(u_{p+1})-{\bf x}(u_{p+1})\leq d(u_{p})<(\frac{i^*}{n}+2j^*)\frac{1}{c}$. This statement implies that $u_{p+1}\in B_{i^*}$, which is a contradiction.

Now, we suppose that $p=t$. Note that $(\frac{i^*}{n}+2j^*-2)\frac{1}{c}\leq d(u_1)-{\bf x}(u_1)$ (by the minimality of $j^*$), and $d(u_t)<(\frac{i^*}{n}+2j^*)\frac{1}{c}$. We get that $d(u_t) < d(u_1)-{\bf x}(u_1) + \frac{2}{c}$. In other words, $d(u_t)-d(u_1)+{\bf x}(u_1) < \frac{2}{c}$. Let des$(s)$ denote the set consisting of $s$ and its descendants in $F$. Since  $F$ is a clique tree, we have that $V(Q)\subseteq \bigcup_{s'\in\mathrm{des}(s)}\beta(s')$. Thus, any path from $r_G$ to $u_t$ that realizes $d(u_t)$ contains a vertex from $\beta(s)$. Since there exists a path from $r_G$ to $u_t$ that realizes $d(u_t)$, we deduce that there exists a path, $P_t$, from $r_G$ to $u_t$ that realizes $d(u_t)$ and contains a vertex $x\in N_G[u_1]$. Let $P^*_t$ denote the subpath of $P_t$ between $x$ and $u_t$, and let $P^*$ denote the path that starts at $u_1$ and then traverses $P^*_t$. Then, ${\bf x}(V(P^*)) \leq {\bf x}(u_1) + {\bf x}(V(P^*_t)) = {\bf x}(u_1) + d(u_t)-d(x)+{\bf x}(x)$.
Note that $d(u_1)\leq d(x)+{\bf x}(u_1)$, and therefore ${\bf x}(V(P^*))\leq {\bf x}(u_1) + d(u_t) - (d(u_1) - {\bf x}(u_1)) +{\bf x}(x) = {\bf x}(u_1) +{\bf x}(x) + (d(u_t)-d(u_1)+{\bf x}(u_1))$. Since $h({\bf x})=\emptyset$ and $d(u_t)-d(u_1)+{\bf x}(u_1) < \frac{2}{c}$, we get that ${\bf x}(V(P^*))<\frac{4}{c}$. The symmetric analysis of the subpath of $P$ between $u_1=v_{\widehat{i}}$ and $v_1$ shows that there exists a path $P^{**}$ between $u_1$ and $v_1$ such that ${\bf x}(V(P^{**}))<\frac{4}{c}$. Overall, we get that there exists a path, $P'$, between $v_1=s$ and $v_\ell=u_\ell=t$ such that ${\bf x}(V(P'))<\frac{8}{c}$. Since $c \geq 8$, we reach a contradiction to the assumption that ${\bf x}$ is a fractional solution.

%
%!TEX root = mainChordal.tex

\section{Distance-Hereditary Vertex Deletion}\label{sec:dist}
In this section we prove Theorem~\ref{thm:newApprox2DH}. We start with preliminaries. 

%\begin{theorem}\label{thm:newApprox2DH}
%\wDHfull admits a $\approxDH$-factor approximation algorithm.
%\end{theorem}

\paragraph{Preliminaries.} A graph $G$ is \emph{distance hereditary} if every connected induced subgraph $H$ of $G$, for all $u,v \in V(H)$ the number of vertices in shortest path between $u$ and $v$ in $G$ is same as the number of vertices in shortest path between $u$ and $v$ in $H$. Another characterization of distance hereditary graphs is the graph not containing an induced sub-graph isomorphic to a house, a gem, a domino or an induced cycle on $5$ or more vertices (refer Figure~\ref{fig:dh-obs}). We refer to a house, a gem, a domino or an induced cycle on at least $5$ vertices as a \dho. A \dho\ on at most $\numdhd$ vertices is a small \dho. A \biclique\ is a graph $G$ with vertex bipartition $X,Y$ each of them being non-empty such that for each $x\in X$ and $y \in Y$ we have $\{x,y\} \in E(G)$. We note here that, $X$ and $Y$ need not be independent sets in a \biclique\ $G$.

\begin{figure}[t]
\centering
\includegraphics[scale=0.6]{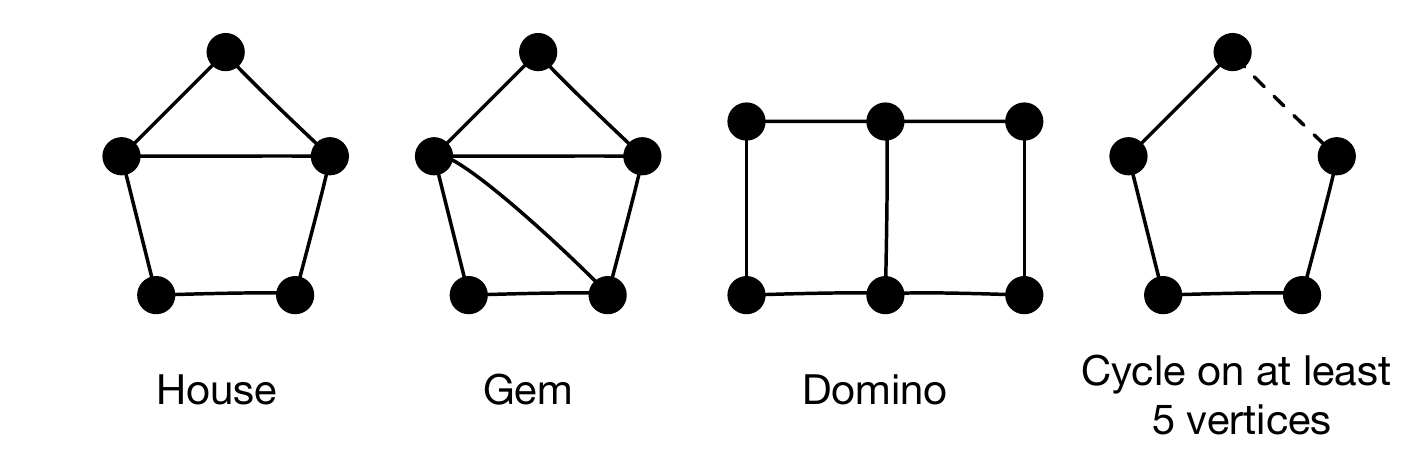}
\caption{Obstruction set for distance hereditary graphs}
\label{fig:dh-obs}
\end{figure}

Clearly, we can assume that the weight $w(v)$ of each vertex $v\in V(G)$ is positive, else we can insert $v$ into any solution. Our approximation algorithm for \wDH\ comprises of two components. The first component handles the special case where the input graph $G$ consists of a biclique $C$ and a distance hereditary $H$. Here, we also assume that the input graph has no ``small'' \dho. We show that when input restricted to these special instances \wDH\ admits an $\mathcal{O}(\log ^2 n)$-factor approximation algorithm. 

The second component is a recursive algorithm that solves general instances of the problem. Initially, it easily handles ``small'' \dho. Then, it gradually disintegrates a general instance until it becomes an instance of the special form that can be solved in polynomial time. More precisely, given a problem instance, the algorithm divides it by finding a maximal biclique $M$ (using an exhaustive search which relies on the guarantee that $G$ has no ``small'' \dho) and a small separator $S$ (using an approximation algorithm) that together break the input graph into two graphs significantly smaller than their origin. %It first removes $M \cup S$ and solves each of the two resulting subinstances by calling itself recursively; then, it inserts $M$ back into the graph, and uses the solutions it obtained from the recursive calls to construct an instance of the special case solved by the first component. 

\subsection{Biclique+ Distance Hereditary Graph}\label{sec:approxBicliqueHereditary}
%!TEX root = mainChordal.tex

In this subsection we handle the special case where the input graph $G$ consists of a biclique $C$ and a distance hereditary graph $H$.
More precisely, along with the input graph $G$ and the weight function $w$, we are also given a biclique $C$ and a distance hereditary graph $H$ such that $V(G)=V(C)\cup V(H)$, where the vertex-sets $V(C)$ and $V(H)$ are disjoint. 
Here, we also assume that $G$ has no \dho\ on at most $\numdhd$ vertices, which means that every \dho in $G$ is a chordless cycle of strictly more than $\numdhd$ vertices.
Note that the edge-set $E(G)$ may contain edges between vertices in $C$ and vertices in $H$. 
We call this special case the {\em Biclique + Distance Hereditary special case}. Our objective is to prove the following result.

\begin{lemma}\label{lemma:DHnewApproxSpecial}
    The Biclique + Distance Hereditary special case of \wDH admits an $\OO(\log^2 n)$-factor approximation algorithm.
\end{lemma}

We assume that $n\geq 2^{12}$, else the input instance can be solve by brute-force \footnote{This assumption simplifies some of the calculations ahead.}. Let $c$ be a fixed constant (to be determined later). In the rest of this subsection, we design a $c\cdot\log n$-factor approximation algorithm for the Biclique + Distance Hereditary special case of \wDH.

\bigskip
{\noindent\bf  Recursion.} Our approximation algorithm is a recursive algorithm. We call our algorithm \alg{DHD-APPROX}, and define each call to be of the form $(G',w',C,H',{\bf x})$. Here, $G'$ is an induced subgraph of $G$ such that $V(C)\subseteq V(G')$, and $H'$ is an induced subgraph of $H$. The argument ${\bf x}$ is discussed below. We remark that we continue to use $n$ to refer to the size of the vertex-set of the input graph $G$ rather than the current graph $G'$.

\bigskip
{\noindent\bf Arguments.} While the execution of our algorithm progresses, we keep track of two arguments: the number of vertices in the current distance hereditary graph $H'$ that are assigned a non-zero value by $\x$, which we denote by $\alpha(G')$ and the fractional solution $\x$.

\begin{observation}
The measure $\alpha(G')$ can be computed in polynomial time.
\end{observation}

A fractional solution ${\bf x}$ is a function ${\bf x}: V(G')\rightarrow [0,\infty)$ such that for every chordless cycle $Q$ of $G'$ on at least $5$ vertices it holds that ${\bf x}(V(Q))\geq 1$. An optimal fractional solution minimizes the weight $w'({\bf x})=\sum_{v\in V(G')}w'(v)\cdot{\bf x}(v)$. Clearly, the solution to the instance $(G',w')$ of \wDH is at least as large as the weight of an optimal fractional solution.  Although we initially compute an optimal fractional solution ${\bf x}$ (at the initialization phase that is described below), during the execution of our algorithm, we manipulate this solution so it may no longer be optimal. Prior to any call to \alg{DHD-APPROX} with the exception of the first call, we ensure that ${\bf x}$ satisfies the following invariants:
\begin{itemize}
\item {\bf Low-Value Invariant}: For any $v\in V(G')$, it holds that ${\bf x}(v) < 1/\log n$.
\item {\bf Zero-Biclique Invariant}: For any $v\in V(C)$, it holds that ${\bf x}(v) = 0$.
\end{itemize}

We note that the Low-Value Invariant used here is simpler than the one used in Section~\ref{sec:approxCliqueChordal} since it is enough for the purpose of this section.  

\bigskip
{\noindent\bf Goal.} The depth of the recursion tree will be bounded by $\Delta=\OO(\log n)$, where the depth of initial call is $1$. The correctness of this claim is proved when we explain how to perform a recursive call. For each recursive call to \alg{DHD-APPROX}$(G',w',C,H',{\bf x})$, we aim to prove the following.

\begin{lemma}\label{lemma:DHapproxRecursiveCall}
For any $\delta\in\{1,2,\ldots,\Delta\}$, each recursive call to \alg{DHD-APPROX}  of depth $\delta \geq 2$ returns a solution that is at least $\opt$ and at most $(\frac{\log n}{\log n + 4})^\delta\cdot c\cdot\log n\cdot\log(\alpha(G'))\cdot w'({\bf x})$. Moreover, it returns a subset $U\subseteq V(G')$ that realizes the solution.
\end{lemma}

At the initialization phase, we see that in order to prove Lemma \ref{lemma:DHnewApproxSpecial}, it is sufficient to prove Lemma \ref{lemma:DHapproxRecursiveCall}.

\bigskip
{\noindent\bf Initialization.} Initially, the graphs $G'$ and $H'$ are simply set to be the input graphs $G$ and $H$, and the weight function $w'$ is simply set to be input weight function $w$. Moreover, we compute an optimal fractional solution ${\bf x}={\bf x}_{\mathrm{init}}$ by using the ellipsoid method. Recall that the following claim holds.

\begin{observation}
The solution of the instance $(G',w')$ of \wDH is lower bounded by $w'({\bf x}_{\mathrm{init}})$.
\end{observation} 

Moreover, it holds that $\alpha(G')\leq n$, and therefore to prove Lemma \ref{lemma:DHnewApproxSpecial}, it is sufficient to return a solution that is at least $\opt$ and at most $c\cdot\log n\cdot\log(\alpha(G))\cdot w({\bf x})$ (along with a subset that realizes the solution). Part of the necessity of the stronger claim given by Lemma \ref{lemma:DHapproxRecursiveCall} will become clear at the end of the initialization phase.

We would like to proceed by calling our algorithm recursively. For this purpose, we first need to ensure that ${\bf x}$ satisfies the low-value and zero-biclique invariants, to which end we use the following notation. We let $h({\bf x})=\{v\in V(G'): {\bf x}(v)\geq 1/\log n\}$ denote the set of vertices to which ${\bf x}$ assigns high values. Note that we can assume for each $v \in h({\bf x})$, we have ${\bf x}(v) \leq 1$. Moreover, given a biclique $M$ in $G'$, we let $({\bf x}\setminus M): V(G')\rightarrow [0,\infty)$ denote the function that assigns 0 to any vertex in $M$ and $(1+\frac{4}{\log n}){\bf x}(v)$ to any other vertex $v\in V(G')$. Now, to adjust ${\bf x}$ to be of the desired form both at this phase and at later recursive calls, we rely on the following lemmata.

\begin{lemma}\label{lem:DHadjustLowVal}
Define $\widehat{G}=G'- h({\bf x})$, $\widehat{w}=w'|_{V(\widehat{G})}$ and $\widehat{\bf x}={\bf x}|_{V(\widehat{G})}$. Then, $c'\cdot\log n\cdot\log(\alpha(\widehat{G}))\cdot w'(\widehat{\bf x}) +  w'(h({\bf x})) \leq c'\cdot\log n\cdot\log(\alpha(G))\cdot w'({\bf x})$, where $c' \geq 1$.
\end{lemma}

\begin{proof}
By the definition of $h({\bf x})$, it holds that $w'(\widehat{\bf x})\leq w'({\bf x})-\frac{1}{\log n}\cdot w'(h({\bf x}))$. Since $\widehat{G}$ is an induced subgraph of $G'$, it also holds that $\alpha(\widehat{G})\leq\alpha(G')$.
Thus, $c'\cdot\log n\cdot\log(\alpha(\widehat{G}))\cdot w'(\widehat{\bf x}) +  w'(h({\bf x})) \leq c'\cdot\log n\cdot\log(\alpha(G'))\cdot (w'({\bf x})-\frac{1}{\log n}\cdot w'(h({\bf x}))) +  w'(h({\bf x}))\leq c'\cdot\log n\cdot\log(\alpha(G))\cdot w'({\bf x})$.
\end{proof}

Thus, it is safe to update $G'$ to $G'- h({\bf x})$, $w'$ to $w'|_{V(\widehat{G})}$, $H'$ to $H'- h({\bf x})$ and ${\bf x}$ to ${\bf x}|_{V(\widehat{G})}$, where we ensure that once we obtain a solution to the new instance, we add $w'(h({\bf x}))$ to this solution and $h({\bf x})$ to the set realizing it.

\begin{lemma} \label{lemma:DHhole biclique}
Let $Q$ be a chordless cycle on at least $5$ vertices and $M$ be a biclique in $G'$ with vertex partitions as $V(M)=M_1 \uplus M_2$ such that $V(Q) \cap V(M) \neq \emptyset$. Then there is a chordless cycle $Q'$ on at least $5$ vertices that intersects $M$ in at most $3$ vertices such that $E(Q' \setminus M) \subseteq E(Q \setminus M)$. Furthermore, $Q'$ is of one of the following three types.
\begin{itemize}
    \item $Q' \cap M$ is a single vertex
    \item $Q' \cap M$ is an edge in $G[M]$
    \item $Q' \cap M$ is an induced path on $3$ vertices in $M$.
\end{itemize}
\end{lemma}

\begin{proof}
Observe that no chordless cycle on $5$ or more vertices may contain two vertices from each of $M_1$ and $M_2$, as that would imply a chord in it.
Now, if the chordless cycle $Q$ already satisfies the required conditions we output it as $Q'$.

First consider the case, when $Q \cap M$ contains exactly two vertices that don't have an edge between them.
Then the two vertices, say $v_1, v_2$, are both either in $M_1$ or in $M_2$. Suppose that they are both in $M_1$ and consider some vertex $u \in M_2$. 
Let $P_1$ is the longer of the two path segments of $Q$ between $v_1$ and $v_2$, and note that it must length at least $3$. 
Then observe that $G'[P_1 \cup \{u, v_1, v_2\}]$ contains a {\dho}, as $v_1, v_2$ have different distances depending on if $u$ is included in an induced subgraph or not. 
And further, it is easy to see that this {\dho} contains the induced path $v_1, u, v_2$. 
However, as all small obstructions have been removed from the graph, we have that $Q'$ is a chordless cycle in $G'$ on at least $5$ vertices. Furthermore, $Q' \cap M$ is the induced path $(v_1, u , v_2)$, in $G'$ and $E(Q' \setminus M) \subseteq E(Q \setminus M)$.

Now consider the case when $Q \cap M$ contains exactly three vertices. Observe that it cannot contain two vertices of $M_1$ and one vertex of $M_2$, or vice versa, as $Q$ doesn't satisfy the required conditions.
Therefore, $Q \cap M$ contains exactly three vertices from $M_1$ (or from $M_2$), which again don't form an induced path of length $3$. So there is an independent set of size $2$ in $Q \cap M$, and now, as before, we can again obtain the chordless cycle $Q'$ on at least $5$ vertices with $E(Q' \setminus M) \subseteq E(Q \setminus M)$. Before we consider the other cases, we have the following claim.

\begin{claim} \label{claim:biclique-property}
Let $M$ be a biclique in $G'$ with vertex partition as $V(M)= M_1 \uplus M_2$. Then $G'[M]$ has no induced $P_4$.
\end{claim}
\begin{proof}
Let $P$ be any induced path of length $4$ in $G'[M]$.
Then, either $V(P) \subseteq M_1$ or $V(P) \subseteq M_2$.
Now consider any such path $P$ in $M_1$ and some vertex $u \in M_2$. Then $G'[P \cup \{u\}]$ contains a {\dho} of size $5$ which is a contradiction to the fact that $G'$ has no small obstructions.
\end{proof}

Next, let $Q \cap M$ contain $4$ or more vertices. Note that in this case all these vertices are all either in $M_1$ or in $M_2$ since otherwise, $Q$ would not be a chordless cycle in $G'$ on at least $5$ vertices. Let us assume these vertices lie in $M_1$ (other case is symmetric). Let $v_1, v_2, v_3, \cdots, v_\ell \in M_1 \cap Q$ be the sequence of vertices obtained when we traverse $Q$ starting from an arbitrary vertex, where $\ell \geq 4$. By Claim~\ref{claim:biclique-property} they cannot form an induced path on $4$ vertices, i.e. $G'[V(Q) \cap M_1]$ consists of at least two connected components. Without loss of generality we may assume that $v_1$ and $v_\ell$ are in different components. Observe that the only possible edges between these vertices may be at most two of the edges $(v_1,v_2), (v_2,v_3)$, and $(v_3,v_\ell)$.
Hence, we conclude that either $v_1, v_3$ or $v_2, v_\ell$ are a distance of at least $3$ in $Q$. Let us assume that $v_2, v_\ell$ are at distance $3$ or more in $Q$, and the other case is symmetric and $P_{23}, P_{3\ell}$ be the paths not containing $v_1$ in $Q$ between $v_2$ and $v_3$, and $v_3$ and $v_\ell$, respectively.
%Observe that they are at a distance $2$ in $G[V(Q) \cup \{u\}]$, where $u \in M_2$ is any arbitrary vertex.
Notice that for any $u \in M_2$ the graph $G'[\{u\} \cup V(P_{23}) \cup V(P_{3\ell})]$ contains a {\dho}. Since the graph is free of all small obstruction, this {\dho}, denoted $\hat Q$, must be a chordless cycle on at least $5$ vertices. Furthermore this obstruction can contain at most $2$ vertices from $\{v_2, v_3, v_\ell\}$, as otherwise there would be a chord in it. Hence $\hat Q \cap M$ contains strictly fewer vertices than $Q \cap M$. Moreover, we have $E(\hat Q \setminus M) \subseteq E(Q \setminus M)$. Now, by a recursive application of this lemma to $\hat Q$, we obtain the required $Q'$.
\end{proof}

A consequence of the above lemma is that, whenever $M$ is a biclique in $G'$, we may safely ignore any \dho that intersects $M$ in more than $3$ vertices. This leads us to the following lemma.

\begin{lemma}\label{lem:DHadjustZeroClique}
Given a biclique $M$ in $G'$, the function $({\bf x}\setminus M)$ is a valid fractional solution such that $w'({\bf x}\setminus M)\leq (1+\frac{4}{\log n})w'({\bf x})$. 
\end{lemma}

\begin{proof}
To prove that $({\bf x}\setminus M)$ is a valid fractional solution, let $Q$ be some chordless cycle (not on $4$ vertices) in $G'$. We need to show that $({\bf x}\setminus M)(V(Q))\geq 1$. By our assumption $Q$ can contain at most $3$ vertices from $M$. Thus, since $\bf x$  is a valid fractional solution, it holds that ${\bf x}(V(Q)\setminus V(M))\geq 1 - \frac{3}{\log n}$. By the definition of $({\bf x}\setminus M)$, this fact implies that $({\bf x}\setminus M)(V(Q))=({\bf x}\setminus M)(V(Q)\setminus V(M)) \geq (1+ \frac{4}{\log n})(1-\frac{3}{\log n}) = 1 + \frac{1}{\log n} - \frac{12}{(\log n)^2}\geq 1$, where the last inequality relies on the assumption $n\geq 2^{12}$.
    
    For the proof of the second part of the claim, note that $w'({\bf x}\setminus M)=(1+\frac{4}{\log n})$ $w'({\bf x}|_{V(G')\setminus V(M)})\leq (1+\frac{4}{\log n})w'({\bf x})$.
\end{proof}

We  call \alg{DHD-APPROX} recursively with the fractional solution $({\bf x}\setminus C)$, and by Lemma \ref{lem:DHadjustZeroClique}, $w'({\bf x}\setminus C)\leq (1+\frac{4}{\log n})w'({\bf x})$.  If Lemma \ref{lemma:DHapproxRecursiveCall} were true, we return a solution that is at least $\opt$ and at most $(\frac{\log n}{\log n + 4}) \cdot c \cdot\log n\cdot\log(\alpha(G))\cdot w({\bf x}\setminus M) \leq c\cdot\log n\cdot\log(\alpha(G))\cdot w({\bf x})$ as desired. In other words, to prove Lemma \ref{lemma:DHnewApproxSpecial}, it is sufficient that we next focus only on the proof of Lemma \ref{lemma:DHapproxRecursiveCall}. The proof of this lemma is done by induction. When we consider some recursive call, we assume that the solutions returned by the additional recursive calls that it performs, which are associated with graphs $\widetilde{G}$ such that $\alpha(\widetilde{G})\leq\frac{3}{4}\alpha(G')$, complies with the conclusion of the lemma.

\bigskip
{\noindent\bf Termination.} Once $G'$ becomes a distance hereditary graph, we return 0 as our solution and $\emptyset$ as the set that realizes it. Clearly, we thus satisfy the demands of Lemma \ref{lemma:DHapproxRecursiveCall}. In fact, we thus also ensure that the execution of our algorithm terminates once $\alpha(G')<\log n$.

\begin{lemma}\label{lem:DHapproxTerminate}
If $\alpha(G') < \log n$, then $G'$ is a distance hereditary graph.
\end{lemma}

\begin{proof}
Suppose that $G'$ is not a distance hereditary graph. Then, it contains an obstruction $Q$. Since $\bf x$ is a valid fractional solution, it holds that ${\bf x}(V(Q))\geq 1$. But $\bf x$ satisfies the low-value invariant therefore, it holds that ${\bf x}(V(Q)) < |V(Q)|/\log n$. These two observations imply that $|V(Q)|>\log n$. Furthermore, at least $\log n$ of these vertices are assigned a non-zero value by ${\bf x}$, i.e. $\alpha(G') \geq \log n$. Therefore, if $\alpha(G') < \log n$, then $G'$ must be a distance hereditary graph.
\end{proof}

The fact that, the recursive calls are made onto graphs where the distance hereditary subgraph contains at most $3/4$ the number of vertices in the current distance hereditary subgraph, we observe the following.

\begin{observation}\label{obs:depth}
    The maximum depth of the recursion tree is bounded by $q\cdot\log n$ for some fixed constant $q$.
\end{observation}

\bigskip
{\noindent\bf Recursive Call.} 
Since $H'$ is a distance hereditary graph, it has a rank-width-one decomposition $(\cal T, \phi)$, where $\cal T$ is a binary tree and $\phi$ is a bijection from $V(G')$ to the leaves of $\cal T$.
Furthermore, rank-width of $\cal T$ is $1$, which means that for any edge of the tree, by deleting it, we obtain a partition of the leaves in $\cal T$. This partition induces a cut of the graph, where the set of edges crossing this cut forms a biclique $M$, with vertex partition as $V(M)=M_1 \uplus M_2$ in the graph.
By standard arguments on trees, we deduce that $\cal T$ has an edge that defines a partition such that after we remove the biclique edges between $M_1$ and $M_2$ from $G'$ we obtain two (not necessarily connected) graphs, $H_1$ and $H_2$, such that $|V(H_1)|,|V(H_2)|\leq\frac{3}{4}|V(H')|$
and $M_1 \subseteq H_1$, $M_2 \subseteq H_2$.
Note that the bicliques $M$ and $C$ are vertex disjoint.
We proceed by replacing the fractional solution $\bf x$ by $({\bf x}\setminus M)$.
For the sake of clarity, we denote ${\bf x}^*=({\bf x}\setminus M)$.
Let $G_1=G'[V(H_1) \cup V(C) \cup V(M)]$, $G_2=G'[V(H_2) \cup V(C) \cup V(M)]$.
%and observe that $\alpha(G_1),\alpha(G_2)\leq \frac{3}{4}\alpha(G')$ with respect to ${\bf x}^*$.

% Let $G_1=G'[V(\widehat{H_1})\cup V(M)\cup V(C)]$, $H_1=H'[V(\widehat{H_1})\cup V(M)]$, $G_2=G'[V(\widehat{H_2})\cup V(M)\cup V(C)]$ and $H_2=H'[V(\widehat{H_2})\cup V(M)]$, and observe that $\alpha(G_1),\alpha(G_2)\leq\frac{2}{3}\alpha(G')+2\leq\frac{3}{4}\alpha(G')$. 
% Here, the last inequality holds because $\alpha(G')\geq\lfloor(\log n)/2\rfloor$ and $n\geq 2^{48}$.

We adjust the current instance by relying on Lemma \ref{lem:DHadjustLowVal} so that $\bf x^*$ satisfies the low-value invariant (in the same manner as it is adjusted in the initialization phase). In particular, we remove $h({\bf x^*})$ from $G'$,$H'$, $G_1$, $H_1$, $G_2$ and $H_2$, and we let $(G^*,w^*,C,H^*,{\bf x}^*)$, $G_1^*$, $H_1^*$ $G_2^*$ and $H_2^*$ denote the resulting instance and graphs. Observe that, now we have $\alpha(G^*_1),\alpha(G^*_2)\leq\frac{3}{4}\alpha(G')$. We will return a solution that is at least $\opt$ and at most $(\frac{\log n}{\log n + 4})^{\delta+1}\cdot c\cdot\log n\cdot\log(\alpha(G'))\cdot w^*({\bf x}^*)$, along with a set that realizes it.\footnote{Here, the coefficient $(\frac{\log n}{\log n + 4})^\delta$ has been replaced by the smaller coefficient $(\frac{\log n}{\log n + 4})^{\delta+1}$.} In the analysis we will argue this it is enough for our purposes. %By Lemmata~\ref{lem:DHadjustLowVal} and \ref{lem:DHadjustZeroClique}, to prove Lemma \ref{lemma:DHapproxRecursiveCall},

%By Lemmata \ref{lem:DHadjustLowVal} and \ref{lem:DHadjustZeroClique}, to prove Lemma \ref{lemma:DHapproxRecursiveCall}, we now need to return a solution that is at least $\opt$ and at most $(\frac{\log n}{\log n + 4})^{\delta+1}\cdot c\cdot\log n\cdot\log(\alpha(G'))\cdot w^*({\bf x}^*)$, along with a set that realizes it.\footnote{Here, the coefficient $(\frac{\log n}{\log n + 4})^\delta$ has been replaced by the smaller coefficient $(\frac{\log n}{\log n + 4})^{\delta+1}$.}

Next, we define two subinstances, $I^*_1=(G^*_1,w^*|_{V(G^*_1)},C,H^*_1,{\bf x}^*|_{V(G^*_1)})$ and $I^*_2=(G^*_2,w^*|_{V(G^*_2)},$ $C,H^*_2,{\bf x}^*|_{V(G^*_2)})$. We solve each of these subinstances by a recursive call to \alg{DHD-APPROX}, and thus we obtain two solutions of sizes, $s^*_1$ to $I^*_1$ and $s^*_2$ to $I^*_2$, and two sets that realize these solutions, $S_1^*$ and $S_2^*$. By the inductive hypothesis, we have the following observations.

\begin{observation}\label{obs:inductionHit}
$S^*_1\cup S^*_2$ intersects any chordless cycle on at least $6$ vertices in $G^*$ that lies entirely in either $G^*_1$ or $G^*_2$.
\end{observation}

\begin{observation}\label{obs:inductionWeight}
Given $i\in\{1,2\}$, $s^*_i\leq (\frac{\log n}{\log n + 4})^{\delta+1}\cdot c\cdot\log n\cdot\log(\alpha(G^*_i))\cdot w({\bf x}^*_i)$.
\end{observation}

Moreover, since ${\bf x}^*(V(C)\cup V(M))=0$, we also have the following observation.

\begin{observation}\label{obs:inductionSum}
$w^*({\bf x}^*_1) + w^*({\bf x}^*_2) = w^*({\bf x}^*)$.
\end{observation}

% BU FIX START

We say that a cycle in $G^*$ is {\em bad} if it is a chordless cycle not on four vertices that belongs entirely to neither $G^*_1$ nor $G^*_2$. %On the one hand, observe that each bad cycle of $G^*$ contains a path on at least three vertices between a vertex in $V(C)$ and a vertex in $V(M)$ whose internal vertices belong only to $H^*_1\setminus V(M)$, and a path on at least three vertices between a vertex in $V(C)$ and a vertex in $V(M)$ whose internal vertices belong only to $H^*_2\setminus V(M)$. On the other hand, since $C$ and $M$ are cliques, if $G^*$ contains a path  on at least three vertices between a vertex in $V(C)$ and a vertex in $V(M)$ whose internal vertices belong only to $H^*_1\setminus V(M)$, and a path on at least three vertices between a vertex in $V(C)$ and a vertex in $V(M)$ whose internal vertices belong only to $H^*_2\setminus V(M)$, then $G^*$ contains a bad cycle. Thus, we can decide in polynomial time whether $G^*$ contains a bad cycle, and according to the result, we execute one of the procedures below.
Next, we show how to intersect bad cycles.

%\bigskip
%{\noindent\bf No Bad Cycles.}  First, suppose that $G^*$ does not contain any bad cycle. In this case, $S^*=S^*_1+S^*_2$ hits any chordless cycle in $G^*$, and it holds that $w^*(S^*) \leq s^*_1+s^*_2$. By Observations \ref{obs:inductionWeight} and \ref{obs:inductionSum}, $s^*_1+s^*_2\leq (\frac{\log n}{\log n + 4})^{\delta+1}\cdot c\cdot\log n\cdot\log(\alpha(G'))\cdot w^*({\bf x}^*)$, and thus we are done.

\bigskip
{\noindent\bf Bad Cycles.} 
Let us recall the current state of the graph $G'$.
$G'$ is partitioned into a biclique $C$ and a distance hereditary graph $H'$. Furthermore, there is a biclique $M$ with vertex bipartition as $M_1$ and $M_2$ so that deleting the edges between $M_1$ and $M_2$, gives a balanced partition of $H'$ into $H_1$ and $H_2$. 
Now, by Lemma~\ref{lemma:DHhole biclique}, we may ignore any chordless cycle that intersects either of the two bicliques $C$ and $M$ in more than three vertices each, and this allows us to update our fractional feasible solution to $\x^* = (\x / M)$.
Then we recursively solve the instances $G_1^*$ and $G_2^*$ and remove the returned solution. Now consider the remaining graph, and any obstructions that are left. As the graph no longer contains small obstructions, it is clear that any remaining obstruction is a chordless cycle on at least $6$ vertices and is a {\em bad cycle}. We first examine the relation between bad cycles and pairs $(v,u)$ of vertices $v\in V(C)$ and $u\in V(M)$. 

\begin{lemma}\label{lem:DHrelationBadCyc}
    If a bad cycle exists, then there must also be a bad cycle $Q$ such that $Q \setminus (M \cup C)$ is a union of two internally vertex disjoint and non-adjacent path-segments, $P_1$ and $P_2$ such that, $P_1 \subseteq G_1$ and $P_2 \subseteq G_2$, and each of them connect a pair of vertices in $M \times C$.
\end{lemma}
\begin{proof} Let $Q'$ be a bad cycle.
    Let us recall that the input graph $G'$ can be partitioned into the biclique $C$ and a distance hereditary graph $H'$. Hence $Q' \cap C \neq \emptyset$. Furthermore, if $Q' \cap M = \emptyset$, then $Q'$ is preserved in $G' - M$. This means that $Q'$ is either present in $G_1^*$, or in $G_2^*$, and hence it cannot be a bad cycle, which is a contradiction. Hence $Q' \cap M \neq \emptyset$ as well. 
    Finally, $Q'$ contains vertices from both $H_1$ and $H_2$, which implies
    $Q' \cap G_1$ and $Q' \cap G_2$ are both non-empty as well.
    
    Now, by applying Lemma~\ref{lemma:DHhole biclique} to $Q'$ and $C$, we obtain a bad cycle $\hat Q$ such that $\hat Q \cap C$ is either a single vertex, or an edge or an induced path of length three. Since, $M \cap C = \emptyset$, we can again apply Lemma~\ref{lemma:DHhole biclique} to $\hat Q$ and $M$, and obtain a bad cycle $Q$ such that each of $Q \cap C$ and $Q \cap M$ is either a single vertex, or, an edge or an induced path of length three.
    Hence, $Q - (V(M) \cup V(C))$ is a pair of internally disjoint paths, whose endpoints are in $M \times C$.
    Furthermore, one of these paths, denoted $P_1$, is contained in $G_1$, and the other, denoted $P_2$, is contained in $G_2$.
\end{proof}

\noindent
%As the above lemma shows, we may safely ignore all bad-cycles $Q$ for which $Q \setminus (M \cup C)$ is not a pair of internally vertex disjoint paths such that one of the paths is contained in $G_1$ and the other is contained in $G_2$.
The above lemma (Lemma~\ref{lem:DHrelationBadCyc}) implies that it is safe to ignore all the bad cycles that don't satisfy the conclusion of this lemma. 
We proceed to enumerate some helpful properties of those bad cycles that satisfy the above lemma.
We call $P_1,P_2$ the \emph{path segments} of the bad cycle $Q$. 
%
%Further note that, $Q \cap (M \cup C)$ is a subset of upto $6$ vertices that induce two connected component. It is clear that we can enumerate all possible configurations of $Q \cap (M \cup C)$ in polynomial time, and test if that choice leads to bad cycle. And if so, the following lemma and the subsequent discussion, shall demonstrate that we can encode the task of hitting all bad cycles as a weighted multicut problem.

\begin{lemma} 
Suppose $P_1,P_2$ are path segments of a bad cycle $Q$ where $P_1\subseteq G_1 - S^*_1$ and $P_2\subseteq G_2 - S^*_2$, where $S^*_1$ and $S^*_2$ are a solution to $G^*_1$ and $G^*_2$, respectively. Then for any $P_1'$ which is an induced path in $G_1 - S^*_1$ with the same endpoints as $P_1$ we have that $Q' = G'[(Q \cap (M \cup C)) \cup V(P_1') \cup V(P_2)]$ is also a bad cycle.
\end{lemma}
\begin{proof}
    Observe that, $P_1$ and $P_1'$ are paths between the same endpoints in $G_1 - S^*_1$, which is a distance hereditary graph. Therefore, $P_1'$ is an induced path of the same length as $P_1$. Furthermore, no vertex in $P_1'$ is adjacent to a vertex in $Q - P_1$. Hence $Q'$ is also a bad cycle.
\end{proof}
The above lemma allows us to reduce the problem of computing a solution that intersects all bad-cycles, to computing a solution for an instance of {\sc Weighted Multicut}. 
More formally, let $Q$ be a bad cycle with path segments $P_1$ and $P_2$ , the feasible fractional solution $\x^*$ assigns a total value of at least $1$ to the vertices in $Q$. As $\x^*$ assigns $0$ to every vertex in $M \cup C$, we have that at least one of $P_1$ or $P_2$ is assigned a total value of at least $1/2$. Suppose that it were $P_1$ then $2 \x^*$ assigns a total value $1$ to $P_1$ in $G_1$. This fractional solution is a solution to the {\sc Weighted Multicut} problem defined on the pairs of vertices in $C \times M$, which are separated by $2\x^*$ in $G'$ (whose description is given below).

%Suppose that it were $P_1$. Then above lemma implies that, if $\x^*(P_1) \geq 1/2$, then for any other path $P_1'$ in $G_1$ with the same endpoints as $P_1$, $\x^*(P_1') \geq 1/2$ as well. Indeed, for every path $P_1'$ in $G_1$ between the endpoints of $P_1$, we have $\x^*(P_1')  \geq 1/2$. Therefore, $2 \x^*$ assigns a total value $1$ to every path between the end-points of $P_1$ in $G_1$. This fractional solution is a solution to the {\sc Weighted Multicut} problem defined on the pairs of vertices in $C \times M$, which are separated by $2\x^*$ in $G'$.
%
%
%In light Lemma \ref{lem:DHrelationBadCyc}, to hit bad cycles, we now examine how the fractional solution ${\bf x}^*$ handles pairs $(v,u)$ of vertices $v\in V(C)$ and $u\in V(M)$.
%

Given $i\in\{1,2\}$, let $2{\bf x}^*_i$ denote the fractional solution that assigns to each vertex the value assigned by ${\bf x}^*_i$ times 2. %Moreover, let $\widehat{G}_1 = G_1\setminus (V(C)\cup V(M))$ and $\widehat{G}_2 = G_2\setminus (V(C)\cup V(M))$. Observe that $\widehat{G}_1$ and $\widehat{G}_2$ are distance hereditary graphs.
For a pair $(v,u)$ of vertices such that $v\in V(C)$ and $u\in V(M)$ we call $(v,u)$ an \emph{important pair} if there is a bad cycle $Q$ with path segments $P_1$ and $P_2$ that connects $v$ and $u$. Let $S^*_1$ and $S^*_2$ be a solution to $G^*_1$ and $G^*_2$, respectively (obtained recursively). For an important pair $(v,u)$ we let ${\cal P}_1(v,u)$ denote the set of any (simple) path $P_1$ between $v$ and $u$ whose internal vertices belong only to $G_1 - S^*_1$ and which does not contain any edge such that one of its endpoints belongs to $V(C)$ while the other endpoint belongs to $V(M)$. Symmetrically, we let ${\cal P}_2(v,u)$ denote the set of any path $P_2$ between $v$ and $u$ whose internal vertices belong only to $G_2 - S^*_2$ and which does not contain any edge such that one of its endpoints belongs to $V(C)$ while the other endpoint belongs to $V(M)$.

\begin{lemma}\label{lem:DHfracCutPairAlg}
For an important pair $(v,u)$ of vertices where $v\in V(C)$ and $u\in V(M)$, in polynomial time we can compute an index $i(v,u)\in\{1,2\}$ such that for any path $P\in{\cal P}_i(v,u)$, $2{\bf x}^*_i(V(P))\geq 1$.
\end{lemma}

\begin{proof}
Let $(v,u)$ be an important pair of vertices with $v\in V(C)$ and $u\in V(M)$. We start by arguing that such an index exists. Assuming a contradiction, suppose there exists $P_1 \in {\cal P}_1(v,u)$ and $P_2 \in {\cal P}_2(v,u)$ such that $2{\bf x}^*_1(V(P_1))< 1$ and $2{\bf x}^*_2(V(P_2))< 1$. Recall that we have a bad cycle bad cycle $Q$ in $G' - (S^*_1 \cup S^*_2)$ with paths segments as $P_1$ and $P_2$ which connects $v$ and $u$. But this implies that $2x^*(Q)<1$, contradicting that $x^*$ was a feasible solution to $G' - (S^*_1 \cup S^*_2)$. Therefore, such an index always exists.

For any index $j\in\{1,2\}$, we use Dijkstra's algorithm to compute the minimum weight of a path between $v$ and $u$ in the graph $\widehat{G}^*_i$ where the weights are given by $2{\bf x}^*_i$. In case the minimum weight is at least $1$, we have found the desired index $i(v,u)$. Moreover, we know that for at least one index $j\in\{1,2\}$, the minimum weight should be at least $1$ (if the minimum weight is at least 1 for both induces, we arbitrarily decide to fix $i(v,u)=1$).
\end{proof}

% Observe that each chordless cycle $Q$ in $G^*$ that does {\em not} entirely belong to either $G^*_1$ or $G^*_2$ contains a path $P_1$ on at least three vertices between a vertex $v_1$ in $V(C)$ and a vertex $u_1$ in $V(M)$ whose internal vertices belong only to $H^*_1\setminus V(M)$, and a path $P_2$ on at least three vertices between a vertex $v_2$ in $V(C)$ and a vertex $u_2$ in $V(M)$ whose internal vertices belong only to $H^*_2\setminus V(M)$. Observe that since $\bf x^*$ is a valid fractional solution and since ${\bf x}^*(V(C)\cup V(M))=0$, for each chordless cycle $Q$ in  $G^*$ that does {\em not} entirely belong to either $G^*_1$ or $G^*_2$, there exist $i_Q\in\{1,2\}$, $v_Q\in V(C)$, $u_Q\in V(M)$ and a subpath $P_Q$ of $Q$ between $v_Q$ and $u_Q$ whose internal vertices belong to $H^*_i\setminus V(M)$ and such that ${\bf x}^*_i(V(P))\geq 1/2$. Let ${\cal Q}_1$ denote the set of such chordless cycles $Q$ such that $i_Q=1$, and let ${\cal Q}_2$ denote the set of such chordless cycles $Q$ such that $i_Q=2$. Moreover, given $i\in\{1,2\}$, let $2{\bf x}^*_i$ denote the function that assigns to each vertex the value assigned by ${\bf x}^*_i$ times 2.
We say that an important pair $(u,v)$ is \emph{separated} in $G_i$, if the index assigned by Lemma~\ref{lem:DHfracCutPairAlg} assigns $i$ to ${\cal P}_i(u,v)$. Now, for every important pair $(v,u)$ such that $v\in V(C), u\in V(M)$ and $\{v,u\}\notin E(G')$, we perform the following operation. We check if this pair is separated in $G_1$, and if so, then we initialize ${\cal T}_1(v,u)=\emptyset$. Then for each pair of neighbors of $x$ of $v$ and $y$ of $u$, we add the pair $(x,y)$ to ${\cal T}_1(u,v)$. The set ${\cal T}_2(u,v)$ is similarly defined. 
At this point, we need to rely on approximate solutions to the {\sc Weighted Multicut} problem which is given by theorem below (Theorem~\ref{thm:multicut}).

\begin{theorem}[\cite{GVY96}]\label{thm:multicut}
Given an instance of {\sc Weighted Multicut}, one can find (in polynomial time) a solution that is at least $\opt$ and at most $d\cdot\log n\cdot \fopt$ for some fixed constant $d>0$, along with a set that realizes it.
\end{theorem}

%\defparproblemOpt{{\sc Weighted Multicut}}{An undirected graph $G$, a weight function $w: V(G)\rightarrow\mathbb{R}$ and a set $T=\{(s_1,t_1),\ldots,(s_k,t_k)\}$ of $k$ pairs of vertices of $G$.}{What is the minimum weight of a subset $S\subseteq V(G)$ such that for any pair $(s_i,t_i)\in{\cal T}$, $G\setminus S$ does not have any path between $s_i$ and $t_i$?}

Here, a fractional solution ${\bf y}$ is a function ${\bf y}: V(G)\rightarrow [0,\infty)$ such that for every pair $(s_i,t_i)\in{\cal T}$ and any path $P$ between $s_i$ and $t_i$, it holds that ${\bf y}(V(P))\geq 1$. An optimal fractional solution minimizes the weight $w({\bf y})=\sum_{v\in V(G)}w(v)\cdot{\bf y}(v)$. Let \fopt\ denote the weight of an optimal fractional solution. %We will employ the algorithm given by the following result by Garg et al.~\cite{GVY96}.

%\begin{theorem}[\cite{GVY96}]\label{thm:multicut}
%Given an instance of {\sc Weighted Multicut}, one can find (in polynomial time) a solution that is at least $\opt$ and at most $d\cdot\log n\cdot \fopt$ for some fixed constant $d>0$, along with a set that realizes it.
%\end{theorem}

By employing the algorithm given by Lemma \ref{lem:DHfracCutPairAlg}, we next construct two instances of {\sc Weighted Multicut}. The first instance is $J_1=(\widehat{G}^*_1,w^*_1,{\cal T}_1=\{{\cal T}_1(v,u): v\in V(C), u\in V(M), i(v,u)=1, \mbox{ and } (v,u) \mbox{ is an important pair}\})$ and the second instance is $J_2=(\widehat{G}^*_2,w^*_2,{\cal T}_2=\{{\cal T}_2(v,u): v\in V(C), u\in V(M), i(v,u)=2, \mbox{ and } (v,u) \mbox{ is an important pair}\})$. By Lemma \ref{lem:DHfracCutPairAlg}, $2{\bf x}^*_1$ and $2{\bf x}^*_2$ are valid solutions to $J_1$ and $J_2$, respectively. Thus, by calling the algorithm given by Theorem \ref{thm:multicut} with each instance, we obtain a solution $r_1$ to the first instance, along with a set $R_1$ that realizes it, such that $r_1\leq 2d\cdot \log|V(G^*_1)|\cdot w^*({\bf x}^*_1)$, and we also obtain a solution $r_2$ to the second instance, along with a set $R_2$ that realizes it, such that $r_2\leq 2d\cdot \log|V(G^*_2)|\cdot w^*({\bf x}^*_2)$. 
Now by Observation \ref{obs:inductionHit} and Lemma \ref{lem:DHrelationBadCyc}, we have obtained a set $S^*=S^*_1\cup S^*_2\cup R_1\cup R_2$ for which we have the following observation.

\begin{observation}
$S^*$ intersects any chordless cycle in $G^*$, and it holds that $w^*(S^*) \leq s^*_1+s^*_2+r_1+r_2$.
\end{observation}

%Thus, it is sufficient to show $s^*_1+s^*_2+r_1+r_2\leq (\frac{\log n}{\log n + 4})^{\delta+1}\cdot c\cdot\log n\cdot\log(\alpha(G'))\cdot w^*({\bf x}^*)$. Recall that for any $i\in\{1,2\}$, $r_i\leq 2d\cdot \log|V(G^*_i)|\cdot w^*({\bf x}^*_i)$. Thus, by Observation \ref{obs:inductionWeight} and since for any $i\in\{1,2\}$, $|V(G^*_i)|\leq n$ and $\alpha(G^*_i)\leq\frac{3}{4}\alpha(G')$, we have that
%
%\[w^*(S^*)\leq (\frac{\log n}{\log n + 4})^{\delta+1}\cdot c\cdot\log n\cdot\log(\frac{3}{4}\alpha(G'))\cdot (w^*({\bf x}^*_1)+w^*({\bf x}^*_2)) + 2d\cdot \log n\cdot (w^*({\bf x}^*_1)+w^*({\bf x}^*_2)).\]

We start by showing that $s^*_1+s^*_2+r_1+r_2 + w(h(x))\leq (\frac{\log n}{\log n + 4})^{\delta+1}\cdot c\cdot\log n\cdot\log(\alpha(G'))\cdot w^*({\bf x}^*)$. Recall that for any $i\in\{1,2\}$, $r_i\leq 2d\cdot \log|V(G^*_i)|\cdot w^*({\bf x}^*_i)$. Thus, by Observation \ref{obs:inductionWeight} and since for any $i\in\{1,2\}$, $|V(G^*_i)|\leq n$ and $\alpha(G^*_i)\leq\frac{3}{4}\alpha(G')$, we have that

\[w^*(S^*)\leq (\frac{\log n}{\log n + 4})^{\delta+1}\cdot c\cdot\log n\cdot\log(\frac{3}{4}\alpha(G'))\cdot (w^*({\bf x}^*_1)+w^*({\bf x}^*_2)) + 2d\cdot \log n\cdot (w^*({\bf x}^*_1)+w^*({\bf x}^*_2)).\]

By Observation \ref{obs:inductionSum}, we further deduce that 

\[w^*(S^*)\leq \left((\frac{\log n}{\log n + 4})^{\delta+1}\cdot c\cdot\log(\frac{3}{4}\alpha(G'))+2d\right)\cdot\log n\cdot w^*({\bf x}^*).\]

Now, it only remains to show that $(\frac{\log n}{\log n + 4})^{\delta+1}\cdot c\cdot\log(\frac{3}{4}\alpha(G'))+2d\leq (\frac{\log n}{\log n + 4})^{\delta+1}\cdot c\cdot\log\alpha(G')$, which is equivalent to $2d\leq (\frac{\log n}{\log n + 4})^{\delta+1}\cdot c\cdot\log(\frac{4}{3})$. Observe that $\delta\leq q\cdot\log n - 1$ for some fixed constant $q$. Indeed, it initially holds that $\alpha(G)\leq n$, at each recursive call, the number of vertices assigned a non-zero value by ${\bf x}^*$ decreases to at most a factor of $3/4$ of its previous value, and the execution terminates once this value drops below $(\log n)/2$. Thus, it is sufficient to choose the constant $c$ so that $2d\leq (\frac{\log n}{\log n + 4})^{q\cdot\log n}\cdot c\cdot\log(\frac{4}{3})$. As the term $(\frac{\log n}{\log n + 4})^{q\cdot\log n}$ is lower bounded by $1/(e^{4q})$, it is sufficient that we fix $c=2\cdot e^{4q}\cdot d\cdot 1/\log(\frac{4}{3})$.

Note that $2d\leq (\frac{\log n}{\log n + 4})^{q\cdot\log n}\cdot c\cdot\log(\frac{4}{3})$, where $d \geq 1$. Therefore, $(\frac{\log n}{\log n + 4})^{q\cdot\log n}\cdot c \geq 1$. This together with Lemma~\ref{lem:DHadjustLowVal} and~\ref{lem:DHadjustZeroClique} implies that $w'(S^*) + w'(h(x)) \leq (\frac{\log n}{\log n + 4})^{\delta} \cdot c\cdot \log(\alpha(G')) w'(x)$, which proves Lemma~\ref{lemma:DHapproxRecursiveCall}. 

%
%In this subsection we handle the special case where the input graph $G$ consists of a biclique $B$ and a distance hereditary graph $H$. More precisely, along with the input graph $G$ and the weight function $w$, we are also given a biclique $B$ a distance hereditary graph $H$ such that $V(G)=V(B) \cup V(H)$, where the vertex-sets $V(B)$ and $V(H)$ are disjoint. Here, we also assume that $G$ has no small \dho\ on at most $\obssizeDH$ vertices. Note that the edge-set $E(G)$ may contain edges between vertices in $B$ and vertices in $H$.  \hly{TODO: Write how we solve in poly time.}

%%%%%%%%%%%%%%
\subsection{General Graphs}\label{sec:approxGenGraphs}
In this section we handle general instances by developing a $\approxDH$-factor approximation algorithm for \wDH, \alg{Gen-DHD-APPROX}, thus proving the correctness of Theorem \ref{thm:newApprox2DH}. %We exactly follow the approach presented in~\cite{kimDHVD2016} for solving the general instance.

%The exact value of the constant $d \geq \max\{96,2c\}$ is determined later.\footnote{Recall that $c$ is the constant we fixed to ensure that the approximation ratio of \alg{APPROX} is bounded by $c\cdot\log n$.} This algorithm is based on recursion, and during its execution, we often encounter instances that are of the form of the Clique+Chordal special case of \wcdel, which will be dealt with using the algorithm \alg{APPROX} of Section \ref{sec:approxCliqueChordal}.

\bigskip
{\noindent\bf  The Recursive Algorithm.} We define each call to our algorithm \alg{Gen-DHD-APPROX} to be of the form $(G',w')$, where $(G',w')$ is an instance of \wDH\ such that $G'$ is an induced subgraph of $G$, and we denote $n'=|V(G')|$. We ensure that after the initialization phase, the graph $G'$ never contains a {\dho} on at most $\DHbddObs$ vertices. We call this invariant the {\em $\OO_{\DHbddObs}$-free invariant}. In particular, this guarantee ensures that the graph $G'$ always contains only a small number of maximal bicliques, as stated in the following lemma.

\begin{lemma}[Lemma 3.5 \cite{OsmallfreeDHVDkim}]\label{lem:numMaxCliquesDH}
    Let $G$ be a graph on $n$ vertices with no \dho\ on at most $6$ vertices. Then $G$ contains at most $(n^3+5n)/6$ maximal bicliques, and they can be enumerated in polynomial time.
\end{lemma}

\bigskip
{\noindent\bf Goal.} For each recursive call \alg{Gen-DHD-APPROX}$(G',w')$, we aim to prove the following.

\begin{lemma}\label{lemma:approxRecursiveCallGenDH}
\alg{Gen-DHD-APPROX} returns a solution that is at least $\opt$ and at most $\frac{d}{2}\cdot\log^3n'\cdot\opt$. Moreover, it returns a subset $U\subseteq V(G')$ that realizes the solution. Here $d$ is a constant, which will be determined later.
\end{lemma}

At each recursive call, the size of the graph $G'$ becomes smaller. Thus, when we prove that Lemma~\ref{lemma:approxRecursiveCallGenDH} is true for the current call, we assume that the approximation factor is bounded by $\frac{d}{2}\cdot\log^3\widehat{n}\cdot\opt$ for any call where the size $\widehat{n}$ of the vertex-set of its graph is strictly smaller than~$n'$.

\bigskip
{\noindent\bf Initialization.} We are given $(G,w)$ as input, and first we need to ensure that the $\OO_{\DHbddObs}$-free invariant is satisfied. 
%
%An easy way to ensure this condition is to greedily find and remove all vertices of any obstruction containing up to $\DHbddObs$ vertices of the graph.
%However, it will be more useful ... to solve an LP for the problem \hly{rephrase}, and obtain an optimal fractional solution $\x$. We then remove all vertices such that $\x(v) \geq \frac{1}{\DHbddObs}$, and let $S$ denote the collection of these vertices. Observe that, since $\x(V(G)) \leq \opt$ we have $\x(S) \leq 20\opt$, and $G' = G - S$ is rid of every small obstruction.   
%Furthermore, $\x$, now restricted to the graph $G'$ is a feasible LP solution in $G'$ that assigns a ``low'' LP value (less than $\LPvalDH$) to every vertex, and furthermore $\x(V(G')) \leq \opt$. This property will be rather useful, when dealing with the special case of a biclique + distance hereditary graph.\hly{Give it a formal name ?}
% 
For this purpose, we update $G$ as follows. 
First, we let $\OO_{\DHbddObs}$ denote the set of all \dho\ on at most $\DHbddObs$ vertices of $G$. Clearly, $\OO_{\DHbddObs}$ can be computed in polynomial time and it holds that $|\OO_{\DHbddObs}|\leq n^{\OO(1)}$. Now, we construct an instance of {\sc Weighted $\DHbddObs$-Hitting Set}, where the universe is $V(G)$, the family of all setsof size at most $\DHbddObs$ in $\OO_{\DHbddObs}$, and the weight function is $w'$. Since each \dho\ must be intersected, therefore, the optimal solution to our {\sc Weighted $\DHbddObs$-Hitting Set} instance is at most $\opt$. By using the standard $c'$-approximation algorithm for {\sc Weighted $c'$-Hitting Set}~\cite{KleinbergT05}, which is suitable for any fixed constant $c'$, we obtain a set $S\subseteq V(G)$ that intersects all the \dho\ in $\OO_{\DHbddObs}$ and whose weight is at most $\DHbddObs \cdot \opt$.  Having the set $S$, we remove its vertices from $G$ to obtain the graph $G'$,
and $w' = w|_{G'}$. Now that the $\OO_{\DHbddObs}$-free invariant is satisfied, we can call {\sf Gen-DHD-APPROX} on $(G',w')$ and to the outputted solution, we add $w(S)$ and $S$. 

We note that during the execution of the algorithm, we update $G'$ only by removing vertices from it, and thus it will always be safe to assume that the $\OO_{\DHbddObs}$-free invariant is satisfied. 
Now, by Lemma~\ref{lemma:approxRecursiveCallGenDH}, % applied to $(G',w')$, %where $w'$ is $w$ restricted to $V(G')$, 
we obtain a solution of weight at most  $\frac{d}{2}\cdot \log^3n\cdot \opt + 50\cdot\opt\leq d\cdot\log^3n\cdot\opt$, then combined with $S$, it allows us to conclude the correctness of Theorem \ref{thm:newApprox2DH}. %We further assume that the graph is connected after the initialization face. Such an assumption is valid because, otherwise we can solve the problem in each of the connected components. 
  
\bigskip
{\noindent\bf Termination.} 
Observe that due to Lemma \ref{lem:numMaxCliquesDH}, we can test in polynomial time, if our current graph $G'$ is of the special kind that can be partitioned into a biclique and a distance hereditary graph: we examine each maximal biclique of $G'$, and check whether after its removal we obtain a distance hereditary graph. Once $G'$ becomes such a graph that consists of a biclique and a distance hereditary graph, we solve the instance $(G',w')$ by calling algorithm \alg{DHD-APPROX}. 
Observe that this returns a solution of value $\OO(\log^2 n \cdot \opt)$ which is also $\OO(\log^3 n \cdot \opt)$.
%\hly{Fix the statement} Since we output an optimal sized solution in the base case, we thus ensure that at the base case of our induction, Lemma \ref{lemma:approxRecursiveCallGenDH} holds.

\bigskip
{\noindent\bf Recursive Call.} Similar to the case for \wcdel, instead computing a balanced separators with a maximal clique and some additional  vertices, here we find a balanced separator that comprises of a biclique and some additional, but small number of vertices. Existence of such a separator is guaranteed by Lemma~\ref{lem:bicliqueBalSep}. From Lemma~\ref{lem:numMaxCliquesDH}, it follows that the graph with no \dho\ of size at most $\DHbddObs$ contains at most $\OO(n^3)$ maximal bicliques and they can enumerated in polynomial time. We use the weighted variant of Lemma 3.8 from~\cite{OsmallfreeDHVDkim} in Lemma~\ref{lem:bicliqueBalSep}. The proof of Lemma~\ref{lem:bicliqueBalSep} remains exactly the same as that in Lemma 3.8 of~\cite{OsmallfreeDHVDkim}.

%A characterization of the distance hereditary graphs is given by the following theorem:
%\begin{lemma}\label{lem:splitdecomposition-star-clique}(Bouchet~\cite{}). A graph is  distance hereditary graph if and only if each bag in its canonical split decomposition is either a star or a clique.\end{lemma}

%\begin{lemma}\label{lem:splitdecomposition}(Cunningham and Edmonds~\cite{}). A connected graph upto isomorphism has a unique canonical split decomposition.\end{lemma}

%\begin{theorem}\label{polytime-split-decomposition}(Dahlhaus~\cite{}). Given a graph $G$, the canonical split decomposition of $G$ can be computed in time $\OO(|V (G)| + |E(G)|)$.
%\end{theorem}

\begin{lemma}[Lemma 3.8~\cite{OsmallfreeDHVDkim}] \label{lem:bicliqueBalSep}
Let $G'$ be a connected graph on $n'$ vertices not containing any \dho\ of size at most $\DHbddObs$ and $w : V(G) \rightarrow \mathbb{R}$ be a weight function. Then in polynomial time we can find a balanced vertex separator $K \uplus X$ such that the following conditions are satisfied.
  \begin{itemize}
  \item $K$ is a biclique in $G$ or an empty set;
  \item $w(X) \leq q \cdot \log n' \cdot \opt$, where $q$ is some fixed constant.
  \end{itemize}
Here, \opt\ is the weight of the optimum solution to \wDH\ of $G$.
\end{lemma}
We note that we used the $\OO(\log n')$-factor approximation algorithm by Leighton and Rao \cite{BalancedSeparator} in Lemma~\ref{lem:bicliqueBalSep} to find the balanced separator, instead of the $\OO(\sqrt{\log n'})$-factor approximation algorithm by Feige et al.~\cite{FeigeHL08}, as the algorithm by Feige et al. is randomized. 
Let us also remark that if $K$ is a biclique, then there is a bipartition of the vertices in $K$ into $A \uplus B$, where both $A$ and $B$ are non-empty,
which will be crucially required in later arguments.
%
%As an immediate corollary to Lemma~\ref{lem:separatorDH} and the fact that we can enumerate all the maximal bicliques in polynomial time, we have the observation:
%
%\begin{observation} \label{lem:dividebiCliqueDH}
%There is a polynomial time algorithm that finds a maximal biclique $M$ of $G'$ and a subset $S \subseteq V(G') \setminus M$ of weight at most $q \cdot \log n' \cdot \opt$  for some fixed constant $q$ such that $M \cup S$ is a balanced separator for $G'$.
%\end{observation}
%\begin{proof}
%We enumerate the set of all maximal bicliques $\mathcal{B}$, of $G$ in polynomial time using Lemma~\ref{lem:numMaxCliquesDH}. For each $B \in \mathcal{B}$, we do the following (we also do it for empty set). Let $G' = G - B$. If all the connected components in $G'$ are of size at most $2n/3$, then indeed $B$ is a balanced separator of $G$. Otherwise, we use the algorithm presented in~\cite{} for computing the 
%\end{proof}
%
%For the analysis of a recursive call, let $S^*$ denote a hypothetical set that realizes the optimal solution $\opt$ of the current instance $(G',w')$. Moreover, let $(D,M,\beta)$ be the canonical split decomposition of $G'\setminus S^*$ such that each bag is either a star or a clique. The claimed structure of the canonical split decomposition is guaranteed by Lemma~\ref{lem:splitdecomposition-star-clique}. We use Lemma 3.8 in~\cite{kimDHVD2016} 

Next, we apply in Lemma~\ref{lem:bicliqueBalSep} to $(G',w')$ to obtain a pair $(K,X)$. Since $K \cup X$ is a balanced separator for $G'$, we can partition the set of connected components of $G'\setminus (M\cup S)$ into two sets, ${\cal A}_1$ and ${\cal A}_2$, such that for $V_1=\bigcup_{A\in{\cal A}_1}V(A)$ and $V_2=\bigcup_{A\in{\cal A}_2}V(A)$ it holds that $n_1,n_2\leq\frac{2}{3}n'$ where $n_1=|V_1|$ and $n_2=|V_2|$. 
We then define two inputs of (the general case) \wDH: $I_1=(G'[V_1],w'_{V_1})$ and $I_2=(G'[V_2],w'_{V_2})$. Let $\opt_1$ and $\opt_2$ denote the optimal solutions to $I_1$ and $I_2$, respectively. Observe that since $V_1 \cap V_2=\emptyset$, it holds that $\opt_1+ \opt_2 \leq\opt$. We solve each of the two sub-instances by recursively calling algorithm \alg{Gen-DH-APPROX}. 
By the inductive hypothesis, we obtain two sets, $S_1$ and $S_2$, such that $G'[V_1] \setminus S_1$ and $G'[V_2]\setminus S_2$ are both distance hereditary graphs, and $w'(S_1) \leq \frac{d}{2}\cdot\log^3 n_1\cdot\opt_1$ and $w'(S_2)\leq \frac{d}{2}\cdot\log^3 n_2\cdot\opt_2$. Now, if $K$ were an empty set then it is easy to see that $X \cup S_1 \cup S_2$ is a feasible solution to the instance $(G', w')$. Now let us bound the total weight of this subset.
\[\begin{array}{ll}
w'(X \cup S_1 \cup S_2) & \leq w'(X) + w'(S_1) + w'(S_2) \\
& \leq q \cdot \log n' \cdot \opt + \frac{d}{2}\cdot(\log^3 n_1 \cdot \opt_1 + \log^3 n_2 \cdot \opt_2) \\
%\end{array}\]
\\
& \text{Recall that $n_1,n_2\leq\frac{2}{3}n' $ and $\opt_1+\opt_2 \leq \opt$.} \\ 
\\
%\[\begin{array}{ll}
& < q \cdot \log n' \cdot \opt + \frac{d}{2}\cdot \log^3 \frac{2}{3}n' \cdot \opt \\
%& < \frac{d}{2} \cdot (\log n')^3 \cdot \opt + \log n' \cdot \opt \cdot (q + \frac{d}{2}\cdot (\log \frac{3}{2})^2- \frac{d}{2}\cdot 2 \cdot \log \frac{3}{2}).
& < \frac{d}{2} \cdot \log^3 n' \cdot \opt 
\end{array}\]

%Now, to ensure that $w'(X \cup S_1 \cup S_2)\leq\frac{d}{2}\cdot\log^3 n'\cdot\opt$, it is sufficient to ensure that $q+1+ \frac{d}{2}\cdot (\log \frac{3}{2})^2- \frac{d}{2}\cdot 2 \cdot \log \frac{3}{2}\leq 0$, which can be done by fixing $\displaystyle{d = \big(2 \log\frac{3}{2} - \log^2 \frac{3}{2} \big)^{-1} \cdot 2(q+1)}$.

The more interesting case is when $K$ is a biclique. 
Then, we first remove $X \cup S_1 \cup S_2$ from the graph, and note that the above bound also holds for this subset of vertices.
Now observe that the graph $G'' = G' - (X \cup S_1 \cup S_2)$ can be partitioned into a biclique $K$ and a distance hereditary graph $H = G[(V_1 \cup V_2) \setminus (S_1 \cup S_2)]$, along with the weight function $w'' = w'_{V(G'')}$.
Thus we have an instance  of the Biclique + Distance Hereditary Graph spacial case of WDHVD.
Furthermore, note that we retained a fractional feasible solution $\x$ to the LP of the initial input $G', w'$, which upperbounds the value of a fractional feasible solution $\x''$ to the LP of the instance $G'',w''$.
We apply the algorithm {\sf DHD-APPROX} on $(G'',w'', K, H, \x'')$ which outputs a solution $S$ such that $w''(S) = w'(S) = \OO(\log^2 n \cdot \opt)$.

%We now define an input of the Biclique+ Distance Hereditary Graph, special case of \wDH: $J=(G'[(V_1\cup V_2\cup M)\setminus (S_1\cup S_2)],w'|_{(V_1\cup V_2\cup M)\setminus (S_1\cup S_2)})$. Observe that since $G'[V_1]\setminus S_1$ and $G'[V_2]\setminus S_2$ are distance hereditary graphs and $M$ is a biclique, this is an instance of Biclique+Distance Hereditary Graph special case of \wDH.  We solve this instance by calling the algorithm \alg{DHD-APPROX}. We thus obtain a set, $\widehat{S}$, such that $G'[(V_1\cup V_2\cup M)\setminus (S_1\cup S_2\cup\widehat{S})]$ is a distance hereditary graph, and $w'(\widehat{S})\leq \opt$ (since $|(V_1\cup V_2\cup M)\setminus (S_1\cup S_2)|\leq n'$ and the optimal solution of each of the subinstances is at most $\opt$).

Observe that, any obstruction in $G' \setminus S$ is either completely contained in $G'[V_1 \setminus S]$, or completely contained in $G'[V_2 \setminus S]$ or it contains at least one vertex from $K$.
This observation, along with the fact that $G'[(V_1\cup V_2\cup K) \setminus (S_1\cup S_2\cup\widehat{S})]$ is a distance hereditary graph, implies that $G' \setminus T$ is a distance hereditary graph where $T= X \cup S_1\cup S_2\cup\widehat{S}$. Thus, it is now sufficient to show that $w'(T)\leq\frac{d}{2}\cdot (\log n')^3 \cdot \opt$.
By the discussion above, we have that {\sf DHD-APPROX} returns a solution of value $c \log^2 n \cdot \opt$, where $c$ is some constant.
%Therefore we have, 

\[\begin{array}{ll}
w'(T) & \leq w'(S) + w'(S_1) + w'(S_2) + w'(\widehat{S}_1) + w'(\widehat{S}_2)\\
& \leq q\cdot\log n'\cdot\opt + \frac{d}{2}\cdot(\log^3n_1\cdot\opt_1 + \log^3n_2\cdot\opt_2) + c\cdot\log^2n'\cdot\opt.
\end{array}\]

Recall that $n_1,n_2\leq\frac{2}{3}n'$ and $\opt_1+\opt_2\leq\opt$. Thus, we have that

\[\begin{array}{ll}
w'(T) & \leq q\cdot\log n'\cdot\opt + \frac{d}{2}\cdot(\log^3\frac{2}{3}n')\cdot\opt + c\cdot\log^2n'\cdot\opt\\
%& \leq \frac{d}{2}\cdot\log^3n'\cdot\opt + (c - \frac{3d}{2}\log\frac{3}{2})\cdot\log^2 n'\cdot\opt.
& \leq \frac{d}{2}\cdot\log^3n'\cdot\opt + (c - d\log\frac{3}{2})\cdot\log^2 n'\cdot\opt.
\end{array}\]

Overall, we conclude that to ensure that $w'(T)\leq\frac{d}{2}\cdot\log^3n'\cdot\opt$, it is sufficient to ensure that $c - d\log\frac{3}{2}\leq 0$, which can be done by fixing $\displaystyle{d=\frac{c}{\log\frac{3}{2}}}$.

\section{Conclusion}
In this paper, we designed \afg-approximation algorithms for \wpfd, \wcvd and \wDHfull (or \wroned). These  
algorithms are the first ones for these problems whose approximation factors are bounded by \afg. Along the way, we also obtained a constant-factor approximation algorithm for {\sc Weighted Multicut} on chordal graphs. 
All our algorithms are based on the same recursive scheme. We believe that the scope of applicability of our approach is very wide. We would like to conclude our paper with the following concrete open problems. 
\begin{itemize}
\setlength{\itemsep}{-2pt}
\item Does \wpfd admit a constant-factor approximation algorithm? Furthermore, studying families $\mathscr{F}$ that do not necessarily contain a planar graph is another direction for further research.
\item Does \wcvd admit a constant-factor approximation algorithm?
\item Does \wretad admit a \afg-factor approximation algorithm?
\item On which other graph classes {\sc Weighted Multicut} admits a constant-factor approximation?
\end{itemize}

% BIBLIOGRAPHY
\bibliographystyle{siam}
\bibliography{references}

\begin{thebibliography}{10}

\bibitem{soda17chordal}
{\sc A.~Agrawal, D.~Lokshtanov, P.~Misra, S.~Saurabh, and M.~Zehavi}, {\em
  Feedback vertex set inspired kernel for chordal vertex deletion}, in
  Proceedings of the 28th ACM-SIAM Symposium on Discrete Algorithms (SODA),
  2017, pp.~1383--1398.

\bibitem{BafnaBF99}
{\sc V.~Bafna, P.~Berman, and T.~Fujito}, {\em A {$2$}-approximation algorithm
  for the undirected feedback vertex set problem}, SIAM Journal on Discrete
  Mathematics, 12 (1999), pp.~289--297.

\bibitem{BansalRU1}
{\sc N.~Bansal, D.~Reichman, and S.~W. Umboh}, {\em {LP}-based robust
  algorithms for noisy minor-free and bounded treewidth graphs}, in Proceedings
  of the 28th ACM-SIAM Symposium on Discrete Algorithms (SODA), 2017,
  pp.~1964--1979.

\bibitem{BansalRU2}
{\sc N.~Bansal and S.~W. Umboh}, {\em Personal communication.}, 2017.

\bibitem{Bar-YehudaE81}
{\sc R.~Bar-Yehuda and S.~Even}, {\em A linear-time approximation algorithm for
  the weighted vertex cover problem}, Journal of Algorithms, 2 (1981),
  pp.~198--203.

\bibitem{BarYGJ98}
{\sc R.~Bar-Yehuda, D.~Geiger, J.~Naor, and R.~M. Roth}, {\em Approximation
  algorithms for the feedback vertex set problem with applications to
  constraint satisfaction and {B}ayesian inference}, SIAM Journal on Computing,
  27 (1998), pp.~942--959.

\bibitem{Borie1992}
{\sc R.~B. Borie, R.~G. Parker, and C.~A. Tovey}, {\em Automatic generation of
  linear-time algorithms from predicate calculus descriptions of problems on
  recursively constructed graph families}, Algorithmica, 7 (1992),
  pp.~555--581.

\bibitem{CourcelleMR00}
{\sc B.~Courcelle, J.~A. Makowsky, and U.~Rotics}, {\em Linear time solvable
  optimization problems on graphs of bounded clique-width}, Theory of Computing
  Systems, 33 (2000), pp.~125--150.

\bibitem{CourcelleO00}
{\sc B.~Courcelle and S.~Olariu}, {\em Upper bounds to the clique width of
  graphs}, Discrete Applied Mathematics, 101 (2000), pp.~77--114.

\bibitem{diestel-book}
{\sc R.~Diestel}, {\em Graph Theory, 4th Edition}, vol.~173 of Graduate texts
  in mathematics, Springer, 2012.

\bibitem{C4FreeNumCliques}
{\sc M.~Farber}, {\em On diameters and radii of bridged graphs}, Discrete
  Mathematics, 73 (1989), pp.~249--260.

\bibitem{FeigeHL08}
{\sc U.~Feige, M.~Hajiaghayi, and J.~R. Lee}, {\em Improved approximation
  algorithms for minimum weight vertex separators}, {SIAM} Journal on
  Computing, 38 (2008), pp.~629--657.

\bibitem{Fiorini:2009ipco}
{\sc S.~Fiorini, G.~Joret, and U.~Pietropaoli}, {\em Hitting diamonds and
  growing cacti}, in Proceedings of the 14th Conference on Integer Programming
  and Combinatorial Optimization (IPCO), vol.~6080, 2010, pp.~191--204.

\bibitem{FominLMPS16}
{\sc F.~V. Fomin, D.~Lokshtanov, N.~Misra, G.~Philip, and S.~Saurabh}, {\em
  Hitting forbidden minors: Approximation and kernelization}, {SIAM} Journal on
  Discrete Mathematics, 30 (2016), pp.~383--410.

\bibitem{FominLMS12}
{\sc F.~V. Fomin, D.~Lokshtanov, N.~Misra, and S.~Saurabh}, {\em Planar
  f-deletion: Approximation, kernelization and optimal fpt algorithms}, in
  Proceedings of IEEE 53rd Annual Symposium on Foundations of Computer Science
  (FOCS), 2012, pp.~470--479.

\bibitem{FominLRS11}
{\sc F.~V. Fomin, D.~Lokshtanov, V.~Raman, and S.~Saurabh}, {\em
  Bidimensionality and {EPTAS}}, in Proceedings of the 22nd ACM-SIAM Symposium
  on Discrete Algorithms (SODA), 2011, pp.~748--759.

\bibitem{FominLS12}
{\sc F.~V. Fomin, D.~Lokshtanov, and S.~Saurabh}, {\em Bidimensionality and
  geometric graphs}, in Proceedings of the 23rd ACM-SIAM Symposium on Discrete
  Algorithms (SODA), 2012, pp.~1563--1575.

\bibitem{FominLST10}
{\sc F.~V. Fomin, D.~Lokshtanov, S.~Saurabh, and D.~M. Thilikos}, {\em
  Bidimensionality and kernels}, in Proceedings of the 21st ACM-SIAM Symposium
  on Discrete Algorithms (SODA), 2010, pp.~503--510.

\bibitem{GVY96}
{\sc N.~Garg, V.~V. Vazirani, and M.~Yannakakis}, {\em Approximate max-flow
  min-(multi)cut theorems and their applications}, SIAM Journal on Computing,
  25 (1996), pp.~235--251.

\bibitem{golovin06}
{\sc D.~Golovin, V.~Nagarajan, and M.~Singh}, {\em Approximating the
  $k$-multicut problem}, in ACM-SIAM Symposium on Discrete Algorithms (SODA),
  2006, pp.~621--630.

\bibitem{Golumbic80}
{\sc M.~C. Golumbic}, {\em Algorithmic Graph Theory and Perfect Graphs},
  Academic Press, New York, 1980.

\bibitem{hammer1990completely}
{\sc P.~L. Hammer and F.~Maffray}, {\em Completely separable graphs}, Discrete
  applied mathematics, 27 (1990), pp.~85--99.

\bibitem{HlinenyOSG08}
{\sc P.~Hlinen{\'{y}}, S.~Oum, D.~Seese, and G.~Gottlob}, {\em Width parameters
  beyond tree-width and their applications}, The Computer Journal, 51 (2008),
  pp.~326--362.

\bibitem{howorka1977characterization}
{\sc E.~Howorka}, {\em A characterization of distance-hereditary graphs}, The
  quarterly journal of mathematics, 28 (1977), pp.~417--420.

\bibitem{JansenPili2016}
{\sc B.~M.~P. Jansen and M.~Pilipczuk}, {\em Approximation and kernelization
  for chordal vertex deletion}, in Proceedings of the 28th ACM-SIAM Symposium
  on Discrete Algorithms (SODA), 2017, pp.~1399--1418.

\bibitem{OsmallfreeDHVDkim}
{\sc E.~J. {Kim} and O.~{Kwon}}, {\em {A Polynomial Kernel for
  Distance-Hereditary Vertex Deletion}}, ArXiv e-prints,  (2016).

\bibitem{KleinbergT05}
{\sc J.~Kleinberg and E.~Tardos}, {\em Algorithm design}, Addison-Wesley, 2005.

\bibitem{BalancedSeparator}
{\sc T.~Leighton and S.~Rao}, {\em Multicommodity max-flow min-cut theorems and
  their use in designing approximation algorithms}, Journal of the ACM, 46
  (1999), pp.~787–--832.

\bibitem{LewisY80}
{\sc J.~M. Lewis and M.~Yannakakis}, {\em The node-deletion problem for
  hereditary properties is {NP}-complete}, Journal of Computer and System
  Sciences, 20 (1980), pp.~219--230.

\bibitem{LundY93}
{\sc C.~Lund and M.~Yannakakis}, {\em The approximation of maximum subgraph
  problems}, in Proceedings of the 20nd International Colloquium on Automata,
  Languages and Programming (ICALP), vol.~700, 1993, pp.~40--51.

\bibitem{moon1965cliques}
{\sc J.~W. Moon and L.~Moser}, {\em On cliques in graphs}, Israel journal of
  Mathematics, 3 (1965), pp.~23--28.

\bibitem{NemT74}
{\sc G.~L. Nemhauser and L.~E. Trotter, Jr.}, {\em Properties of vertex packing
  and independence system polyhedra}, Mathematical Programming, 6 (1974),
  pp.~48--61.

\bibitem{Oum05}
{\sc S.~Oum}, {\em Rank-width and vertex-minors}, Journal of Combinatorial
  Theory, Series B, 95 (2005), pp.~79--100.

\bibitem{Oum08}
\leavevmode\vrule height 2pt depth -1.6pt width 23pt, {\em Approximating
  rank-width and clique-width quickly}, ACM Transactions on Algorithms, 5
  (2008).

\bibitem{Oum16}
\leavevmode\vrule height 2pt depth -1.6pt width 23pt, {\em Rank-width:
  Algorithmic and structural results}, CoRR, abs/1601.03800 (2016).

\bibitem{OumS06}
{\sc S.~Oum and P.~D. Seymour}, {\em Approximating clique-width and
  branch-width}, Journal of Combinatorial Theory, Series B, 96 (2006),
  pp.~514--528.

\bibitem{Robertson:graph-minor}
{\sc N.~Robertson and P.~D. Seymour}, {\em Graph minors. v. excluding a planar
  graph}, Journal of Combinatorial Theory Series B, 41 (1986), pp.~92--114.

\bibitem{robertson-minor-testing}
{\sc N.~Robertson and P.~D. Seymour}, {\em Graph minors .xiii. the disjoint
  paths problem}, Journal of Combinatorial Theory, Series B, 63 (1995),
  pp.~65--110.

\bibitem{RobertsonS04}
\leavevmode\vrule height 2pt depth -1.6pt width 23pt, {\em Graph minors. {XX.}
  wagner's conjecture}, Journal of Combinatorial Theory, Series {B}, 92 (2004),
  pp.~325--357.

\bibitem{sachs1970berge}
{\sc H.~Sachs}, {\em On the berge conjecture concerning perfect graphs},
  Combinatorial Structures and their Applications, 37 (1970), p.~384.

\bibitem{TsukiyamaIAS77}
{\sc S.~Tsukiyama, M.~Ide, H.~Ariyoshi, and I.~Shirakawa}, {\em A new algorithm
  for generating all the maximal independent sets}, {SIAM} Journal on
  Computing, 6 (1977), pp.~505--517.

\bibitem{Yannakakis79}
{\sc M.~Yannakakis}, {\em The effect of a connectivity requirement on the
  complexity of maximum subgraph problems}, Journal of the ACM, 26 (1979),
  pp.~618--630.

\bibitem{Yannakakis94}
\leavevmode\vrule height 2pt depth -1.6pt width 23pt, {\em Some open problems
  in approximation}, in Proceedings of 2nd Italian Conference on Algorithms and
  Complexity, Second {(CIAC)}, 1994, pp.~33--39.

\end{thebibliography}

\appendix

\end{document}